\documentclass[a4paper,onecolumn,11pt]{quantumarticle}

\pdfoutput=1
\usepackage[utf8]{inputenc}
\usepackage[T1]{fontenc}

\usepackage{mathptmx,amsmath,mathtools,commath,amsfonts,amssymb,amsthm}

\usepackage{etoolbox}
\usepackage{bm}
\usepackage{xcolor}
\usepackage{graphicx}
\usepackage{subcaption}
\usepackage{dcolumn}
\usepackage{makecell}
\usepackage{url}
\usepackage{pifont}
\usepackage[font=small,labelfont=bf]{caption}
\usepackage[normalem]{ulem}

\newcommand{\greencircle}{{\scriptsize \color{mygreen}\ding{108}}}

\newcommand{\Dt}{\Delta t}
\newcommand{\Dy}{\Delta y}
\newcommand{\kernel}{\kappa}
\newcommand{\weight}{\mu}

\usepackage{pdfpages}
\usepackage{tikz}
\usepackage{braket}
\usepackage{hyperref}
\usepackage[capitalise]{cleveref}
\usepackage[numbers,sort&compress]{natbib}

\newcommand{\bc}{\boldsymbol{c}}

\newcommand{\bz}{\boldsymbol{z}}
\newcommand{\bt}{\boldsymbol{t}}
\newcommand{\bs}{\boldsymbol{s}}

\newlength\figureheight
\newlength\figurewidth
\usepackage{pgfplots}

 \newtheorem{thm}{Theorem}
 \newtheorem{lem}{Lemma}
 \newtheorem{cor}{Corollary}
 \newtheorem{prop}{Property}
 \newtheorem{claim}{Claim}

\newcommand{\externalizetikz}{0}   

\definecolor{myblue}{rgb}{0,0.4470,0.7410}
\definecolor{myred}{rgb}{0.8500,0.3250,0.0980}
\definecolor{myorange}{rgb}{0.9290,0.6940,0.1250}
\definecolor{mypurple}{rgb}{0.4940,0.1840,0.5560}
\definecolor{mygreen}{rgb}{0.4660,0.6740,0.1880}
\definecolor{mylightblue}{rgb}{0.3010,0.7450,0.9330}
\definecolor{mydarkred}{rgb}{0.6350,0.0780,0.1840}

\ifnum\externalizetikz=2
  \usepackage{tikzexternal}
  \tikzsetexternalprefix{Tikzfig/}
\else
  \usepackage{tikz}
  \usepackage{pgfplots}
  \usetikzlibrary{calc}
  \usetikzlibrary{intersections}
  \usetikzlibrary{patterns.meta}
  \usetikzlibrary{positioning}
  \usetikzlibrary{plotmarks}
  \tikzset{>=latex}
  \usetikzlibrary{shapes,decorations,matrix}
  \usetikzlibrary{backgrounds,fit,decorations.pathreplacing,calc}
  \pgfplotsset{
    compat=newest,
    table/header=false,
    tick label style={font=\footnotesize},
    label style={font=\small},
    legend style={font=\footnotesize},
    legend cell align=left,
    colormap={parula}{
      rgb255=(53,42,135)
      rgb255=(15,92,221)
      rgb255=(18,125,216)
      rgb255=(7,156,207)
      rgb255=(21,177,180)
      rgb255=(89,189,140)
      rgb255=(165,190,107)
      rgb255=(225,185,82)
      rgb255=(252,206,46)
      rgb255=(249,251,14)
    }
  }
  \pgfplotsset{
    myColOne/.style={myblue},
    myColTwo/.style={myred},
    myColThr/.style={myorange},
    myColFou/.style={mypurple},
    myColFiv/.style={mygreen},
    myColSix/.style={mylightblue},
    myColSev/.style={mydarkred}
  }
  \ifnum\externalizetikz=1
    \usepgfplotslibrary{external}
  \fi
\fi

\newcommand{\datfile}[1]{dat/#1.dat}
\newcommand{\plotconv}[6][]{%
\begin{tikzpicture}
\begin{semilogyaxis}[
  width=\columnwidth,%
  xmin=#3,xmax=#4,%
  ymin=#5,ymax=#6,%
  xlabel={\# Observables (Real)},%
  ylabel={Absolute Error},%
  legend style={draw=none,row sep=-2pt},%
  legend pos=north east,%
  #1,%
]
#2
\end{semilogyaxis}
\end{tikzpicture}%
}

\title{Quantum Rational Transformation Using Linear Combinations of Hamiltonian Simulations}
\author{Yizhi Shen}
\email{yizhis@lbl.gov}
\thanks{Equal contributions}
\affiliation{Applied Mathematics and Computational Research Division, Lawrence Berkeley National Laboratory, Berkeley, CA 94720, USA}

\author{Niel Van Buggenhout}
\email{nvanbugg@math.uc3m.es}
\thanks{Equal contributions}
\affiliation{Departamento de Matem\'aticas, Universidad Carlos III de Madrid, Avenida de la Universidad 30, 28911 Legan\'es,
	Spain}

\author{Daan Camps}
\affiliation{National Energy Research Scientific Computing Center, Lawrence Berkeley National Laboratory, Berkeley, CA 94720, USA}

\author{Katherine Klymko}
\affiliation{National Energy Research Scientific Computing Center, Lawrence Berkeley National Laboratory, Berkeley, CA 94720, USA}

\author{Roel Van Beeumen}
\affiliation{Applied Mathematics and Computational Research Division, Lawrence Berkeley National Laboratory, Berkeley, CA 94720, USA}

\begin{document}

\maketitle

\begin{abstract}
    Rational functions are powerful tools in scientific computing, offering rapidly convergent solutions where polynomial functions struggle. Nonetheless, their abilities to advance quantum algorithms remain largely untapped. In this paper, we introduce effective implementations of rational transformations of a target operator on quantum hardware. By leveraging tailored integral representations of the operator resolvent, we show that quantum rational transformations can be implemented efficiently using linear combination of Hamiltonian simulations (LCHS), with favorable scaling in both the maximal and total evolution time. We formulate two complementary LCHS approaches -- discrete-time and continuous-time  each providing unique strategies to approximate the integral representations of a resolvent. We consider quantum rational transformation for the ubiquitous task of approximating functions of a Hermitian operator, with particular emphasis on the elementary signum function which we use for spectral analysis 
    including ground and excited state problems. 
    Our numerical demonstration on spin systems indicates that our approach is efficient and achieves accurate estimation of the low-lying energies. 
\end{abstract}

\section{Introduction}
\label{sec:introduction}

Quantum algorithms for approximating functions of Hermitian operators have gained both theoretical and practical interest across quantum information science. They have been explored for tasks such as computing eigenenergies, simulating quantum systems, and implementing quantum walks~\cite{klymko2022,kirby2023exact,shen2023estimating,Low2017QSP,Low2019hamiltonian,berry2024doubling,Berry2016quantumwalk}. In many-body physics, localized functions, e.g., the step function, are instrumental in extracting essential spectral information, such as specific ground and excited state properties~\cite{motta2020determining,KeDuWa21}. In quantum linear algebra, elementary functions including the inverse, absolute value, and square root are intimately tied to common operator factorizations that find numerous applications. Most notably, the signum function can be employed to obtain the symmetric eigendecomposition and singular value decomposition of a matrix \cite{NaFr16}. However, a matrix function can often not be computed exactly, in such cases it is advantageous to employ a subset of simpler but expressive transformations to approximate a broader set of target functions.

Function approximation via polynomial transformations in quantum algorithms are enabled by approaches such as quantum signal processing (QSP), the quantum singular value transformation (QSVT), and the quantum eigenvalue transformation (QET)~\cite{Tr19,Low2017QSP,Low2019hamiltonian,Gilyen2019QSVT,QET,Dong2022QETU}. Recent advancements have led to substantially improved implementations in terms of the number of ancillae and entangling gates~\cite{Dong2022QETU,berry2024doubling}. However, polynomial transformations fail to accurately approximate non-smooth functions~\cite{Newman1964} such as the discontinuous signum function which results in
impractically high polynomial degree, leading to excessive circuit depth and significant error accumulation on quantum hardware.


We consider rational transformations that excel at approximating functions with discontinuities and singularities, and extend beyond the capabilities of polynomial transformations 
as a subclass~\cite{Tr19}. While polynomial-based algorithms often encounter difficulties with singularities, rational algorithms provide a more robust and general solution, making them a promising alternative. For instance, a rational approximation to the absolute value function achieves fast convergence at a root exponential rate, \textit{i.e.}, the error decays as $e^{-\mathcal{O}(\sqrt{K})}$, where $K$ is the order of the rational approximant.
Techniques such as the rational Krylov projection and interpolation methods~\cite{VBMeMi15,NaSeTr18,CaMeVa19,AnBeGu20} have proven remarkably successful in classical numerical linear algebra. Nevertheless, their potential in quantum computation remains largely unexplored. This is partly because a rational approximation necessitates the computation of resolvents, a process typically reliant on quantum linear solvers that are resource intensive. Unlike previous work on matrix function evaluation~\cite{Tong2021} that assumes a block-encoding model for accessing the operator inverse, here we exploit different kernel representations of the resolvent. These representations can be constructed with Hamiltonian simulations of reasonably short durations on both digital and analog platforms.

We develop an approach based on real-time evolution for implementing the resolvents and thus rational functions. The main ingredient of our approach is time evolution under an effective Hamiltonian. Since a rational function can always be factorized into a sum of partial fractions (sum of resolvents), our key algorithmic primitive is the linear-combination-of-Hamiltonian-simulations (LCHS)~\cite{LCU2012,LCU2023}, specifically linear-combination-of-Hamiltonian-simulations which in general can be implemented with near-optimal query complexity~\cite{an2023optimalcost,an2023quantumalgorithmlinearnonunitary}. We show that both the maximal and total runtime of our Hamiltonian simulations can be bounded through proper selection of non-uniform time samples. Using a bosonic ancillary degree of freedom~\cite{Wallraff2004,Pirkkalainen2013,Gershon2015,Andersen2015,Gan2020}, \textit{e.g.}, a harmonic oscillator, the duration of time evolution can even be maintained constant, albeit with the trade-off of more demanding ancilla state preparation. As Hamiltonian simulation is a native primitive for quantum computers and falls within the BQP complexity class~\cite{Nielsen_Chuang_2010}, our formalism provides a convenient toolkit to construct resolvents and general rational functions on quantum computers.

We examine the implications of the LCHS construction of resolvents and rational functions. First, we show how to implement a tight approximation of the signum function, a bounded discontinuous function, by querying the Hamiltonian simulation operator $e^{-iHt}$. We rigorously analyze for the first time $(i)$ efficient strategies to sample $t$ from a discrete-time perspective, and $(ii)$ an alternative formulation from a continuous-time perspective. We examine the conditions under which $(i)$ and $(ii)$ become applicable or desirable.
Second, we illustrate the application of rational transformations to the ground and excited state problem by constructing an effective spectral filter. For demonstration, we calculate the eigenenergies of representative spin systems and show that applying a rational filter can accelerate convergence.

The remainder of the paper is organized as follows. In \cref{sec:rationalTransforms}, we motivate the importance of rational functions for approximating matrix functions and set up the mathematical basis for the techniques used in this paper. \cref{sec:resolvent_discrete,sec:resolvent_continuous,sec:filter} contain our main technical results. In \cref{sec:resolvent_discrete,sec:resolvent_continuous}, we examine LCHS schemes for efficiently constructing a single resolvent. We introduce complementary discrete- and continuous-time strategies derived from suitable kernel representations of the resolvent. Using resolvents as building blocks, we discuss the construction of quantum rational transformations in \cref{sec:filter}. In particular, we present recipes to construct effective rational approximations to non-smooth functions, where we will focus on the elementary yet versatile signum function. In \cref{sec:numerics_groundExcitedState}, we apply quantum rational transformations to implement a spectral filter from scratch, which allows us to accurately extract the ground and low-lying excited state energies of a many-body system. We conclude in \cref{sec:conclusions}.
\section{Rational transformations and resolvents}
\label{sec:rationalTransforms}

Polynomials are functions in some variable $z$ of the form $p(z) = \alpha_0 + \alpha_1 z + \cdots + \alpha_k z^k$, with $\alpha_i\in\mathbb{C}$.
Their analytical properties are well known and they can be efficiently manipulated on a classical computer, making them one of the most important building blocks for classical algorithms. 
This is exemplified, \textit{e.g.}, by Krylov subspace approaches for solving eigenvalue problems \cite{La50,Ar51} and matrix function approximations \cite{HoLu97}.
In the case of approximating a function $f$ applied to a Hermitian matrix $H\in\mathbb{C}^{N\times N}$, Krylov methods can be interpreted as finding a good polynomial approximation such that $f(H)\approx p(H)$.
In recent years, quantum algorithms based on polynomials have been extensively studied. For example, QSP \cite{Gilyen2019QSVT} and QSVT \cite{Low2017QSP} transform the input matrix using Chebyshev polynomials while quantum subspace methods \cite{motta2020determining,shen2023estimating,kirby2023exact} rely on similar polynomial approximation ideas.

For certain problems, such as the computation of interior eigenvalues or a matrix function approximation for a function that admits a branch cut or a singularity, polynomial-based algorithms might be unable to provide accurate approximations or require a high degree polynomial to achieve a reasonable accuracy. 
Even in the latter case, high degree polynomials can be computationally prohibitive to construct and to manipulate in subsequent tasks.
For these classes of problems, algorithms based on rational functions can reduce the computational cost significantly \cite{Newman1964,DrKn98, Ruhe1984,VBMeMi15}.
A rational function is the ratio of two polynomials, $r(z) = \frac{\beta_0 + \beta_1 z + \cdots + \beta_\ell z^\ell}{\gamma_0 + \gamma_1 z + \cdots + \gamma_k z^k}$, with $\beta_i,\gamma_j\in\mathbb{C}$ and integers $k,\ell\geq 0$.
For the evaluation of a rational function of a matrix, one typically represents the rational function in its partial fraction form $r(z) = p(z) + \frac{c_1}{z-z_1} + \frac{c_2}{z-z_2} + \cdots + \frac{c_k}{z-z_k}$, where $c_k\in \mathbb{C}$ are scalars, $z_i\in\mathbb{C}$ are poles of the rational function and $p(z)$ is of degree $\ell-k+1$, if $\ell\geq k$, and equal to zero, $p(z)\equiv 0$, otherwise.
For simplicity we assumed that all the poles are distinct, \textit{i.e.}, $z_i\neq z_j$ for $i\neq j$, terms of the form $\frac{1}{(z-z_i)^m}$ can be introduced to the partial fraction in the case of repeated poles. In recent quantum computing literature, rational functions were used to approximate matrix functions \cite{TaOhSoUs20, TaOhSoUs22,Tong2021} and to solve eigenvalue problems \cite{FuYeSa21}.

In \Cref{subsec:ratApprox}, we first review that contour integration naturally leads to a rational approximant of a target function. We emphasize that in this paper we \emph{do not} restrict ourselves to only rational approximations obtained via a contour integration. Instead, any rational function can be generated through a partial fraction representation.
In \cref{subsec:quantumRep}, we discuss an integral representation of the terms appearing in a partial fraction form of a rational matrix function.
Two strategies for approximating integrals are provided in \cref{sec:approx_integral}.
In \cref{sec:realtime_resolvent}, we discuss how the approximation strategies can be used to express a rational matrix function in terms of unitary matrices. 
More precisely, these unitary matrices correspond to real-time evolutions under a Hamiltonian (Hamiltonian simulations), especially suited for efficient simulation on quantum hardware.

\subsection{Classical rational approximation}
\label{subsec:ratApprox}

Cauchy's integral formula is a central theorem in complex analysis and allows us to write every scalar function $f:\Omega \rightarrow \mathbb{C}$ that is analytic over a simply connected region $\Omega \subseteq \mathbb{C}$, as a contour integral over the boundary of $\Omega$.
Numerical integration of Cauchy's contour integral formulation naturally yields the following rational approximation
\begin{align}
    f(\omega) = \frac{1}{2\pi i} \oint_{\partial \Omega} dz \frac{f(z)}{z - \omega}  \approx \sum_{k=0}^{K-1} \frac{c_k}{z_k - \omega} =: r(\omega), \qquad \omega \in \Omega, \quad  z_k \in \partial \Omega,\label{eq:Cauchy_integral_formula}
\end{align}
where the \emph{contour} $\partial \Omega$ denotes the boundary of $\Omega$ and $\{z_k\}_{k=0}^{K-1}$ form a suitable discretization of the contour together with weights $\{c_k\}_{k=0}^{K-1}$ .
The rational approximation, $r$, is thus represented in its partial fraction form.
For matrix functions, we have
\begin{align}
   f(H) \approx \sum_{k=0}^{K-1} {c_k}{(z_k - H)^{-1}} = r(H), \qquad z_k \in \partial \Omega,
   \label{eq:Cauchy_integral_discretized_matrix}
\end{align}
where $\partial \Omega$ is chosen such that it encircles a subset of the spectrum of $H$ of interest. 
Typically, the whole spectrum of $H$ is encircled \cite{Higham}, but for important special cases only a subset is encircled, e.g., for counting the number of eigenvalues in an interval \cite{NaPoSa16}.
The factors in the sum, $R(z_k) := (z_k-H)^{-1}$ are called the \emph{resolvents} and $z_k$ the poles of the corresponding resolvents.
For example, if we are interested in the action of $f$ on a subset of negative eigenvalues of $H$, a contour can be chosen and discretized accordingly to encircle only those eigenvalues, as illustrated in \cref{fig:ContourDiscretization}. 
\begin{figure}[!ht]
    \centering
    \setlength\figureheight{4cm}
    \setlength\figurewidth{8cm}	
%
~\hfill%
\begin{tikzpicture}
\begin{axis}[%
axis equal image,
height=\figureheight,
scale only axis,
xmin=-3,
xmax=2,
xtick={-3,-2,...,2},
ymin=-2,
ymax=2,
axis background/.style={fill=white}
]
\addplot [color=myred, forget plot]
  table[row sep=crcr]{%
0	0\\
-0.0020133	0.063424\\
-0.0080452	0.12659\\
-0.018071	0.18925\\
-0.032051	0.25115\\
-0.049929	0.31203\\
-0.071632	0.37166\\
-0.097073	0.42979\\
-0.12615	0.4862\\
-0.15875	0.54064\\
-0.19473	0.59291\\
-0.23396	0.64279\\
-0.27627	0.69008\\
-0.32149	0.73459\\
-0.36945	0.77615\\
-0.41994	0.81458\\
-0.47277	0.84973\\
-0.52773	0.88145\\
-0.58458	0.90963\\
-0.64311	0.93415\\
-0.70308	0.9549\\
-0.76424	0.97181\\
-0.82635	0.98481\\
-0.88916	0.99384\\
-0.95242	0.99887\\
-1.0159	0.99987\\
-1.0792	0.99685\\
-1.1423	0.98982\\
-1.2048	0.9788\\
-1.2665	0.96384\\
-1.3271	0.945\\
-1.3863	0.92235\\
-1.4441	0.89599\\
-1.5	0.86603\\
-1.5539	0.83257\\
-1.6056	0.79576\\
-1.6549	0.75575\\
-1.7015	0.71269\\
-1.7453	0.66677\\
-1.7861	0.61816\\
-1.8237	0.56706\\
-1.858	0.51368\\
-1.8888	0.45823\\
-1.9161	0.40093\\
-1.9397	0.34202\\
-1.9595	0.28173\\
-1.9754	0.22031\\
-1.9874	0.158\\
-1.9955	0.095056\\
-1.9995	0.031728\\
-1.9995	-0.031728\\
-1.9955	-0.095056\\
-1.9874	-0.158\\
-1.9754	-0.22031\\
-1.9595	-0.28173\\
-1.9397	-0.34202\\
-1.9161	-0.40093\\
-1.8888	-0.45823\\
-1.858	-0.51368\\
-1.8237	-0.56706\\
-1.7861	-0.61816\\
-1.7453	-0.66677\\
-1.7015	-0.71269\\
-1.6549	-0.75575\\
-1.6056	-0.79576\\
-1.5539	-0.83257\\
-1.5	-0.86603\\
-1.4441	-0.89599\\
-1.3863	-0.92235\\
-1.3271	-0.945\\
-1.2665	-0.96384\\
-1.2048	-0.9788\\
-1.1423	-0.98982\\
-1.0792	-0.99685\\
-1.0159	-0.99987\\
-0.95242	-0.99887\\
-0.88916	-0.99384\\
-0.82635	-0.98481\\
-0.76424	-0.97181\\
-0.70308	-0.9549\\
-0.64311	-0.93415\\
-0.58458	-0.90963\\
-0.52773	-0.88145\\
-0.47277	-0.84973\\
-0.41994	-0.81458\\
-0.36945	-0.77615\\
-0.32149	-0.73459\\
-0.27627	-0.69008\\
-0.23396	-0.64279\\
-0.19473	-0.59291\\
-0.15875	-0.54064\\
-0.12615	-0.4862\\
-0.097073	-0.42979\\
-0.071632	-0.37166\\
-0.049929	-0.31203\\
-0.032051	-0.25115\\
-0.018071	-0.18925\\
-0.0080452	-0.12659\\
-0.0020133	-0.063424\\
0	-2.4493e-16\\
};
\addplot [color=mygreen, draw=none, mark size=1.7pt, mark=*, mark options={solid, mygreen}, forget plot]
  table[row sep=crcr]{%
0.27484	0\\
};
\addplot [color=mygreen, draw=none, mark size=1.7pt, mark=*, mark options={solid, mygreen}, forget plot]
  table[row sep=crcr]{%
-1.5447	0\\
};
\addplot [color=mygreen, draw=none, mark size=1.7pt, mark=*, mark options={solid, mygreen}, forget plot]
  table[row sep=crcr]{%
1.4238	0\\
};
\addplot [color=mygreen, draw=none, mark size=1.7pt, mark=*, mark options={solid, mygreen}, forget plot]
  table[row sep=crcr]{%
0.46059	0\\
};
\addplot [color=mygreen, draw=none, mark size=1.7pt, mark=*, mark options={solid, mygreen}, forget plot]
  table[row sep=crcr]{%
0.20276	0\\
};
\addplot [color=mygreen, draw=none, mark size=1.7pt, mark=*, mark options={solid, mygreen}, forget plot]
  table[row sep=crcr]{%
0.77287	0\\
};
\addplot [color=mygreen, draw=none, mark size=1.7pt, mark=*, mark options={solid, mygreen}, forget plot]
  table[row sep=crcr]{%
0.39211	0\\
};
\addplot [color=mygreen, draw=none, mark size=1.7pt, mark=*, mark options={solid, mygreen}, forget plot]
  table[row sep=crcr]{%
-1.4122	0\\
};
\addplot [color=mygreen, draw=none, mark size=1.7pt, mark=*, mark options={solid, mygreen}, forget plot]
  table[row sep=crcr]{%
-0.24351	0\\
};
\addplot [color=mygreen, draw=none, mark size=1.7pt, mark=*, mark options={solid, mygreen}, forget plot]
  table[row sep=crcr]{%
-0.21908	0\\
};
\end{axis}
\end{tikzpicture}%
\hfill%
\begin{tikzpicture}
\begin{axis}[%
axis equal image,
height=\figureheight,
scale only axis,
xmin=-3,
xmax=2,
xtick={-3,-2,...,2},
ymin=-2,
ymax=2,
axis background/.style={fill=white}
]
\addplot [color=myred, forget plot]
  table[row sep=crcr]{%
0	0\\
-0.0020133	0.063424\\
-0.0080452	0.12659\\
-0.018071	0.18925\\
-0.032051	0.25115\\
-0.049929	0.31203\\
-0.071632	0.37166\\
-0.097073	0.42979\\
-0.12615	0.4862\\
-0.15875	0.54064\\
-0.19473	0.59291\\
-0.23396	0.64279\\
-0.27627	0.69008\\
-0.32149	0.73459\\
-0.36945	0.77615\\
-0.41994	0.81458\\
-0.47277	0.84973\\
-0.52773	0.88145\\
-0.58458	0.90963\\
-0.64311	0.93415\\
-0.70308	0.9549\\
-0.76424	0.97181\\
-0.82635	0.98481\\
-0.88916	0.99384\\
-0.95242	0.99887\\
-1.0159	0.99987\\
-1.0792	0.99685\\
-1.1423	0.98982\\
-1.2048	0.9788\\
-1.2665	0.96384\\
-1.3271	0.945\\
-1.3863	0.92235\\
-1.4441	0.89599\\
-1.5	0.86603\\
-1.5539	0.83257\\
-1.6056	0.79576\\
-1.6549	0.75575\\
-1.7015	0.71269\\
-1.7453	0.66677\\
-1.7861	0.61816\\
-1.8237	0.56706\\
-1.858	0.51368\\
-1.8888	0.45823\\
-1.9161	0.40093\\
-1.9397	0.34202\\
-1.9595	0.28173\\
-1.9754	0.22031\\
-1.9874	0.158\\
-1.9955	0.095056\\
-1.9995	0.031728\\
-1.9995	-0.031728\\
-1.9955	-0.095056\\
-1.9874	-0.158\\
-1.9754	-0.22031\\
-1.9595	-0.28173\\
-1.9397	-0.34202\\
-1.9161	-0.40093\\
-1.8888	-0.45823\\
-1.858	-0.51368\\
-1.8237	-0.56706\\
-1.7861	-0.61816\\
-1.7453	-0.66677\\
-1.7015	-0.71269\\
-1.6549	-0.75575\\
-1.6056	-0.79576\\
-1.5539	-0.83257\\
-1.5	-0.86603\\
-1.4441	-0.89599\\
-1.3863	-0.92235\\
-1.3271	-0.945\\
-1.2665	-0.96384\\
-1.2048	-0.9788\\
-1.1423	-0.98982\\
-1.0792	-0.99685\\
-1.0159	-0.99987\\
-0.95242	-0.99887\\
-0.88916	-0.99384\\
-0.82635	-0.98481\\
-0.76424	-0.97181\\
-0.70308	-0.9549\\
-0.64311	-0.93415\\
-0.58458	-0.90963\\
-0.52773	-0.88145\\
-0.47277	-0.84973\\
-0.41994	-0.81458\\
-0.36945	-0.77615\\
-0.32149	-0.73459\\
-0.27627	-0.69008\\
-0.23396	-0.64279\\
-0.19473	-0.59291\\
-0.15875	-0.54064\\
-0.12615	-0.4862\\
-0.097073	-0.42979\\
-0.071632	-0.37166\\
-0.049929	-0.31203\\
-0.032051	-0.25115\\
-0.018071	-0.18925\\
-0.0080452	-0.12659\\
-0.0020133	-0.063424\\
0	-2.4493e-16\\
};
\addplot [color=mygreen, draw=none, mark size=1.7pt, mark=*, mark options={solid, mygreen}, forget plot]
  table[row sep=crcr]{%
0.27484	0\\
};
\addplot [color=mygreen, draw=none, mark size=1.7pt, mark=*, mark options={solid, mygreen}, forget plot]
  table[row sep=crcr]{%
-1.5447	0\\
};
\addplot [color=mygreen, draw=none, mark size=1.7pt, mark=*, mark options={solid, mygreen}, forget plot]
  table[row sep=crcr]{%
1.4238	0\\
};
\addplot [color=mygreen, draw=none, mark size=1.7pt, mark=*, mark options={solid, mygreen}, forget plot]
  table[row sep=crcr]{%
0.46059	0\\
};
\addplot [color=mygreen, draw=none, mark size=1.7pt, mark=*, mark options={solid, mygreen}, forget plot]
  table[row sep=crcr]{%
0.20276	0\\
};
\addplot [color=mygreen, draw=none, mark size=1.7pt, mark=*, mark options={solid, mygreen}, forget plot]
  table[row sep=crcr]{%
0.77287	0\\
};
\addplot [color=mygreen, draw=none, mark size=1.7pt, mark=*, mark options={solid, mygreen}, forget plot]
  table[row sep=crcr]{%
0.39211	0\\
};
\addplot [color=mygreen, draw=none, mark size=1.7pt, mark=*, mark options={solid, mygreen}, forget plot]
  table[row sep=crcr]{%
-1.4122	0\\
};
\addplot [color=mygreen, draw=none, mark size=1.7pt, mark=*, mark options={solid, mygreen}, forget plot]
  table[row sep=crcr]{%
-0.24351	0\\
};
\addplot [color=mygreen, draw=none, mark size=1.7pt, mark=*, mark options={solid, mygreen}, forget plot]
  table[row sep=crcr]{%
-0.21908	0\\
};
\addplot [color=myblue, draw=none, mark size=5.0pt, mark=+, mark options={solid, myblue}, forget plot]
  table[row sep=crcr]{%
-0.29289	0.70711\\
-1	1\\
-1.7071	0.70711\\
-1.7071	-0.70711\\
-1	-1\\
-0.29289	-0.70711\\
};
\addplot [color=black, draw=none, mark size=5.0pt, mark=asterisk, mark options={solid, black}, forget plot]
  table[row sep=crcr]{%
0	0\\
-2	1.2246e-16\\
};
\end{axis}
\end{tikzpicture}%
\hfill~    
    
    \caption{Illustration of numerical contour integration. The eigenvalues are marked as (\greencircle) and a contour, in red, encircles the negative eigenvalues. On the right, a discretization of the contour is shown with $K =8$ equidistant points on the circle $z_k = -1 + \cos(2 k \pi/8)+i\sin(2 k\pi/8)$. 
    This leads to 6 complex poles (${\color{myblue} + }$) and 2 real poles ($\ast$), \textit{i.e.}, 6 resolvents to be computed for $z_k\notin \mathbb{R}$ and 2 resolvents for $z_k\in \mathbb{R}$.}
    \label{fig:ContourDiscretization}
\end{figure}
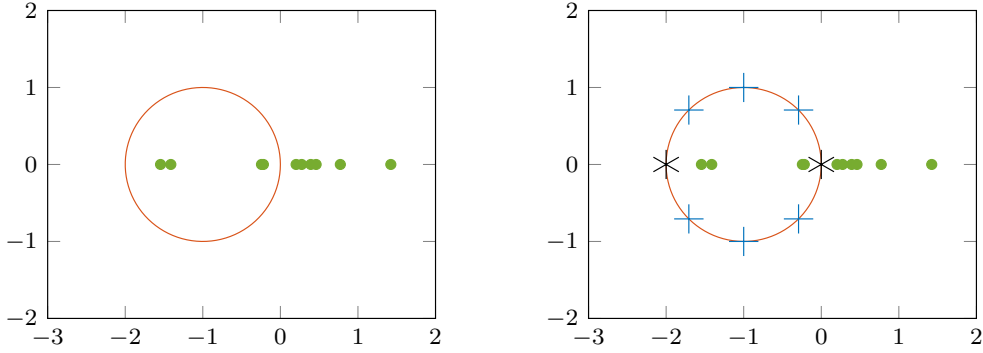

\subsection{Quantum-native representation of resolvents}
\label{subsec:quantumRep}

Although the resolvent is typically non-unitary, we aim to evaluate the resolvent using, ideally simple, unitary circuits that can be evaluated on a quantum computer. In this work, we consider Hamiltonian simulations $e^{-iHt}$ as the central algorithmic ingredient due to its feasibility on both digital and analog quantum platforms. Specifically, we can represent the resolvent in terms of unitary time evolution circuits through the following integral transformation
\begin{align}
    &R(z_k) = (z_k - H)^{-1} =  \int_{0}^{\infty} dq \int_{-\infty}^{\infty} dy \, \kernel_{k}(q,y) e^{iy(z_k - H)q},
    \label{eq:resolvent_unitary}
\end{align}
defined by some suitable choice of the integral kernel $\kernel_k: \mathbb{R} \times \mathbb{R} \rightarrow \mathbb{C}$ \cite{ChKoSo17,KeDuWa21,Yo95}. We immediately observe that the term $e^{iy(z_k-H)q} = e^{iyz_k q} e^{-iHyq}$ can be implemented by a time evolution of duration $t=yq$ under the Hamiltonian $H$. 
As the cost of Hamiltonian simulation scales with simulation time,
we seek a kernel for which $\kernel_k(q,y) e^{iy z_k q}$ decays sufficiently rapidly. In addition to satisfying \cref{eq:resolvent_unitary}, this requirement can impose further restrictions on the functional form of the kernel.

The distinction between poles off the real line (complex poles) and on the real line (real poles) turns out to be essential, and we will adopt different integral representations for these two separate cases.
For a complex  $z_k = a_k + i b_k$, with $a_k,b_k\in \mathbb{R}$, the imaginary part $i b_k$ leads to an exponentially decaying integrand. This means that a Dirac-delta kernel $\kernel_{k}(q,y) = - i \delta(y-1)$ can be employed for the case $\Im z_k=b_k > 0$, which reduces the double integral in \cref{eq:resolvent_unitary} to a single integral,
\begin{align}
    R(z_k) = - i \int_{0}^{\infty} dq \,  e^{-b_k q} \, e^{i(a_k - H)q}, \qquad z_k = a_k + ib_k,
    \label{eq:resolvent_I}
\end{align}
equivalent to the Laplace transform of the matrix inverse (see also \cite[Chapter 9]{Yo95}).
The restriction to $\Im z_k > 0$ is without loss of generality, since for a pole $z_k$ with $\Im z_k < 0$ we can use the identity $R^{\dagger}(z_k) = R(z_k^\ast)$, where $\dagger$ and $\ast$ denote the Hermitian and complex conjugation, respectively.

For a real pole $z_k = a_k\in\mathbb{R}$, the integrand in \cref{eq:resolvent_unitary} picks up a purely oscillatory phase without any decay. In this case, a rapid exponential decay can be introduced by choosing an appropriate Gaussian kernel~\cite{ChKoSo17}. 
Specifically, when the pole is outside the spectral range, \textit{i.e.}, $z_k = a_k \in \mathbb{R} \setminus [E_0, E_{N-1}]$ with $E_0$ and $E_{N-1}$ being the smallest and largest eigenvalues of $H$, the Gaussian kernel,
    $\kernel_{k}(y) = \frac{1}{\pi} {\rm sgn}(a_k - E_0) e^{-y^2/2}$,
leads to the double integral formulation,
\begin{align}
    R(z_k) = \frac{1}{\pi} {\rm sgn}( a_k - E_0 ) \int_{0}^{\infty} dq \, \int_{-\infty}^{\infty} dy \, e^{-y^2/2} e^{iy(a_k - H)q}.
    \label{eq:resolvent_II}
\end{align}
We remark that when the pole does fall in the spectral range, \textit{i.e.}, $a_k \in [E_0, E_{N-1}]$ and $a_k\neq E_n$ for $n=0,1,\dots,N-1$, a modified Gaussian kernel $\frac{i}{2\pi} q e^{-y^2/2}$ with an additional linear factor $q$ realizes a suitable resolvent representation. We will not discuss this case further, but the results presented in the paper can be readily generalized.

\subsection{Approximating the integral representation}
\label{sec:approx_integral}
A quantum-native representation of resolvents requires the computation or approximation of the integral in \cref{eq:resolvent_unitary}. We consider approximations using the generic integral,
\begin{align}
		\int_{a}^{b}dt \, f(t) \weight(t),
  \label{eq:classical_integral}
\end{align}
where $f(t)$ is the integrand function, $\weight(t)$ is the weight function, and $a,b\in \mathbb{R}$ define the integration bounds. The weight function might seem redundant at first glance, but a good pairing between $f(t)$ and $\weight(t)$ will significantly speed up the convergence of an approximation to the integral. We focus on two complementary approaches to approximate integrals: (i) the \emph{discrete approach} sums $f(t)$ at a finite set of time points, and (ii) the \emph{continuous approach} evaluates the integral of the expansion of $f(t)$ in a finite set of basis functions.

The discrete approach relies on an optimal quadrature rule \cite{DaRa84,Ga81}.
A quadrature rule consists of a (well-chosen) set of $J$ samples $\{t_j\}_{j=0}^{J-1}$ and accompanying weights $\{w_j\}_{j=0}^{J-1}$ such that
\begin{align}
		\int_{a}^{b}dt \, f(t) \weight(t) \approx  \sum_{j=0}^{J-1} w_j f(t_j), \qquad t_j\in[a,b],
  \label{eq:quadRule}
\end{align}
is a good approximation. For $\weight(t)\equiv 1$, Newton proposed this approximation for equidistant time samples and Cotes derived a closed-form formula for the weights \cite{Ga81}. Implicitly, this approach is underpinned by the polynomial interpolation idea. A polynomial $p_{J-1}(t)$ interpolating $f(t)$ at $J$ samples $\{t_j\}_{j=0}^{J-1}$ can be constructed and the weights $\{w_j\}_{j=0}^{J-1}$ are obtained from $\int_{a}^{b} p_{J-1}(t) dt$, since polynomials can be integrated analytically. The Newton--Cotes formula, with fixed samples, generally provides an approximation exact for all polynomials up to degree $J-1$.
    Later, Gauss showed that by adjusting the time samples based on the weight function $\weight(t)$ and interval $[a,b]$, it is possible to achieve an approximation exact for all polynomials up to degree $2J-1$ with only $J$ nodes and weights \cite{Ga81}, which is optimal. This insight underscores the importance of choosing the right nodes and weights for a specific integral. Because of the (implicit) polynomial interpolation, the efficiency of our discrete approach depends on how well the function $f(t)$ is approximated by a polynomial. Therefore, different choices of $f(t)$ and $\weight(t)$ can influence the rate of convergence.

    In the continuous approach, instead of implicitly interpolating the integrand $f(t)$ according to the weight function $\weight(t)$, we explicitly expand the weight function itself in a finite basis. In particular, we consider the case that $\weight(t)$ decays sufficiently rapidly, as enabled by the use of Dirac-delta or Gaussian kernels discussed within \cref{subsec:quantumRep}. In this case, a convenient choice of basis functions are the Gaussian functions $e^{- t^2/2\Gamma}$, where the localization parameter $\Gamma$ determines the width of the Gaussian function.
    For a weight function expressed as a linear combination of Gaussian kernels, $\weight(t) = \sum_{j=1}^{G} \alpha_j e^{-t^2/2 \Gamma_j}$, \cref{eq:classical_integral} can be evaluated as a sum $\sum_{j=1}^{G} \alpha_j I_j$, given that the $G$ basis integrals $I_j = \int_a^b dt \, f(t) e^{-t^2/2\Gamma_j}$ can be evaluated efficiently. On the other hand, if the coefficients in our Gaussian basis are not discrete but rather described by a continuous probability distribution, \textit{i.e.}, $\weight(t) \propto \int_{0}^{\infty} d\Gamma \, \rho(\Gamma) e^{-t^2/2\Gamma}$ for a normalized probability density $\rho(\Gamma)$, we may take an additional approximation step to obtain a weighted sum of basis integrals $I_j$. For example, Monte-Carlo sampling \cite{Go16} leads to an unbiased approximation,
    \begin{align}\label{eq:cont_idea}
 		\int_{a}^{b} dt \, f(t) \weight(t) = 		\int_{a}^{b} dt \, f(t) \int_{0}^{\infty} d\Gamma \, \rho(\Gamma) e^{-t^2/2\Gamma} \approx \frac{1}{G} \sum_{j=1}^{G} I_j 
    \end{align}
    where we sample the Gaussian widths $\{\Gamma_j\}_{j=1}^{G}$ i.i.d.~from the density $\rho$. The efficiency of such a continuous approach hinges on both the specific sampling procedure and the ease with which we can evaluate the basis integrals $I_j$.

\subsection{Approximation of resolvents via real-time evolution}
\label{sec:realtime_resolvent}
The integral representation of resolvents can be accurately approximated with both discrete and continuous approaches.
Here, we discuss the general form of these approximations and later detail their quantum-efficient constructions in \cref{sec:resolvent_discrete,sec:resolvent_continuous}, respectively.

\paragraph{Discrete-time approach.} 
Using the quadrature discretization in \cref{eq:quadRule}, the resolvent \eqref{eq:resolvent_unitary} can be computed by evaluating real time evolutions $e^{-iHt_j}$ at a set of discrete times $\{t_j\}_{j}$.
While the trapezoidal rule is commonly suggested in the literature~\cite{KeDuWa21,ChKoSo17}, we propose different quadrature rules that deliver significantly more efficient approximations.

For any 
rational function $r$ in partial fraction form, we can thus write it as a linear combination of Hamiltonian simulations by bootstrapping on the resolvent decomposition, i.e. write down an expansion of the form,
\begin{align}
    r(H) \approx \sum_{j=0}^{J-1} x_{j,d} e^{-iHt_j},\qquad x_{j,d} = x_{j,d}(\boldsymbol{\kappa}, \boldsymbol{c},\boldsymbol{z})~{\rm and}~t_j = t_j(\boldsymbol{c},\boldsymbol{z}),
    \label{eq:rational_LCU}
\end{align}
where $t_j$ gives the simulation time and $x_{j,d}$ the associated `weight' (with the subscript
$d$ indicating discrete-time). The bold symbols in \cref{eq:rational_LCU} represent vector-valued quantities of their respective dimensions. For example, $\bc = (c_0, \ldots, c_{K-1})$ and $\bz = (z_0, \ldots, z_{K-1})$ denote the sets of weights and poles in the approximation \cref{eq:Cauchy_integral_formula} of the rational function. Similarly, we denote our time grid as $\bt = (t_0, \ldots, t_{J-1})$ with
$t_j \leq t_{j+1}$.

\paragraph{Continuous-time approach.} In order to avoid explicit time discretization, we explore alternative Hamiltonian simulations based on \cref{eq:cont_idea}. An exact representation of the resolvent is decomposed into integral components, which are all efficiently constructible using a simple Gaussian basis set as discussed in \cref{sec:approx_integral}. We will show that the Gaussian basis can be physically realized with continuous-variable ancillary wavefunctions~\cite{LCU2023,liu2024hybrid}. Analogous to \cref{eq:rational_LCU}, a rational function can be implemented as a linear combination of real-time evolutions that couple the system and Gaussian ancillae,
\begin{equation}
     r(H) \approx \sum_{j=1}^{G} x_{j,c} \left( {\rm Id} \otimes \bra{0_j}_{c} \right) e^{-i\Tilde{H}} \left( {\rm Id} \otimes \ket{0_j}_{c} \right), \qquad x_{j,c} = x_{j,c}(\boldsymbol{\kappa}, \boldsymbol{c},\boldsymbol{z}),
    \label{eq:rational_LCU2}
\end{equation}
where $\ket{0_j(\boldsymbol{c},\boldsymbol{z})}_{c}$ denotes a Gaussian ancillary state of some characteristic width, $G$ is the total number of such Gaussian states, and
$\Tilde{H}$ describes an effective Hamiltonian that acts jointly on the system and ancillae (${\rm Id}$ is the identity operator on the system Hilbert space). Notably, the effective Hamiltonian is evolved for a unit duration, which can be simulated directly without time sampling. 

\paragraph{Block-encoding and quantum rational transformation.} In both the discrete- and continuous-time approaches, we note that the Hamiltonian simulations can be either performed individually, or combined coherently in a single quantum circuit through block-encoding (BE)~\cite{LCU2012,Berry2015}, a standard routine to represent non-unitary operators on a quantum computer. In addition to time evolution, a block-encoding circuit utilizes two important subroutines, the prepare (PREP) and select (SEL) oracles. For example in the discrete-time approach, the PREP oracle encodes the coefficients $x_j$ into a superposition state $\sum_{j}x_j\ket{j}$ on the auxiliary qubits, and the SEL oracle selectively applies Hamiltonian simulations $e^{-iHt_j}$ conditioned on the auxiliary index $j$. We write $(\alpha, m, \epsilon)$-BE for a block encoding $U_{\rm BE}$ which uses a renormalization factor $\alpha$, requires $m$ ancilla qubits, and incurs additive error $\epsilon$, \textit{i.e.}, $\lVert r(H) - \alpha \bra{0^m} U_{\rm BE} \ket{0^m} \Vert_2 \leq \epsilon$. The connection to BE in the continuous-time setting is further explored in \cref{sec:resolvent_continuous}.

Overall, we will refer to $r(H)$ as a \textit{quantum rational transformation} (QRT) when a rational function $r(\omega)$ is applied to a Hermitian operator $H$ on quantum hardware using techniques such as those exemplified above.

\section{Improved discrete-time LCHS construction of a resolvent}
\label{sec:resolvent_discrete}

The choice of quadrature rules for approximating the integral transformation \cref{eq:resolvent_unitary} is essential for the development of efficient quantum algorithms based on resolvents.
On one hand, the quadrature rule determines the rate of convergence of the approximation in \cref{eq:quadRule} while, on the other hand, it prescribes the Hamiltonian simulations that must be performed on quantum hardware.
\cref{fig:Legendre} outlines the discrete-time LCHS approach for computing a resolvent $R(z)=(z-H)^{-1}$ with $z\notin \mathbb{R}$. 

\begin{figure*}[htbp]
    \centering
    \includegraphics[scale=0.425]{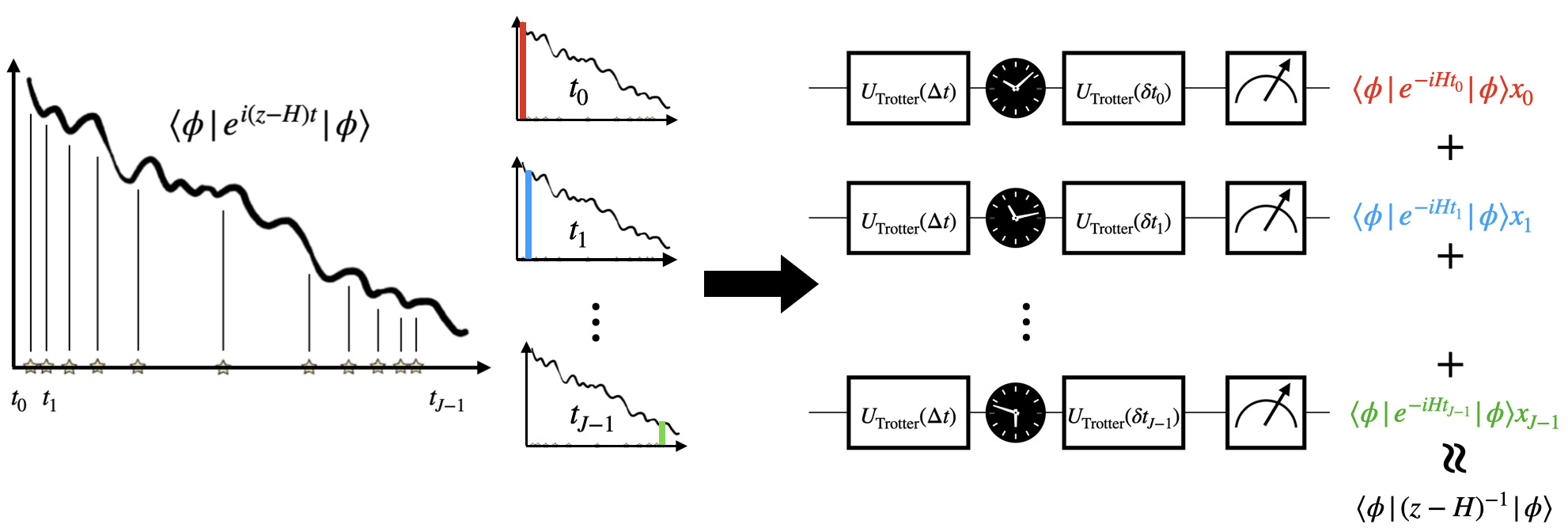}
    \caption{Discrete-time construction of a resolvent $R(z) = (z-H)^{-1}$, with a complex pole $z=a+bi$, via a quadrature rule with nodes $\{t_j\}_{j}$ and
    weights $\{w_j\}_{j}$, \textit{i.e.}, $R(z) \approx \sum_{j} w_j e^{(ia-b)t_j} e^{-iHt_j}$. On quantum hardware the Hamiltonian simulations $U(t_j) = e^{-iHt_j}$ can be performed efficiently. The circuit diagram depicts Trotterized time evolution with multiple fixed steps $\Delta t$ and a final variable step $\delta t_j$, so that $U_{\rm Trotter}(\delta t_j) U_{\rm Trotter}(\Dt)^{M_j} \approx U(t_j)$ for  $t_j = M_j \Delta t + \delta t_j$. Repeated measurements in the computational basis enable us to access scalar quantities such as $\langle \phi \vert U_{\rm Trotter}(\delta t_j) U_{\rm Trotter}(\Dt)^{M_j} \vert \phi \rangle$, which can be collected and combined in post-processing to obtain the desired inner product $\langle \phi\vert R(z) \vert \phi \rangle$.
    }
\label{fig:Legendre}
\end{figure*}

Thus far, previous work has only considered the trapezoidal rule \cite{ChKoSo17,KeDuWa21,TaOhSoUs20,WaXiZh23}, which uses equidistant time nodes but can converge slowly for resolvent approximation.
To address this, we propose using quadrature rules with non-equidistant nodes, which significantly improve convergence.
In \cref{table:quadRules}, we list different rules considered in this work along with their key properties.
The second column shows the node distribution. 
The third and fourth columns quantify two cost metrics that depend on the node distribution. These metrics follow from classical results on orthogonal polynomials \cite{Sz75} and are paramount to our analysis.
For the trapezoidal rule, the weights are uniformly set at $w_j \equiv 1/J$. For the other rules, the nodes are the roots of orthogonal polynomials and the weights are chosen to optimally approximate the integrals from the last column \cite{Ga81}. These nodes and weights are known and can be obtained efficiently. The trapezoidal and Legendre rules employ a constant weight function $\weight(t)\equiv 1$, while the Laguerre and Hermite rules employ rapidly decaying weight functions, $\weight(t)=e^{-t}$ and $\weight(t)=e^{-t^2}$ respectively. These weight functions, in turn, influence the associated quadrature weights.

\begin{table}[tbh!]
\centering
\begin{tabular}{|c|c|c|c|c|}
  \hline
  \makecell{Quadrature \\ rule} & \makecell{Node distribution \\$\{t_j\}_{j=0}^{J-1}$}  &   \makecell{Maximal time \\ $ \max_j\vert t_j\vert$}   &   \makecell{Total time\\ $\sum_j \vert t_j\vert$} & Integral
  \\ \hline
  Trapezoidal & \centering
    \setlength\figureheight{0.01cm}
    \setlength\figurewidth{0.25\textwidth}	
%
\begin{tikzpicture}

\begin{axis}[%
width=0.951\figurewidth,
height=\figureheight,
at={(0\figurewidth,0\figureheight)},
scale only axis,
xmin=0,
xmax=5,
ymin=-1,
ymax=1,
tick style={draw=none},
ytick = \empty,
yticklabels = {},
xtick = {0, 5},
xticklabels = {0, $T$},
axis background/.style={fill=white}
]
\addplot [color=myblue, draw=none, mark=|, mark options={solid, myblue}, forget plot]
  table[row sep=crcr]{%
0	0\\
};
\addplot [color=myblue, draw=none, mark=|, mark options={solid, myblue}, forget plot]
  table[row sep=crcr]{%
0.26316	0\\
};
\addplot [color=myblue, draw=none, mark=|, mark options={solid, myblue}, forget plot]
  table[row sep=crcr]{%
0.52632	0\\
};
\addplot [color=myblue, draw=none, mark=|, mark options={solid, myblue}, forget plot]
  table[row sep=crcr]{%
0.78947	0\\
};
\addplot [color=myblue, draw=none, mark=|, mark options={solid, myblue}, forget plot]
  table[row sep=crcr]{%
1.0526	0\\
};
\addplot [color=myblue, draw=none, mark=|, mark options={solid, myblue}, forget plot]
  table[row sep=crcr]{%
1.3158	0\\
};
\addplot [color=myblue, draw=none, mark=|, mark options={solid, myblue}, forget plot]
  table[row sep=crcr]{%
1.5789	0\\
};
\addplot [color=myblue, draw=none, mark=|, mark options={solid, myblue}, forget plot]
  table[row sep=crcr]{%
1.8421	0\\
};
\addplot [color=myblue, draw=none, mark=|, mark options={solid, myblue}, forget plot]
  table[row sep=crcr]{%
2.1053	0\\
};
\addplot [color=myblue, draw=none, mark=|, mark options={solid, myblue}, forget plot]
  table[row sep=crcr]{%
2.3684	0\\
};
\addplot [color=myblue, draw=none, mark=|, mark options={solid, myblue}, forget plot]
  table[row sep=crcr]{%
2.6316	0\\
};
\addplot [color=myblue, draw=none, mark=|, mark options={solid, myblue}, forget plot]
  table[row sep=crcr]{%
2.8947	0\\
};
\addplot [color=myblue, draw=none, mark=|, mark options={solid, myblue}, forget plot]
  table[row sep=crcr]{%
3.1579	0\\
};
\addplot [color=myblue, draw=none, mark=|, mark options={solid, myblue}, forget plot]
  table[row sep=crcr]{%
3.4211	0\\
};
\addplot [color=myblue, draw=none, mark=|, mark options={solid, myblue}, forget plot]
  table[row sep=crcr]{%
3.6842	0\\
};
\addplot [color=myblue, draw=none, mark=|, mark options={solid, myblue}, forget plot]
  table[row sep=crcr]{%
3.9474	0\\
};
\addplot [color=myblue, draw=none, mark=|, mark options={solid, myblue}, forget plot]
  table[row sep=crcr]{%
4.2105	0\\
};
\addplot [color=myblue, draw=none, mark=|, mark options={solid, myblue}, forget plot]
  table[row sep=crcr]{%
4.4737	0\\
};
\addplot [color=myblue, draw=none, mark=|, mark options={solid, myblue}, forget plot]
  table[row sep=crcr]{%
4.7368	0\\
};
\addplot [color=myblue, draw=none, mark=|, mark options={solid, myblue}, forget plot]
  table[row sep=crcr]{%
5	0\\
};
\end{axis}
\end{tikzpicture}%
   &    \makecell{Controllable \\ $ -\log(\epsilon b)/b$}   & $\frac{TJ}{2}$ & \makecell{$\int_{-1}^{1} f(t) dt$ \\ $f(t)$ periodic!}
  \\ \hline
  Legendre & \centering
    \setlength\figureheight{0.01cm}
    \setlength\figurewidth{0.25\textwidth}	
%
\begin{tikzpicture}

\begin{axis}[%
width=0.951\figurewidth,
height=\figureheight,
at={(0\figurewidth,0\figureheight)},
scale only axis,
xmin=0,
xmax=5,
ymin=-1,
ymax=1,
tick style={draw=none},
ytick = \empty,
yticklabels = {},
xtick = {0, 5},
xticklabels = {0, $T$},
axis background/.style={fill=white}
]
\addplot [color=myblue, draw=none, mark=|, mark options={solid, myblue}, forget plot]
  table[row sep=crcr]{%
0.017179	0\\
};
\addplot [color=myblue, draw=none, mark=|, mark options={solid, myblue}, forget plot]
  table[row sep=crcr]{%
0.09007	0\\
};
\addplot [color=myblue, draw=none, mark=|, mark options={solid, myblue}, forget plot]
  table[row sep=crcr]{%
0.21941	0\\
};
\addplot [color=myblue, draw=none, mark=|, mark options={solid, myblue}, forget plot]
  table[row sep=crcr]{%
0.40221	0\\
};
\addplot [color=myblue, draw=none, mark=|, mark options={solid, myblue}, forget plot]
  table[row sep=crcr]{%
0.63417	0\\
};
\addplot [color=myblue, draw=none, mark=|, mark options={solid, myblue}, forget plot]
  table[row sep=crcr]{%
0.90987	0\\
};
\addplot [color=myblue, draw=none, mark=|, mark options={solid, myblue}, forget plot]
  table[row sep=crcr]{%
1.2228	0\\
};
\addplot [color=myblue, draw=none, mark=|, mark options={solid, myblue}, forget plot]
  table[row sep=crcr]{%
1.5657	0\\
};
\addplot [color=myblue, draw=none, mark=|, mark options={solid, myblue}, forget plot]
  table[row sep=crcr]{%
1.9305	0\\
};
\addplot [color=myblue, draw=none, mark=|, mark options={solid, myblue}, forget plot]
  table[row sep=crcr]{%
2.3087	0\\
};
\addplot [color=myblue, draw=none, mark=|, mark options={solid, myblue}, forget plot]
  table[row sep=crcr]{%
2.6913	0\\
};
\addplot [color=myblue, draw=none, mark=|, mark options={solid, myblue}, forget plot]
  table[row sep=crcr]{%
3.0695	0\\
};
\addplot [color=myblue, draw=none, mark=|, mark options={solid, myblue}, forget plot]
  table[row sep=crcr]{%
3.4343	0\\
};
\addplot [color=myblue, draw=none, mark=|, mark options={solid, myblue}, forget plot]
  table[row sep=crcr]{%
3.7772	0\\
};
\addplot [color=myblue, draw=none, mark=|, mark options={solid, myblue}, forget plot]
  table[row sep=crcr]{%
4.0901	0\\
};
\addplot [color=myblue, draw=none, mark=|, mark options={solid, myblue}, forget plot]
  table[row sep=crcr]{%
4.3658	0\\
};
\addplot [color=myblue, draw=none, mark=|, mark options={solid, myblue}, forget plot]
  table[row sep=crcr]{%
4.5978	0\\
};
\addplot [color=myblue, draw=none, mark=|, mark options={solid, myblue}, forget plot]
  table[row sep=crcr]{%
4.7806	0\\
};
\addplot [color=myblue, draw=none, mark=|, mark options={solid, myblue}, forget plot]
  table[row sep=crcr]{%
4.9099	0\\
};
\addplot [color=myblue, draw=none, mark=|, mark options={solid, myblue}, forget plot]
  table[row sep=crcr]{%
4.9828	0\\
};
\end{axis}
\end{tikzpicture}%
   &    \makecell{Controllable \\ $ -\log(\epsilon b)/b$}   & $\frac{TJ}{2}$ & $\int_{-1}^{1} f(t) dt$
  \\ \hline
  Laguerre & \centering
    \setlength\figureheight{0.01cm}
    \setlength\figurewidth{0.25\textwidth}	
%
\begin{tikzpicture}

\begin{axis}[%
width=0.951\figurewidth,
height=\figureheight,
at={(0\figurewidth,0\figureheight)},
scale only axis,
xmin=0,
xmax=70,
ymin=-1,
ymax=1,
tick style={draw=none},
ytick = \empty,
yticklabels = {},
axis background/.style={fill=white}
]
\addplot [color=myblue, draw=none, mark=|, mark options={solid, myblue}, forget plot]
  table[row sep=crcr]{%
0.07054	0\\
};
\addplot [color=myblue, draw=none, mark=|, mark options={solid, myblue}, forget plot]
  table[row sep=crcr]{%
0.37213	0\\
};
\addplot [color=myblue, draw=none, mark=|, mark options={solid, myblue}, forget plot]
  table[row sep=crcr]{%
0.91658	0\\
};
\addplot [color=myblue, draw=none, mark=|, mark options={solid, myblue}, forget plot]
  table[row sep=crcr]{%
1.7073	0\\
};
\addplot [color=myblue, draw=none, mark=|, mark options={solid, myblue}, forget plot]
  table[row sep=crcr]{%
2.7492	0\\
};
\addplot [color=myblue, draw=none, mark=|, mark options={solid, myblue}, forget plot]
  table[row sep=crcr]{%
4.0489	0\\
};
\addplot [color=myblue, draw=none, mark=|, mark options={solid, myblue}, forget plot]
  table[row sep=crcr]{%
5.6152	0\\
};
\addplot [color=myblue, draw=none, mark=|, mark options={solid, myblue}, forget plot]
  table[row sep=crcr]{%
7.459	0\\
};
\addplot [color=myblue, draw=none, mark=|, mark options={solid, myblue}, forget plot]
  table[row sep=crcr]{%
9.5944	0\\
};
\addplot [color=myblue, draw=none, mark=|, mark options={solid, myblue}, forget plot]
  table[row sep=crcr]{%
12.039	0\\
};
\addplot [color=myblue, draw=none, mark=|, mark options={solid, myblue}, forget plot]
  table[row sep=crcr]{%
14.814	0\\
};
\addplot [color=myblue, draw=none, mark=|, mark options={solid, myblue}, forget plot]
  table[row sep=crcr]{%
17.949	0\\
};
\addplot [color=myblue, draw=none, mark=|, mark options={solid, myblue}, forget plot]
  table[row sep=crcr]{%
21.479	0\\
};
\addplot [color=myblue, draw=none, mark=|, mark options={solid, myblue}, forget plot]
  table[row sep=crcr]{%
25.452	0\\
};
\addplot [color=myblue, draw=none, mark=|, mark options={solid, myblue}, forget plot]
  table[row sep=crcr]{%
29.933	0\\
};
\addplot [color=myblue, draw=none, mark=|, mark options={solid, myblue}, forget plot]
  table[row sep=crcr]{%
35.013	0\\
};
\addplot [color=myblue, draw=none, mark=|, mark options={solid, myblue}, forget plot]
  table[row sep=crcr]{%
40.833	0\\
};
\addplot [color=myblue, draw=none, mark=|, mark options={solid, myblue}, forget plot]
  table[row sep=crcr]{%
47.62	0\\
};
\addplot [color=myblue, draw=none, mark=|, mark options={solid, myblue}, forget plot]
  table[row sep=crcr]{%
55.811	0\\
};
\addplot [color=myblue, draw=none, mark=|, mark options={solid, myblue}, forget plot]
  table[row sep=crcr]{%
66.524	0\\
};
\end{axis}
\end{tikzpicture}%
  & \makecell{Uncontroll. \\ $\approx 4J+2$} & $J^2$ & $\int_{0}^\infty f(t) e^{-t} dt$
  \\ \hline
  Hermite & \centering
    \setlength\figureheight{0.01cm}
    \setlength\figurewidth{0.25\textwidth}	
%
\begin{tikzpicture}

\begin{axis}[%
width=0.951\figurewidth,
height=\figureheight,
at={(0\figurewidth,0\figureheight)},
scale only axis,
xmin=-6,
xmax=6,
ymin=-1,
ymax=1,
tick style={draw=none},
ytick = \empty,
yticklabels = {},
axis background/.style={fill=white}
]
\addplot [color=myblue, draw=none, mark=|, mark options={solid, myblue}, forget plot]
  table[row sep=crcr]{%
-5.3875	0\\
};
\addplot [color=myblue, draw=none, mark=|, mark options={solid, myblue}, forget plot]
  table[row sep=crcr]{%
-4.6037	0\\
};
\addplot [color=myblue, draw=none, mark=|, mark options={solid, myblue}, forget plot]
  table[row sep=crcr]{%
-3.9448	0\\
};
\addplot [color=myblue, draw=none, mark=|, mark options={solid, myblue}, forget plot]
  table[row sep=crcr]{%
-3.3479	0\\
};
\addplot [color=myblue, draw=none, mark=|, mark options={solid, myblue}, forget plot]
  table[row sep=crcr]{%
-2.7888	0\\
};
\addplot [color=myblue, draw=none, mark=|, mark options={solid, myblue}, forget plot]
  table[row sep=crcr]{%
-2.255	0\\
};
\addplot [color=myblue, draw=none, mark=|, mark options={solid, myblue}, forget plot]
  table[row sep=crcr]{%
-1.7385	0\\
};
\addplot [color=myblue, draw=none, mark=|, mark options={solid, myblue}, forget plot]
  table[row sep=crcr]{%
-1.2341	0\\
};
\addplot [color=myblue, draw=none, mark=|, mark options={solid, myblue}, forget plot]
  table[row sep=crcr]{%
-0.73747	0\\
};
\addplot [color=myblue, draw=none, mark=|, mark options={solid, myblue}, forget plot]
  table[row sep=crcr]{%
-0.24534	0\\
};
\addplot [color=myblue, draw=none, mark=|, mark options={solid, myblue}, forget plot]
  table[row sep=crcr]{%
0.24534	0\\
};
\addplot [color=myblue, draw=none, mark=|, mark options={solid, myblue}, forget plot]
  table[row sep=crcr]{%
0.73747	0\\
};
\addplot [color=myblue, draw=none, mark=|, mark options={solid, myblue}, forget plot]
  table[row sep=crcr]{%
1.2341	0\\
};
\addplot [color=myblue, draw=none, mark=|, mark options={solid, myblue}, forget plot]
  table[row sep=crcr]{%
1.7385	0\\
};
\addplot [color=myblue, draw=none, mark=|, mark options={solid, myblue}, forget plot]
  table[row sep=crcr]{%
2.255	0\\
};
\addplot [color=myblue, draw=none, mark=|, mark options={solid, myblue}, forget plot]
  table[row sep=crcr]{%
2.7888	0\\
};
\addplot [color=myblue, draw=none, mark=|, mark options={solid, myblue}, forget plot]
  table[row sep=crcr]{%
3.3479	0\\
};
\addplot [color=myblue, draw=none, mark=|, mark options={solid, myblue}, forget plot]
  table[row sep=crcr]{%
3.9448	0\\
};
\addplot [color=myblue, draw=none, mark=|, mark options={solid, myblue}, forget plot]
  table[row sep=crcr]{%
4.6037	0\\
};
\addplot [color=myblue, draw=none, mark=|, mark options={solid, myblue}, forget plot]
  table[row sep=crcr]{%
5.3875	0\\
};
\end{axis}
\end{tikzpicture}%
  &  \makecell{Uncontroll. \\ $\approx \sqrt{2J}$} & $\frac{3}{5} J^{3/2}$$^{(\star)}$ & $\int_{-\infty}^\infty f(t) e^{-t^2} dt$
  \\ \hline
  \end{tabular}
\caption{Properties of quadrature rules considered for the discrete-time approach to approximate the resolvent $R(z) = (z-H)^{-1}$ with an operator-norm error of $\mathcal{O}(\epsilon)$, where $z = a+bi$. The required number of quadrature nodes $J$ to achieve error $\leq \mathcal{O}(\epsilon)$ depends on the precision $\epsilon$, the choice of quadrature rule, the pole $z$, and Hamiltonian $H$. In the second column we show the distribution of $J=20$ nodes. The cost $T=\max_{j} \vert t_j\vert$ determines the length of integration interval for the trapezoidal or Legendre rule (see \cite{Sz75,Ga81,TrWe14} for details), and therefore the maximal runtime of the $J$ Hamiltonian simulations. The last column tabulates the integral for which the quadrature rule is optimal in the Gaussian sense. Note that any finite interval $[0,T]$ can be mapped linearly onto $[-1,1]$. For the trapezoidal/Legendre rule, the maximal runtime $T$ can be controlled a priori provided that $\epsilon$ and $b$ are known. For the Laguerre/Hermite rule, the maximal time depends on $J$ which is harder to quantify a priori. $^{(\star)}$The total runtime for the Hermite rule is determined empirically.}
\label{table:quadRules}
\end{table}

In this work, we consider the Legendre rule for the $q$-integral and trapezoidal rule for the $y$-integral in \cref{eq:resolvent_unitary}.
In \cref{subsec:Dirac_delta_discrete}, we argue that the Legendre rule is best suited for the $q$-integral, making it the preferred discrete-time LCHS approach for resolvents with complex poles. We further quantify the cost of constructing an $\epsilon$-approximation to a resolvent in terms of the simulation time metrics. 
This is among the first nonasymptotic cost analyses for computing resolvents on quantum hardware.
\cref{subsec:Dirac_delta_discrete_numerics} supports our findings with numerical examples.
\cref{subsec:Gaussian_discrete} shifts focus to quadrature rules for resolvents with real poles.
We propose combining the Legendre and trapezoidal rules,
instead of the common trapezoidal-trapezoidal combination.
A detailed cost analysis is presented, and numerical experiments in \cref{subsec:Gaussian_discrete_numerics} illustrate that our proposed combination outperforms others. While our discussion is restricted to the simple poles, a generalization to repeated poles of higher multiplicities is provided in \cref{sec:repeatedPoles}.

\subsection{Dirac-delta kernel for complex poles}
\label{subsec:Dirac_delta_discrete}
For complex poles, \textit{i.e.}, $z_k = a_k + ib_k$ with $b_k >0$, we consider the single integral representation \eqref{eq:resolvent_I}, and possible discrete-time approximations of the form \eqref{eq:quadRule}, which leads to approximant $\mathcal{I}_J(z_k)$
\begin{align}
    R(z_k) \approx -i\sum_{j=0}^{J-1} w_j  e^{(ia_k-b_k)t_j} e^{-iH t_j} =: -i\sum_{j=0}^{J-1} x_j e^{-iH t_j} =: \mathcal{I}_J(z_k).\label{eq:resolvent_I_discretized}
\end{align}
We absorb the scalar $e^{(ia_k-b_k)t_j}$ into the weights $w_j$, yielding new `weights' $x_j = e^{(ia_k-b_k)t_j}w_j$, to stress that the Hamiltonian simulations $e^{-iHt_j}$ are performed on a quantum computer, and the weighted sum of these real-time evolutions can be done via classical post-processing. 

The following scheme highlights how three important quadrature rules are applied.
All three rules follow the same procedure outlined in \cref{fig:Legendre}, querying a set of Hamiltonian simulations with $H$. Their key difference is in the choice of nodes, \textit{i.e.}, the evolution times, as listed in \cref{table:quadRules}.
\smallskip

{\centering

\begin{tikzpicture}
	\node[draw = none,minimum height = 1cm, text centered,text width=4cm] (R) at (-0.5,1) {$R(z_k)$, see \cref{eq:resolvent_I}};	
	\node[draw = none,minimum height = 1cm, text centered,text width=6.9cm] (I) at (7.4,1) {$ \mathcal{I}_J(z_k) =  -i\sum_{j=0}^{J-1} x_j e^{-iH t_j}$, see \cref{eq:resolvent_I_discretized}};
	\draw[->] (R.east) --node[above,sloped]{Quadrature}  (I.west);	
	\node[draw = none,minimum height = 1cm, text centered,text width=4cm] (Resolvent) at (-0.5,0) {$-i\int_{0}^{\infty} dq~e^{-b_k q} e^{i(a_k-H)q}$};	
	\node[draw = none,minimum height = 1cm, text centered,text width=4.1cm] (Lag) at (6,0) {$\frac{-i}{b_k} \sum_{j=0}^{J-1} w_j^{\rm Lag} e^{\frac{i}{b_k}(a_k-H)t_j^{\rm Lag}}$};	
    \draw[->] (Resolvent.east) --node[above,sloped]{Laguerre} node[below,sloped]{$b_k q \mapsto t$} (Lag.west);	
    
	\node[draw = none,minimum height = 1cm, text centered,text width=3.7cm] (trunc) at (-0.5,-3) {$-i\int_{0}^{T_k^{\rm max}} dq~e^{-b_k q} e^{i(a_k-H)q}$};
	\draw[->] (Resolvent.south) --node[left]{Truncate} node[right]{{$T_k^{\rm max} = -\frac{\log(\epsilon b_k)}{b_k}$}}(trunc.north);	
	\node[draw = none,minimum height = 1cm, text centered,text width=5cm] (Trap) at (6.5,-2) {$-i\sum_{j=0}^{J-1} w_j^{\rm Trap} e^{-b_k t_j^{\rm Trap}} e^{i(a_k-H)t_j^{\rm Trap}}$};	
	
	\node[] (C) at (2,-2){};
	
	\draw[-] (trunc.north east) -- (C.center);
	\draw[->] (C.center) -- node[below] {Trapezoidal} (Trap.west);
	
	\node[draw = none,minimum height = 1cm, text centered,text width=6.2cm] (Leg) at (6.5,-4) { $-i\sum_{j=0}^{J-1}  \frac{w_j^{\rm Leg}T_k^{\rm max}}{2} e^{- \frac{b_kT_k^{\rm max}}{2} (t_j^{\rm Leg}+1)} e^{\frac{iT_k^{\rm max}}{2}(a_k-H) (t_j^{\rm Leg}+1)}$};			
	\node[] (D) at (2,-4){};
	
	\draw[-] (trunc.south east) -- (D.center);
	\draw[->] (D.center) -- node[above] {Legendre} 
 (Leg.west);

    \node[draw=none] (D) at (2.4,-4.5){$q \mapsto \frac{T_{k}^{\rm max}(1+t)}{2}$};

	\draw[black,thick] ($(R.north west)+(-0,0.1)$)  rectangle ($(Leg.south east)+(2.5,-0.4)$);
 
\end{tikzpicture}

}
\smallskip

After a change of variable $b_k q \mapsto t$, \cref{table:quadRules} seemingly suggests the use of a Gauss--Laguerre rule, since we have $\weight(t)=e^{-t}$ and the integration interval $[0,\infty]$. 
Under specific regularity assumptions, the Laguerre rule indeed has the potential to converge exponentially fast~\cite{Lubinsky1983}. However, precise assertions or results about the general rate of convergence remain undetermined.

The trapezoidal and Legendre rule require truncation to a finite interval and use $\weight(t)\equiv 1$.
A truncation with $\epsilon$-control over the error can be obtained thanks to the exponential decay of the integrand, \textit{i.e.},
\begin{align}
    R(z_k)\approx - i \int_{0}^{T_k^{\rm max}} dq \,  e^{-b_k q} \, e^{i(a_k - H)q} , \quad \text{with }T^{\rm max}_{k} = \lVert \bt \rVert_{\infty} = \frac{1}{b_k} \log \frac{1}{\epsilon b_k}.
    \label{eq:Tmax_discrete}
\end{align}
While a trapezoidal discretization, $t^{\rm Trap}_j = j\Dt$ for the fixed timestep $\Dt = \frac{T_k^{\rm max}}{J-1}$, is commonly assumed for real-time algorithms \cite{KeDuWa21}, it can result in a large total simulation time $T^{\rm tot}_{k} = \lVert \bt \rVert_1$.
This is because the trapezoidal rule applied to the nonperiodic truncated integral in \cref{eq:Tmax_discrete} suppresses the approximation error slowly as the number of timesteps increases, \textit{i.e.}, at an algebraic rate of $\mathcal{O}\left( \frac{1}{\Dt^{2}} \right)$ \cite{DaRa84}. 

In this paper, we propose the use of the Gauss--Legendre quadrature rule. Observe that upon changing the variable $q \mapsto \frac{T_{k}^{\rm max}(1+t)}{2}$, 
\begin{align}
    R(z_k) \approx - \frac{i T_{k}^{\rm max}}{2} \int_{-1}^{1} dt \, e^{-b_k T_{k}^{\rm max} (1+t)/2} e^{i(a_k - H)T_{k}^{\rm max} (1+t)/2} \approx \mathcal{I}^{\rm Leg}_{J}(z_k),
    \label{eq:GL_discrete}
\end{align}
where for $\mathcal{I}^{\rm Leg}_J(z_k)$ we take $\bt \in [-1, 1]^{J}$ to be the roots of the degree-$J$ Legendre polynomial. 
For exponentially decaying integrands, it is known that the Legendre rule achieves similar accuracy as the Laguerre rule \cite{Ev05}. Moreover the Legendre rule establishes a robust exponential convergence, which allows us to favorably reduce the number of distinct real-time circuits $J$.
In particular, the convergence follows from the classical results~\cite{Sydow1977,Ga81} that the Gauss--Legendre error for a function $f$, which is analytic within and on the ellipse $\mathcal{E}_{\sigma} = \left\{ \frac{\sigma}{2} e^{i\theta} + \frac{1}{2\sigma} e^{-i\theta}: \theta \in [-\pi, \pi] \right\}$, is bounded by
\begin{align}
   e_{J}[f;\mathcal{E}_{\sigma}] \leq \frac{ 8 \sigma^2 \sup_{t \in \mathcal{E}_{\sigma}} \lvert f(t) \rvert }{\sigma^{2J}(\sigma^2 - 1)},
   \label{eq:GL_error}
\end{align}
where $e_{J}[f]$ labels the difference between the exact and discretized $f$-integral over $[-1, 1]$. Since the integrand in \cref{eq:GL_discrete}, $f(t) = e^{i(z_k-H)T^{\rm max}_k(1+t)/2}$, is analytic on the whole complex plane, this error bound holds for every $\sigma > 1$. With a fixed $\sigma$, the approximation accuracy hence improves exponentially with respect to $J$, highlighting the advantage of the Legendre rule over the algebraic convergence of the trapezoidal rule. Its efficiency, in terms of quantum resources, is summarized in the following theorem.

\begin{thm}
\label{thm:1}
For a complex pole $z_k = a_k + i b_k$ with $b_k > 0$ and tolerance $\epsilon > 0$, the resolvent $R(z_k)$ admits a real-time LCHS construction $\mathcal{I}_{J}(z_k)$ for which $\lVert R(z_k) - \mathcal{I}_{J}(z_k) \rVert_2 
< \epsilon$. The Gauss--Legendre rule defines a time grid $\bt$ such that
\begin{align}
    J =\log_{2} \log \frac{2}{ \epsilon b_k }  + (\eta_k + 1) \log_2 \frac{2}{\epsilon b_k } + 3
    \label{eq:GL_J}
\end{align}
distinct time evolution circuits suffice to construct $\mathcal{I}_{J}(z_k)$, where $\eta_k = \frac{3b_k - 2\sqrt{2} b_k  + a_k^{+}}{4\sqrt{2} b_k}$ and $a_k^{+} = \max_{0 \leq n \leq N-1} \lvert a_k -E_n \rvert$ for eigenvalues $E_n$ of $H$. The corresponding maximal and total evolution time are given by, respectively,
\begin{align}
    T^{\max} = \lVert \bt \rVert_{\infty} = \frac{1}{b_k} \log \frac{2}{\epsilon b_k}, \qquad \text{and} \qquad T^{{\rm tot}} = \lVert \bt \lVert_{1} = \frac{J}{2 b_k} \log \frac{2}{\epsilon b_k}.
    \label{eq:GL_tmax}
\end{align}
\end{thm}
\begin{proof}
    The proof is provided in \cref{subsec:proof_theorem1}.
\end{proof}

The constant $\eta_k$ in \cref{thm:1} is critical for determining the quantum resources since it lacks a logarithmic dependence on the resolvent parameters. Importantly, $\eta_k$, which depends on the complex component of the pole $b_k$ directly and on the real (and complex) component indirectly through the largest distance $a^+_k$ between the pole $z_k$ and the eigenvalues of $H$, can be systematically reduced if we place the pole $z_k$ farther from the real line. We remark that implementing the trapezoidal rule with the exact same maximal runtime necessitates $J = \left[ \frac{1}{\epsilon b_k} \log \frac{2}{\epsilon b_k}  \right]$ Hamiltonian simulations and a total runtime of $T^{\rm tot} = \frac{1}{2 \epsilon b_k^2} \left[ \log \frac{2}{\epsilon b_k} \right]^2$~\cite{KeDuWa21}; our result achieves a respective enhancement by multiplicative factors of $\frac{1}{\epsilon}$ and $\frac{1}{\epsilon} \left[ \log \frac{1}{\epsilon} \right]^{-1}$. Hence, the Gauss--Legendre rule offers a simple, improved LCHS recipe that can lead to substantial resource reduction for resolvent computation.

Given access to the PREP and SEL oracles, \cref{thm:1} states that the discrete-time approach can be alternatively viewed as a $(1, \lceil \log_2 J \rceil, \epsilon)$ BE of the resolvent.

\subsection{Numerical experiments for complex poles}
\label{subsec:Dirac_delta_discrete_numerics}
We now illustrate \cref{thm:1} for the 1D mixed-field Ising model (MFIM),
\begin{align}
     H =  - \sum_{i=1}^{L_{\rm sys}} Z_i Z_{i+1} - h \sum_{i=1}^{L_{\rm sys}} Z_i - g \sum_{i=1}^{L_{\rm sys}} X_i,
     \label{eq:MFIM}
\end{align}
where $h$ and $g$ set the longitudinal and transverse field, respectively, and $L_{\rm sys}$ denotes the number of spins.
For simplicity, we assume $\lVert H \rVert_2 \leq 1$ by rescaling the spectrum $H \mapsto \frac{H}{\Vert H \Vert_2}$.

To implement the real-time evolution $e^{-iHt}$, we employ a Trotter splitting $H = H_1 + H_2$, where $H_1 = -\sum_{i} Z_i Z_{i+1} - h \sum_{i} Z_i $ and $H_2 =- g \sum_{i} X_i$ can be simulated efficiently. 
For $M$ Trotter steps, each of size $\Dt = \frac{t}{M}$, the first-order Trotter scheme reads $e^{-iHt} \approx \left( e^{-iH_1 \Dt} e^{-iH_2 \Dt}\right)^M$ and incurs an error of $\mathcal{O}(t \Dt)$ \cite{Lloyd1996}.
This implies a Trotter step size of $\Dt = \frac{\epsilon}{T^{\rm max}}$ to reach a target accuracy of $\epsilon$ for all $t_j \in\bt$. For equidistant nodes $\bt$ of 
the trapezoidal rule, a fixed Trotter step size suffices to simulate the $J$ evolutions $e^{-i H t_j}$ necessary to construct a resolvent. However, the distance between adjacent Laguerre or Legendre nodes is not constant.
In this case, we take
\begin{align}
        e^{-iHt_j} \approx \left( e^{-iH_1 \Dt} e^{-iH_2 \Dt} \right)^{M_j}  e^{-iH_1  \delta t_j} e^{-iH_2 \delta t_j},
        \label{eq:Trotter_Variable}
\end{align}
where each $t_j\in \bt$ is approached as closely as possible by $M_j$ $\Dt$-steps and a final variable step of size $\delta t_j < \Dt$, as is illustrated in \cref{fig:Trotter}.
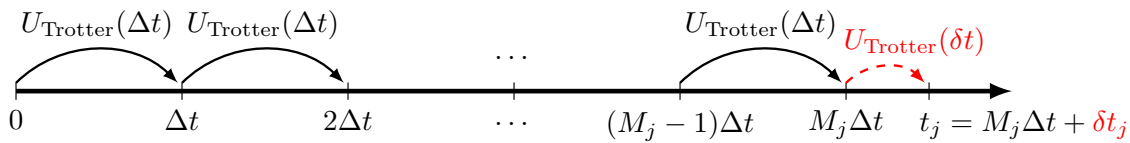
\begin{figure}[hbt!]
    \centering
    \begin{tikzpicture}
\draw[ultra thick, ->] (0,0) -- (12,0);

\foreach \x in {0,2,4,6,8,10,11}
\draw (\x cm,3pt) -- (\x cm,-3pt);

\draw[ultra thick] (0,0) node[below=3pt,thick] {0} node[above=3pt] {};
\draw[ultra thick] (2,0) node[below=3pt,thick] {$\Delta t$} node[above=3pt] {};
\draw[ultra thick] (4,0) node[below=3pt,thick] {$2\Dt$} node[above=3pt] {};
\draw[ultra thick] (6,0) node[below=7pt, thick] {$\dots$} node[above=7pt] {$\dots$};
\draw[ultra thick] (8,0) node[below=3pt, thick] {$(M_j-1)\Dt$} node[above=3pt] {};
\draw[ultra thick] (10,0) node[below=3pt, thick] {$M_j\Dt$} node[above=3pt] {};
\draw[ultra thick] (12.15,0) node[below=3pt, thick] {$t_j=M_j\Dt + {\color{red}\delta t_j}$} node[above=3pt] {};

\draw [thick,->] (0,0.1) to [out=45,in=135] node [above] {$U_{\rm Trotter}(\Dt)$} (1.95,0.1);
\draw [thick,->] (2,0.1) to [out=45,in=135] node [above] {$U_{\rm Trotter}(\Dt)$} (3.95,0.1);
\draw [thick,->] (8,0.1) to [out=45,in=135] node [above] {$U_{\rm Trotter}(\Dt)$} (9.95,0.1);
\draw [thick,->,dashed,red] (10,0.1) to [out=45,in=135] node [above] {$\qquad U_{\rm Trotter}(\delta t)$} (10.95,0.1);

\end{tikzpicture}
    \caption{Simulation of $e^{-iHt_j}$ by Trotterized time evolution $U_{\rm Trotter}(t) := e^{-iH_1 t}e^{-iH_2 t}$. The evolution is advanced for multiple fixed time steps $\Dt$, common for all simulations $t_j\in \bt$, and a variable shorter time step $\delta t_j<\Dt$ dependent on $t_j$.}
    \label{fig:Trotter}
\end{figure}

\paragraph{Convergence as a function of simulation time.}
For the MFIM with $L_{\rm sys} = 8$ spins, we verify that the Legendre rule outperforms the trapezoidal and Laguerre rules. We fix the Hamiltonian model parameters $(h, g) = (1, \frac{2}{3})$, and consider the complex pole $z = -0.8+0.1i$. \cref{fig:TFIM_convergence} displays the approximation error $\Vert R(z)-\mathcal{I}_J(z)\Vert_2$ relative to the scaled total evolution time $\frac{T^{\rm tot}}{T^{\rm max}}$ for the different quadrature rules and requested accuracy of $\epsilon = 10^{-3}$ and $\epsilon = 10^{-6}$.
We take a Trotter step of $\Dt= \frac{\epsilon}{T^{\rm max}}$.
The trapezoidal rule shows an algebraic convergence, while the Laguerre rule converges exponentially and incurs a simulation cost $T^{\rm tot} = \mathcal{O}(J^2)$. The Legendre rule converges exponentially and is the most efficient.
We also note that the onset of the convergence of the Legendre rule occurs earlier for lower requested accuracy, since the truncation $T^{\rm max}$ is smaller in this case, making the truncated integral easier to approximate.

\begin{figure}[!htb]
    \centering
    \setlength\figureheight{2cm}
    \setlength\figurewidth{0.85\textwidth}	
    \input{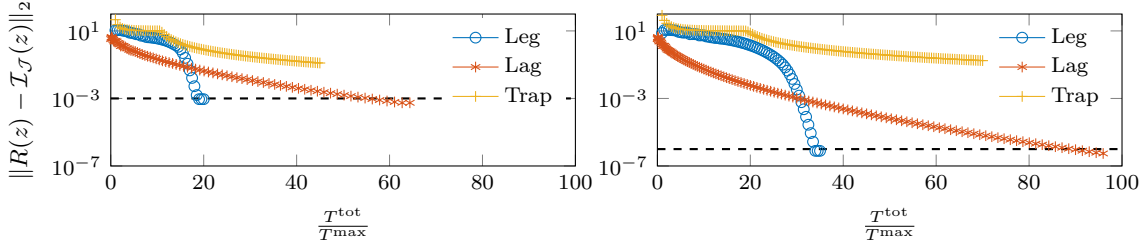}
    
    \caption{Approximation error $\Vert R(z)-\mathcal{I_J}(z)\Vert_2$ as a function of the scaled total evolution time $\frac{T^{\rm tot}}{T^{\rm max}}$, for $z = -0.8+0.1i$ and the (scaled) MFIM Hamiltonian with $L_{\rm sys} =8$ spins. The approximation $\mathcal{I_J}(z)$ is computed using the trapezoidal ({\color{myorange}+}), Legendre ({\color{myblue}$\circ$}) and Laguerre (\textcolor{myred}{$\ast$}) rule for a requested accuracy of $\epsilon$ ($10^{-3}$ left and $10^{-6}$ right) and Trotter step of $\Dt = \frac{\epsilon}{T^{\rm max}}$. 
    }
    \label{fig:TFIM_convergence}
\end{figure}

\paragraph{Simulation time as a function of pole position.}
\cref{eq:GL_tmax} indicates that a pole closer to the real line is more costly to simulate.
For selected poles $z_k$, a target accuracy $\epsilon = 10^{-6}$, and Trotter step $\Dt = \frac{\epsilon}{T^{\rm max}}$, we report the cost $\frac{T_k^{\rm tot}}{T_k^{\rm max}}$ for each of the quadrature rules to obtain an $\epsilon$-approximation satisfying $\Vert R(z_k)-\mathcal{I}_J(z_k)\Vert_2 < \epsilon$.
We also report the predicted cost in terms of the total evolution time for the Legendre rule in \cref{eq:GL_J} and \cref{eq:GL_tmax}.
\cref{fig:approachR} shows these quantities for poles $z_k = -0.8+2^{-k+1}i$ (left pane) and $z_k = 2^{-k+1}i$ (right pane), with $k=1,\dots,6$.
The Legendre rule proves to be the most efficient for all considered poles $z_k$. The Laguerre rule leads to numerical underflow for $z_k=-0.8+2^{-4}i$ and $z_k=2^{-5}i$ before reaching the target accuracy, whereas
the trapezoidal rule appears orders of magnitude more expensive. We only report data where the cost is below or approximately $10^4$.
The predicted cost for the Legendre rule captures the overall trend of the actual cost as the pole moves closer to the real line.
The difference between the two panes is attributed to the largest distance $a^{+}$ of the poles from an eigenvalue measured on the real line (and therefore independent of $k$ in this experiment).
A larger $a^+$ increases the cost for all considered quadrature rules.

\begin{figure}[!ht]
    \centering
    \setlength\figureheight{2cm}
    \setlength\figurewidth{0.85\textwidth}	
%
\begin{tikzpicture}
\begin{axis}[%
width=0.48\textwidth,
xmin=0,
xmax=1,
xlabel=$b$,
ymode=log,
ymin=1,
ymax=1e5,
yminorticks=true,
ylabel=$\frac{T^{\rm tot}}{T^{\rm max}}$,
legend style={draw=none, font=\footnotesize, fill=none},
title={$a^{+} = 1.36$}
]
\addplot [color=myblue, draw=none, mark=o, mark options={solid, myblue, thick}, mark size=2.5]
  table[row sep=crcr]{%
1 6\\
0.5 9\\
0.25 15\\
0.125 28\\
0.0625 53.501\\
0.03125 105.996\\
};
\addlegendentry{Leg}
\addplot [color=myblue, line width=1.25]
  table[row sep=crcr]{%
1 16.7209\\
0.5 20.0176\\
0.25 26.3004\\
0.125 38.7967\\
0.0625 64.2006\\
0.03125 116.3799\\
0.0156 224.0284\\
};
\addlegendentry{Thm 1}
\addplot [color=myred, draw=none, mark=asterisk, mark options={solid, myred, thick}, mark size=2.5]
  table[row sep=crcr]{%
1 13.5092\\
0.5 20.5568\\
0.25 36.2379\\
0.125 72.3712\\
};
\addlegendentry{Lag}
\addplot [color=myorange, draw=none, mark=+, mark options={solid, myorange, thick}, mark size=2.5]
  table[row sep=crcr]{%
1 3453.7\\
0.5 6116.958\\
0.25 11919\\
};
\addlegendentry{Trap}
\end{axis}
\end{tikzpicture}%
\hfill%
\begin{tikzpicture}
\begin{axis}[%
width=0.48\textwidth,
xmin=0,
xmax=1,
xlabel=$b$,
ymode=log,
ymin=1,
ymax=1e5,
yminorticks=true,
legend style={draw=none, font=\footnotesize, fill=none},
title={$a^{+} = 1$}
]
\addplot [color=myblue, draw=none, mark=o, mark options={solid, myblue, thick}, mark size=2.5]
  table[row sep=crcr]{%
1 5.5\\
0.5 7.5000\\
0.25 12.000\\
0.125 21.5\\
0.0625 40.5\\
0.03125 83\\
};
\addlegendentry{Leg}
\addplot [color=myblue,  line width=1.25]
  table[row sep=crcr]{%
1 16.0627\\
0.5 18.6385\\
0.25 23.4162\\
0.125 32.7773\\
0.0625 51.6564\\
0.03125 90.2885\\
0.015625 169.8334\\
};
\addlegendentry{Thm 1}
\addplot [color=myred, draw=none, mark=asterisk, mark options={solid, myred, thick}, mark size=2.5]
  table[row sep=crcr]{%
1 6.8924\\
0.5 9.5055\\
0.25 18.1819\\
0.125 35.8766\\
0.0625 71.9024\\
};
\addlegendentry{Lag}
\addplot [color=myorange, draw=none, mark=+, mark options={solid, myorange, thick}, mark size=2.5]
  table[row sep=crcr]{%
1 3097.4\\
0.5 4986.9\\
0.25 10588\\
};
\addlegendentry{Trap}
\end{axis}
\end{tikzpicture}%
    \caption{Cost metric $\frac{T^{\rm tot}}{T^{\rm max}}$ for the scaled total evolution time for quadrature rules necessary to compute an $\epsilon$-approximation to $(z-H)^{-1}$ with accuracy $\epsilon=10^{-6}$.
    A Trotter step of $\Dt = \frac{\epsilon}{T^{\rm max}}$ is used to evolve the MFIM Hamiltonian with $L_{\rm sys} = 8$ spins. We consider $z=-0.8+bi$ on the left and $z=bi$ on the right, which yields $a^+=1.36$ and $a^+=1$ respectively. The Legendre rule ({\color{myblue}$\circ$}) performs more efficiently than the Laguerre (\textcolor{myred}{$\ast$}) and trapezoidal ({\color{myorange}$+$}) rule. The solid line shows the predicted cost for the Legendre rule derived in \cref{thm:1}.}
    \label{fig:approachR}
\end{figure}
We comment that similar results for Heisenberg spin systems of varying dimensions and lattice geometries are presented in \cref{app:Heisenberg}.

\subsection{Gaussian kernel for real poles}
\label{subsec:Gaussian_discrete}
For a real pole outside the spectral range, \textit{i.e.}, $z_k = a_k \in \mathbb{R}$
and $a_k \notin [E_0, E_{N-1}]$, the resolvent can be represented as the double integral of \cref{eq:resolvent_II}.
For clarity and ease of notation, we will omit the overall sign factor from here on.
This double integral can be discretized in two steps. First, we discretize the inner $y$-integral,
\begin{align}
    R(z_k) \approx  \frac{1}{\pi}  \sum_{\ell =0}^{L_y-1} x_{y,\ell} \int_{0}^{\infty} dq \, e^{-y_\ell^2/2} e^{i y_\ell (a_k - H) q }.
    \label{eq:resolvent_II_marginal}
\end{align}
Second, by recognizing \cref{eq:resolvent_II_marginal} as a sum of $L_y$ integrals of the form as in \cref{eq:resolvent_I}, we discretize the remaining integrals over $q$ using $L_q$ nodes as before, \textit{i.e.}, by repeating \cref{eq:quadRule} $L_y$ times. The resulting approximation is also denoted as $\mathcal{I}_J(z_k) \approx R(z_k)$. 

The $q$-integrals can be efficiently approximated by the Legendre rule as justified in \cref{subsec:Dirac_delta_discrete}.
Thus, for the remainder of this section we only examine discretization of the inner $y$-integral. The following scheme provides an overview of the quadrature rules to be discussed.
\bigskip

{\centering

	\begin{tikzpicture}
		\node[draw = none,minimum height = 1cm, text centered,text width=4cm] (R) at (-0.5,1) {$R(z_k)$, see \cref{eq:resolvent_II}};	
		\node[draw = none,minimum height = 1cm, text centered,text width=7.3cm] (I) at (7.6,1) {$\mathcal{I}_J(z_k)$, see \cref{eq:resolvent_II_marginal} combined with \cref{eq:resolvent_I_discretized}};
		\draw[->] (R.east) --node[above,sloped]{Quadrature}  (I.west);	
		
		\node[draw = none,minimum height = 1cm, text centered,text width=5cm] (Resolvent) at (-0.5,0) {$\int_{0}^{\infty} dq \, \int_{-\infty}^{\infty} dy \, e^{-y^2/2} e^{iy(a_k - H)q}$};	
		\node[draw = none,minimum height = 1cm, text centered,text width=4cm] (Her) at (6,0) {$\sum_{l=0}^{L-1} \int_{0}^{\infty} dq \,   e^{i\sqrt{2}y_l^{\rm Her}(a_k - H)q}$};	
		\draw[->] (Resolvent.east) --node[above,sloped]{Hermite} node[below,sloped]{$\frac{y}{\sqrt{2}}\mapsto t$} (Her.west);	
		\node[draw = none,minimum height = 1cm, text centered,text width=4.8cm] (trunc) at (-0.5,-3) {$\int_{0}^{\infty} dq \, \int_{-y_k^{\rm max}}^{y_k^{\rm max}} dy \, e^{-y^2/2} e^{iy(a_k - H)q}$};
		\draw[->] (Resolvent.south) --node[left]{Truncate} node[right]{{$y_k^{\rm max} = \sqrt{ \log \frac{1}{\epsilon a_k^{-} }}$}}(trunc.north);	
		\node[draw = none,minimum height = 1cm, text centered,text width=7cm] (Trap) at (7.5,-2) {$\sum_{l=0}^{L-1} w_l^{\rm Trap} e^{-\frac{(y_l^{\rm Trap})^2}{2}} \int_{0}^{\infty} dq \,   e^{iy_l^{\rm Trap}(a_k - H)q}$};	
		
		\node[] (C) at (2,-2){};
		
		\draw[-] (0.5,-2.5) -- (C.center);
		\draw[->] (C.center) -- node[below] {Trapezoidal} (Trap.west);
		
		\node[draw = none,minimum height = 1cm, text centered,text width=8cm] (Leg) at (7,-4) { $y_k^{\rm max}\sum_{l=0}^{L-1}  \int_{0}^\infty dq~ w_l^{\rm Leg} e^{-\frac{(y_k^{\rm max} t_l^{\rm Leg})^2}{2}} e^{i y_k^{\rm max} t_l^{\rm Leg} (a_k-H)q}$ };			
		\node[] (D) at (1,-4){};
		
		\draw[-] (0.5,-3.5) -- (D.center);
		\draw[->] (D.center) -- node[above] {Legendre} node[below]{\small $\frac{y}{y_k^{\rm max}}\mapsto t$} (Leg.west);
		
		\draw[black,thick] ($(R.north west)+(-0.5,0.1)$)  rectangle ($(Leg.south east)+(0.5,-0.2)$);
			
	\end{tikzpicture}

}
\smallskip

Since the integration interval is the whole real line $[-\infty, \infty]$, Gauss--Hermite quadrature appears to be a natural candidate as suggested by \cref{table:quadRules}. 
The Gauss--Hermite rule can be applied to the inner $y$-integral in \cref{eq:resolvent_II} after a change of variable $\frac{y}{\sqrt{2}} \mapsto t$, which reveals the weight function $\weight(t)=e^{-t^2}$. Unfortunately, the lack of a tractable error analysis generally complicates the estimation of quantum resources, much like the case of the Laguerre quadrature rule for the $q$-integral.

In this work, we $\epsilon$-truncate the $y$-integral by exploiting the exponential decay of the integrand,
\begin{equation}
  R(z_k) \approx \frac{1}{\pi} \int_0^{\infty} dq \int_{-y^{\rm max}}^{y^{\rm max}} dy e^{-y^2/2} e^{iy(a_k-H)q}, \quad \text{with }   \pm y^{\rm max} = \pm \sqrt{ \log \frac{1}{\epsilon a_k^{-} } }.
\end{equation}
This truncation allows for the use of the trapezoidal and Legendre rule, these rules imply $w(t)\equiv 1$.
The Legendre rule achieves exponential convergence, which follows from similar arguments as those in \cref{subsec:Dirac_delta_discrete}.
The trapezoidal rule, $y_\ell = (\ell - \frac{L-1}{2}) \Dy$ with $\Dy = \frac{2 y^{\rm max}}{L-1}$, also converges exponentially based on the fact that the integrand function is quasi-periodic on the interval $[-y^{\rm max}, y^{\rm max}]$, \textit{i.e.}, $f(-y^{\rm max}) = f(y^{\rm max})+\mathcal{O}(\delta)$ for a small $\delta$ that can be controlled by the truncation parameter $y^{\rm max}$. 
The trapezoidal rule converges exponentially up to an accuracy $\mathcal{O}(\delta)$: this is a well-established result for integration of rapidly decaying analytic functions~\cite{TrWe14}. In particular, the trapezoidal error for a function $f$ that is analytic in the horizontal strip, $\mathcal{S}_{\sigma} = \left \{ \omega: \abs{\Im \omega} \leq \sigma \right \}$, is bounded by,
\begin{align}
   e_{\Dy}[f] \leq \frac{2  \sup_{s \in [-\sigma, \sigma]} \int_{-\infty}^{\infty} dy \, \lvert f(y + i s) \rvert }{ e^{2\pi \sigma / \Dy} - 1 },
   \label{eq:GLT_error} 
\end{align}
where $e_{\Dy}[f]$ gives the approximation error of the exact $f$-integral over $(-\infty, \infty)$ with a discretization step $\Dy$. 

Building upon \cref{thm:1}, we quantify the resource efficiency of the combined Legendre ($q$-integral) and trapezoidal ($y$-integral) rule for approximating~\cref{eq:resolvent_II} in the following theorem. 

\begin{thm}
\label{thm:2}
For a real pole $z_k = a_k \in \mathbb{R}$ with $a_k \notin [E_{0}, E_{N-1}]$ and tolerance $\epsilon > 0$, the resolvent $R(z_k)$ admits a real-time LCHS construction $\mathcal{I}_{J}(z_k)$ for which $\lVert R(z_k) - \mathcal{I}_{J}(z_k) \rVert_2 < \epsilon$. Let the truncation parameter $y^{\rm max}$ be chosen such that the $y$-integrand is periodic up to a $\delta$-perturbation with $\mathcal{O}(\delta)<\epsilon$.
Then the combination of the Gauss--Legendre (q-integral) and trapezoidal (y-integral) rules, with $L_q$ and $L_y$ nodes, respectively, defines a time grid $\bt$ such that, up to leading order,
\begin{align}
    \begin{split}
        J=L_q L_y = \bigg[ \log \frac{32 e^{\pi}}{ \sqrt{\pi} \epsilon^2 a_k^{+} a_k^{-}} + \frac{2\sqrt{2} a_k^{+} }{ a_k^{-} } \log \frac{4}{\epsilon a_k^{-}}  \bigg] \bigg[ \log_2 \frac{ \sqrt{2}}{\pi\epsilon a^{-}_k } + \frac{ a_k^{+} }{2 a_k^{-}} \log_2 \frac{4}{\epsilon a_k^{-}} + 7 \bigg]
    \end{split}
    \label{eq:GLT_J}
\end{align}
distinct time evolution circuits suffice to construct $\mathcal{I}_{J}(z_k)$,
where $a_k^{+} = \max_{0 \leq n \leq N-1} \lvert a_k -E_n \rvert$ and $a_k^{-} = \min_{0 \leq n \leq N-1} \lvert a_k -E_n \rvert$ for eigenvalues $E_n$ of $H$. Additionally, the maximal and total evolution time can also be controlled by
\begin{align}
    T^{\max} = \lVert \bt \rVert_{\infty} = \frac{2\sqrt{2}}{a_k^{-}} \log \frac{4}{\epsilon a_k^{-} }, \qquad \text{and} \qquad T^{\rm tot} = \lVert \bt \lVert_{1} =  \frac{ \sqrt{2} J}{2 a_k^{-}} \log \frac{4}{\epsilon a_k^{-}}.
    \label{eq:GLT_tmax}
\end{align}
\end{thm}
\begin{proof}
    The proof is provided in \cref{subsec:proof_theorem2}.
\end{proof}

Whereas for complex poles, the simulation time depends strongly on the complex component of the pole, for real poles it is the smallest distance $a^-_k$, between the real pole $z_k = a_k$ and an eigenvalue of $H$, which influences the simulation time.
In \cref{thm:2}, the condition number of the shifted operator, $\frac{a_k^{+}}{a_k^{-}} = {\rm cond}(z_k - H)$, is essential in bounding the quantum resources cost. We notice that the asymptotic cost of employing the Dirac-delta kernel and the Gaussian kernel in the discrete-time LCHS approach is comparable since $J = \mathcal{O}\big( \log {\rm poly} \frac{{\rm constant}}{\epsilon} \big)$ and $T^{\rm max} = \mathcal{O}\big(\log \frac{{\rm constant}}{\epsilon}\big)$ for both cases. Despite the same asymptotic dependence, the actual implementation cost can vary significantly between the kernels. Finally, we remark that the result for real poles is directly applicable to solving linear systems of equations, a routine fundamental in quantum linear algebra.

\subsection{Numerical experiments for real poles}
\label{subsec:Gaussian_discrete_numerics}
For the MFIM Hamiltonian in \cref{eq:MFIM}, we explore different discrete LCHS constructions for poles on the real line.
Again we consider the rescaled Hamiltonian with $(h, g) = (1, \frac{2}{3})$ and $L_{\rm sys} = 8$.

\paragraph{Convergence as a function of simulation time.} We compare the three quadrature rules for the $y$-integral (Hermite, Trapezoidal, and Legendre) combined with a Legendre discretization of the $q$-integral.
The first term in \cref{eq:GLT_J} provides an estimate for the number of Legendre nodes, $L_q\approx \log{\frac{32e^\pi}{\sqrt{\pi} \epsilon^2 a_k^+ a_k^-}} + \frac{2\sqrt{2} a_k^+}{a_k^-}\log{\frac{4}{\epsilon a_k^-}}$, required to reach a desired accuracy.
However, as we have noticed that it often overestimates the optimal choice, we empirically determine $L_q$ closer to the optimal and report these values in the figures.
The convergence of the Legendre, trapezoidal, and Hermite rule is shown in the top panel of \cref{fig:MFIM_convergence_double_3} for the real pole $z=-1.1$ and accuracy $\epsilon = 10^{-6}$.
While all three combinations exhibit an exponential rate of convergence, the onset of convergence occurs earliest for the trapezoidal rule, resulting in the most efficient approximation.
In the bottom panel of \cref{fig:MFIM_convergence_double_3}, we compare the current standard in literature \cite{ChKoSo17}, \textit{i.e.}, the trapezoidal rule for both integrals, to our proposed Legendre-trapezoidal combination which is more efficient.

\begin{figure}[!htb]
  \centering
    \setlength\figureheight{2cm}
    \setlength\figurewidth{0.75\textwidth}	
%
  \begin{tikzpicture}
			\begin{axis}[%
				width=0.476\figurewidth,
				height=\figureheight,
				scale only axis,
				xmode=log,
				xmin=100,
				xmax=1e+06,
				xminorticks=true,
				xlabel={$\frac{T^{\rm tot}}{T^{\rm max}}$},
				ymode=log,
				ymin=1e-09,
				ymax=100,
				yminorticks=true,
				ylabel={$\Vert R(z)-\mathcal{I_J}(z)\Vert_2$},
				]
				\addplot [color=myblue, draw=none, mark=o, mark options={solid, myblue, thick}, mark size=2.5, forget plot]
				table[row sep=crcr]{%
					248.05	3.6147\\
					756.91	4.5489\\
					1266.1	2.8164\\
					1775.4	4.5454\\
					2284.6	4.5547\\
					2793.9	3.357\\
					3303.2	3.3164\\
					3812.5	2.8055\\
					4321.8	2.5676\\
					4831.1	2.5603\\
					5340.4	2.3295\\
					5849.7	2.1604\\
					6359	1.9969\\
					6868.3	1.8519\\
					7377.5	1.7483\\
					7886.8	1.6391\\
					8396.1	1.5526\\
					8905.4	1.4016\\
					9414.7	1.3017\\
					9924	0.8784\\
					9.9494e+03  7.4509e-01\\
					9.9749e+03  6.1614e-01\\
					1.0000e+04  4.9711e-01\\
					1.0026e+04  3.9164e-01\\
					1.0051e+04  3.0152e-01\\
					1.0077e+04  2.2703e-01\\
					1.0102e+04  1.6731e-01\\
					1.0128e+04  1.2076e-01\\
					1.0153e+04  8.5419e-02\\
					1.0179e+04  5.9253e-02\\
					1.0204e+04  4.0330e-02\\
					1.0230e+04  2.6949e-02\\
					1.0255e+04  1.7688e-02\\
					1.0281e+04  1.1409e-02\\
					1.0306e+04  7.2350e-03\\
					1.0331e+04  4.5127e-03\\
					1.0357e+04  2.7696e-03\\
					1.0382e+04  1.6731e-03\\
					1.0408e+04  9.9531e-04\\
					10433	0.00058319\\
					1.0459e+04  3.3672e-04\\
					1.0484e+04  1.9160e-04\\
					1.0510e+04  1.0751e-04\\
					1.0535e+04  5.9469e-05\\
					1.0561e+04  3.2470e-05\\
					1.0586e+04  1.7472e-05\\
					1.0611e+04  9.2985e-06\\
					10637	4.8652e-06\\
					10662	2.5329e-06\\
					10688	1.2828e-06\\
					10713	6.6093e-07\\
					10739	3.1718e-07\\
					10764	1.6912e-07\\
					10790	6.9808e-08\\
					10815	4.6214e-08\\
					10841	2.2242e-08\\
					10866	2.2242e-08\\
					10892	2.2242e-08\\
					10917	2.2242e-08\\
					10943	2.2242e-08\\
					11452	2.2242e-08\\
					11961	2.2242e-08\\
					12470	2.2242e-08\\
				};
				
				\addplot [color=myorange, draw=none, mark=+, mark options={solid, myorange, thick}, mark size=2.5, forget plot]
				table[row sep=crcr]{%
					222.22	18.785\\
					620.69	3.2678\\
					1020.4	4.5454\\
					1420.3	3.7568\\
					1820.2	2.8117\\
					2220.2	2.546\\
					2620.2	2.1731\\
					3020.1	1.9296\\
					3420.1	1.7338\\
					3820.1	1.5551\\
					4220.1	1.4017\\
					4620.1	1.2712\\
					5020.1  0.97519\\
					5040    0.7585\\
					5060.1  0.5009\\
					5080    0.2711\\
					5100.1  0.1171\\
					5120.0  3.9643e-02\\
					5140.1  1.0380e-02\\
					5160.0  2.0839e-03\\
					5180.1  3.1872e-04\\
					5200.0  3.6981e-05\\
					5220.1  3.2345e-06\\
					5240.0  2.2200e-07\\
					5260.1  1.9033e-08\\
					5280.0  1.9030e-08\\
					5300.1  1.9028e-08\\
					5320.0  1.9026e-08\\
					5340.1  1.9024e-08\\
					5360.0  1.9022e-08\\
					5380.1  1.9020e-08\\
					5400.0  1.9019e-08\\
					5420.1  1.9017e-08\\
					5820.1  1.8986e-08\\
					6220.1  1.8964e-08\\
					6620.1  1.8948e-08\\
					7020.1  1.8935e-08\\
					7420.1  1.8926e-08\\
					7820.1  1.8918e-08\\
					8220	1.9077e-08\\
					8620	1.9381e-08\\
					9020	1.9643e-08\\
					9420	1.987e-08\\
					9820	2.0068e-08\\
				};
				
				\addplot [color=myred, draw=none, mark=asterisk, mark options={solid, myred, thick}, mark size=2.5, forget plot]
				table[row sep=crcr]{%
					1.7014e+05	1.6107\\
					1.963e+05	1.5314\\
					2.2367e+05	1.4754\\
					2.5221e+05	1.4197\\
					2.8187e+05	1.3592\\
					3.1261e+05	1.3249\\
					3.4439e+05	1.2711\\
					3.7719e+05	1.2332\\
					4.1096e+05	1.2094\\
					4.4568e+05	1.0344\\
					4.4921e+05	0.96437\\
					4.5274e+05	0.87984\\
					4.5628e+05	0.78298\\
					4.5983e+05	0.6776\\
					4.6339e+05	0.56868\\
					4.6696e+05	0.46173\\
					4.7054e+05	0.36192\\
					4.7413e+05	0.27338\\
					4.7773e+05	0.19869\\
					4.8133e+05	0.13877\\
					4.8495e+05	0.093046\\
					4.8857e+05	0.059836\\
					4.922e+05	0.036881\\
					4.9585e+05	0.021776\\
					4.995e+05	0.01231\\
					5.0316e+05	0.006661\\
					5.0682e+05	0.0034487\\
					5.105e+05	0.0017081\\
					5.1419e+05	0.00080925\\
					5.1788e+05	0.00036667\\
					5.2159e+05	0.00015888\\
					5.253e+05	6.5837e-05\\
					5.2902e+05	2.6089e-05\\
					5.3275e+05	9.8864e-06\\
					5.3649e+05	3.5829e-06\\
					5.4024e+05	1.2418e-06\\
					5.44e+05	9.8106e-07\\
					5.4776e+05	9.8106e-07\\
					5.5154e+05	9.8106e-07\\
					5.5532e+05	9.8106e-07\\
					5.9361e+05	9.8106e-07\\
				};
        \draw [dashed,thick,black] (0,1e-6) -- (1e+6,1e-6);
			\end{axis}
    				\setlength{\tabcolsep}{5pt}
				\renewcommand\arraystretch{1.7}
				  \node (somenode) [shape=rectangle] at (9.2,1.2) {
			{\footnotesize \begin{tabular}{l|llll}		
				\quad $q/y$-int & $q^{\rm max}$                                         & $L_q$ & $y^{\rm max}$                             \\ \hline
			{\color{myblue}$\circ$} Leg/Leg     & $\frac{2}{a^-_k}\sqrt{\log\frac{2}{\epsilon a_k^-}}$   & 80    & $\sqrt{2\log \frac{2}{\epsilon a^-_k}}$  \\
			{\color{myorange}$+$} Leg/Trap    & $\frac{2}{a^-_k}\sqrt{\log\frac{2}{\epsilon a_k^-}}$   & 80    & $\sqrt{2\log \frac{2}{\epsilon a^-_k}}$  \\
			{\textcolor{myred}{$\ast$}} Leg/Her      & $\frac{2.6}{a^-_k}\sqrt{\log\frac{2}{\epsilon a_k^-}}$ & 80    & $[-\infty,\infty]$                  
			\end{tabular}}
		}; 
		\end{tikzpicture}

		\begin{tikzpicture}
			\begin{axis}[%
width=0.476\figurewidth,
height=\figureheight,
scale only axis,
xmode=log,
xmin=100,
xmax=10000,
xminorticks=true,
xlabel={$\frac{T^{\rm tot}}{T^{\rm max}}$},
ylabel={$\Vert R(z)-\mathcal{I_J}(z)\Vert_2$},
ymode=log,
ymin=1e-9,
ymax=100,
yminorticks=true,
]
\addplot [color=mypurple, draw=none, mark=triangle, mark options={solid, mypurple, thick}, mark size=2.5, forget plot]
  table[row sep=crcr]{%
135.5	32.111\\
474.26	4.8485\\
813.01	2.3701\\
1151.8	1.4407\\
1490.5	0.95387\\
1829.3	0.65429\\
2168	0.45149\\
2506.8	0.30625\\
2845.5	0.20058\\
3184.3	0.12547\\
3523	0.074579\\
3861.8	0.042048\\
4200.6	0.022475\\
4539.3	0.011387\\
4878.1	0.0054684\\
5216.8	0.0024893\\
5555.6	0.0010741\\
5894.3	0.00043927\\
6233.1	0.00017029\\
6571.8	6.2573e-05\\
6910.6	2.1793e-05\\
7249.3	7.1939e-06\\
7588.1	2.2498e-06\\
7926.9	6.6501e-07\\
8265.6	1.8182e-07\\
8604.4	2.3614e-08\\
8943.1	5.2019e-09\\
9281.9	5.2063e-09\\
9620.6	5.2101e-09\\
9959.4	5.2136e-09\\
};

\addplot [color=myorange, draw=none, mark=+, mark options={solid, myorange, thick}, mark size=2.5, forget plot]
  table[row sep=crcr]{%
135.5	2.7115\\
474.26	0.66235\\
813.01	0.18938\\
1151.8	0.046483\\
1490.5	0.019945\\
1829.3	0.0038781\\
2168	0.00064959\\
2506.8	0.00010486\\
2845.5	1.9592e-05\\
3184.3	4.7607e-06\\
3523	8.9114e-07\\
3861.8	1.3702e-07\\
4200.6	2.6245e-08\\
4539.3	1.9267e-08\\
4878.1	1.8112e-08\\
5216.8	1.9694e-08\\
5555.6	2.1544e-08\\
5894.3	2.0418e-08\\
6233.1	2.4598e-08\\
6571.8	2.1891e-08\\
6910.6	1.422e-08\\
7249.3	1.8205e-08\\
7588.1	1.7068e-08\\
7926.9	1.1647e-08\\
8265.6	1.3353e-08\\
8604.4	1.2811e-08\\
8943.1	1.1111e-08\\
9281.9	1.0123e-08\\
9620.6	1.0255e-08\\
9959.4	8.9718e-09\\
};
\draw [dashed,thick,black] (0,1e-6) -- (1e+4,1e-6);
\end{axis}
    				\setlength{\tabcolsep}{5pt}
				\renewcommand\arraystretch{1.7}
				  \node (somenode) [shape=rectangle] at (9.3,1.6) {
			{\footnotesize \begin{tabular}{l|llll}		
				\quad $q/y$-int & $q^{\rm max}$                                          & $y^{\rm max}$    & $L_y$                          \\ \hline
				{\color{mypurple} $\triangle$} Trap/Trap     & $\frac{2}{a^-_k}\sqrt{\log\frac{2}{\epsilon a_k^-}}$      & $\sqrt{2\log \frac{2}{\epsilon a^-_k}}$ & 270 \\
				{\color{myorange}$+$} Leg/Trap    & $\frac{2}{a^-_k}\sqrt{\log\frac{2}{\epsilon a_k^-}}$      & $\sqrt{2\log \frac{2}{\epsilon a^-_k}}$ & 270        
			\end{tabular}}
		}; 
		\end{tikzpicture}
    \caption{Error of quadrature approximations to $R(z_k)$ for the MFIM Hamiltonian with $L_{\rm sys}=8$ spins, a requested accuracy $\epsilon=10^{-6}$, and $z_k = -1.1$, in terms of the cost metric $\frac{T^{\rm tot}}{T^{\rm max}}$ for scaled total evolution time. The table shows the parameter settings and quadrature rules used for the outer $q$-integral and inner $y$-integral. \textbf{Top:} Number of nodes $L_y$ to discretize the $y$-integral is varied. \textbf{Bottom:} Number of nodes $L_q$ to discretize the $q$-integral is varied. Note: $a_k^-=0.1$.
    }
    \label{fig:MFIM_convergence_double_3}
\end{figure}
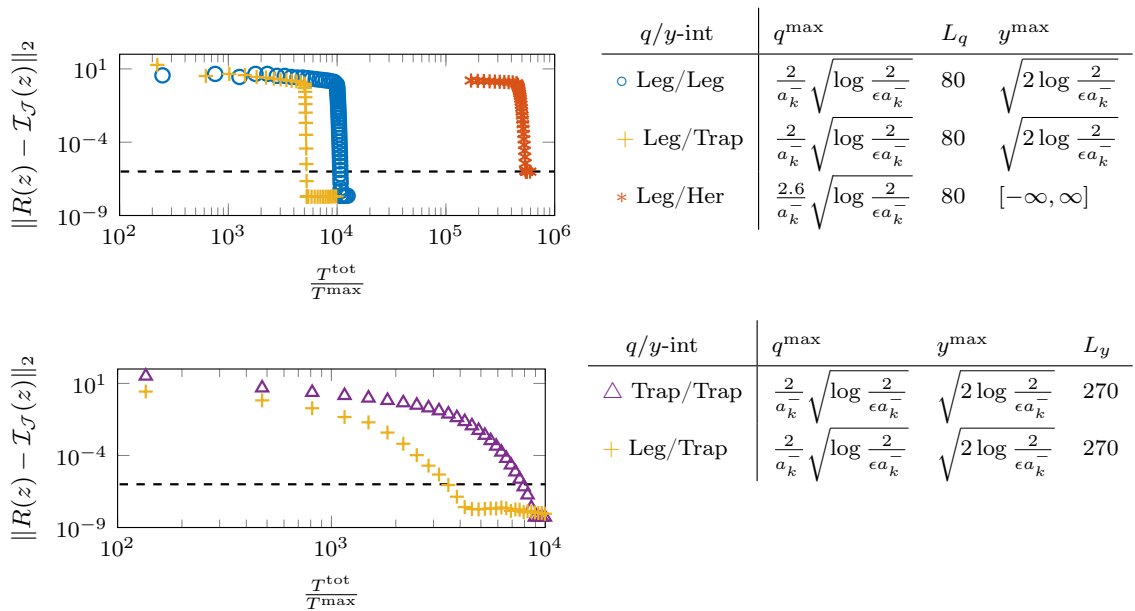

\paragraph{Simulation time as a function of pole position.} \cref{eq:GLT_J} implies that the cost of approximating $R(z_k)$ increases as $z_k$ moves closer to an eigenvalue.
In \cref{fig:cost_double} we report the cost $\frac{T^{\rm tot}}{T^{\rm max}_k}$ required to approximate $R(z_k)$ up to accuracy $\epsilon = 10^{-6}$ as $z_k$ approaches $E_0$.
The Hermite rule is significantly more expensive, and the trapezoidal rule is the most efficient. The estimate for $L_q$ suggests that for $z_k=-1-2^{-k-1}$ ($a_k^- = 2^{-k-1}$) the number of nodes should increase approximately by a multiplicative factor $2$ if $k$ increases by 1. However, the figure shows that a linearly increasing number of nodes is sufficient. The bound provided in \cref{eq:GLT_tmax}, based on the $L_q$ reported next to the figure, is shown as the solid line in \cref{fig:cost_double}. It largely overestimates the observed cost and seems to exhibit different asymptotic behavior.
The difference between the observed and predicted asymptotic behavior 
might be explained by the use of loose (pessimistic) bounds in the proof of Theorem 2. A more careful analysis of the errors could result in tighter bounds and a more accurate prediction of the asymptotic behavior, an important direction for future work.

\begin{figure}[!htb]
    \centering
    \setlength\figureheight{2cm}
    \setlength\figurewidth{0.75\textwidth}	
    		\begin{tikzpicture}

\begin{axis}[%
width=0.476\figurewidth,
height=\figureheight,
scale only axis,
xmin=0,
xmax=0.6,
xlabel=$a^{-}$,
ymode=log,
ymin=1,
ymax=1e9,
yminorticks=true,
ylabel=$\frac{T^{\rm tot}}{T^{\rm max}}$,
axis background/.style={fill=white},
]
\addplot [color=myblue , draw=none, mark=o, mark options={solid, myblue, thick}, mark size=2.5, forget plot]
  table[row sep=crcr]{%
0.5 1019\\
0.25 2923.9\\
0.125 7591.6\\
0.0625 19014\\
0.03125 46559\\
};

\addplot [color=myorange, draw=none, mark=+, mark options={solid, myorange, thick}, mark size=2.5, forget plot]
  table[row sep=crcr]{%
0.5 502.6\\
0.25 1425\\
0.125 3710\\
0.0625 9315\\
0.03125 22908\\
};
\addplot [color=myorange, forget plot]
  table[row sep=crcr]{%
0.5 3158000\\
0.25 8337000\\
0.125 26588000\\
0.0625 96631000\\
0.03125 789980000\\
};
\addplot [color=myred, draw=none, mark=asterisk, mark options={solid, myred, thick}, mark size=2.5, forget plot]
  table[row sep=crcr]{%
0.5 30671\\
0.25 242960\\
0.125 2045700 \\
};

\end{axis}

    				\setlength{\tabcolsep}{4pt}
				\renewcommand\arraystretch{1.7}
				  \node (somenode) [shape=rectangle] at (9.5,1.2) {
			{\footnotesize \begin{tabular}{l|llll}		
				\quad $q/y$-int & $q^{\rm max}$                                         & $L_q$ & $y^{\rm max}$                             \\ \hline
			{\color{myblue}$\circ$} Leg/Leg     & $\frac{2}{a^-_k}\sqrt{\log\frac{2}{\epsilon a_k^-}}$   & $30+20k$    & $\sqrt{2\log \frac{2}{\epsilon a^-_k}}$  \\
			{\color{myorange}$+$} Leg/Trap    & $\frac{2}{a^-_k}\sqrt{\log\frac{2}{\epsilon a_k^-}}$   & $30+20k$    & $\sqrt{2\log \frac{2}{\epsilon a^-_k}}$  \\
			{\textcolor{myred}{$\ast$}} Leg/Her      & $\frac{4}{a^-_k}\sqrt{\log\frac{2}{\epsilon a_k^-}}$ & $45+30k$    & $[-\infty,\infty]$                  
			\end{tabular}}
		}; 
		\end{tikzpicture}
    \caption{Cost $\frac{T^{\rm tot}}{T^{\rm max}}$ for quadrature rules necessary to compute an $\epsilon$-approximation to $(z-H)^{-1}$ with accuracy $\epsilon=10^{-6}$.
    We consider $z= -1-2^{-k-1}$ for $k=0,1,\dots, 4$. The solid line shows the cost predicted by \cref{thm:2}  for the combination of Leg/Trap rule.
    }
    \label{fig:cost_double}
\end{figure}
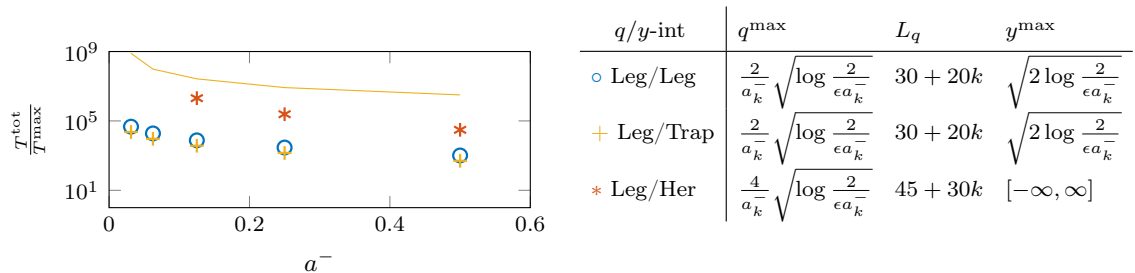

\section{Continuous-time LCHS construction of a resolvent}
\label{sec:resolvent_continuous}
In general, the resources required for 
Hamiltonian simulation must scale at least linearly with the simulation time, due to the no fast-forwarding theorem~\cite{Berry2007,Childs2010}.
This can make long-time simulations 
exceedingly costly, thus limiting the applicability of the discrete-time LCHS approach elaborated in \cref{sec:resolvent_discrete}.
In pursuit of a complementary strategy, we consider constructing the resolvent through the \textit{continuous-time} LCHS approach.
This allows us to control the cost of Hamiltonian simulations by the introduction of continuous-variable ancillae~\cite{LCU2023}.

The main idea behind the continuous-time approach is to associate the integral form of \cref{eq:resolvent_unitary} with some spatially extended ancillary state, \textit{e.g.}, a continuous-variable wavefunction describing a harmonic oscillator. The ancilla is coupled to the system via a total Hamiltonian $\Tilde{H} = H_{\rm sys} \otimes V$, where the system Hamiltonian $H_{\rm sys}(H;z_k)$ contains the operator of interest $H$, and the ancillary potential $V(\hat{q},\hat{y};z_k)$ depends on two (commuting) position operators $\hat{q}$ and $\hat{y}$. This approach requires hybrid quantum information processing which utilizes both discrete qubits and continuous Gaussian states~\cite{Andersen2015,liu2024hybrid}.
\cref{fig:continuous_time} illustrates the continuous-time LCHS procedure for a complex pole $z\notin \mathbb{R}$, which involves sampling a set of 1D Gaussian ancillary wavefunctions.
\begin{figure*}[tbh!]
    \centering
    \includegraphics[scale=0.425]{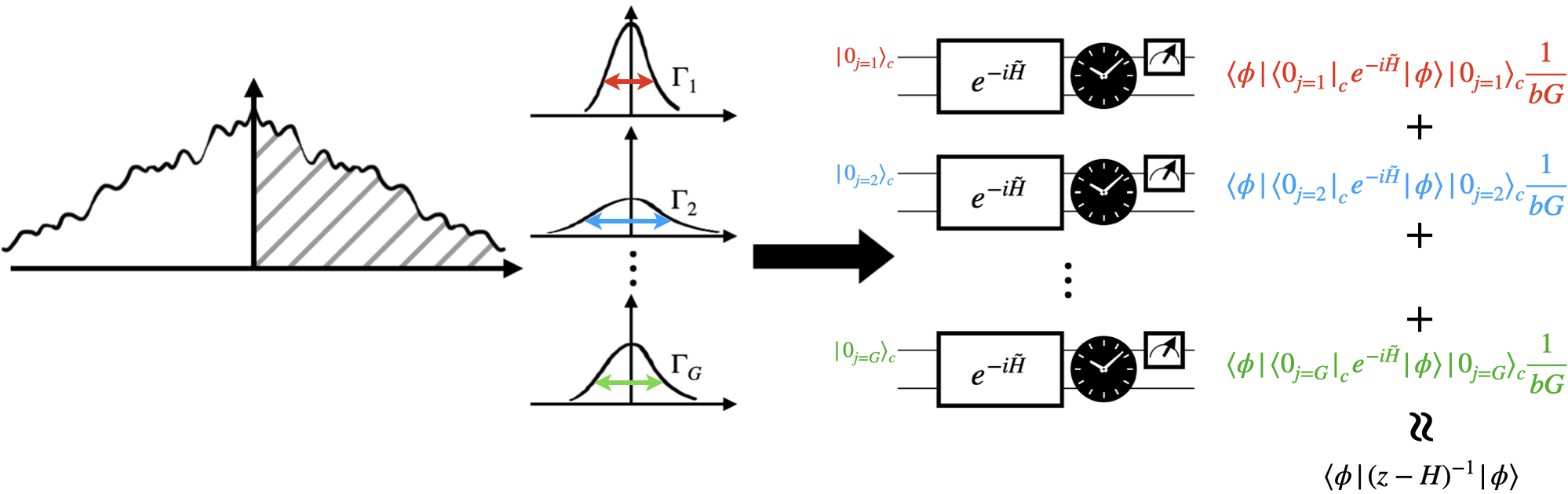}
    \caption{Continuous-time LCHS construction of the resolvent $R(z) = (z-H)^{-1}$ via the Gaussian representation. The integral of a decaying oscillatory function is approximated using a weighted combination of Gaussian integrals of the form $\frac{1}{\sqrt{2\pi\Gamma}} \int_{-\infty}^{\infty} dq \, e^{-q^2/2\Gamma - i(a-H)\abs{q}}$. On quantum hardware the Hamiltonian simulation $U = e^{-i\Tilde{H}}$ with a constant duration can be performed efficiently. The circuit diagram shows the composite evolution with ancilla $\ket{0_j}_c$ prepared in the Gaussian state. Repeated measurements in the computational basis enable us to access scalar quantities such as $\langle \phi \vert \bra{0_j}_{c} e^{- i\Tilde{H}} \vert \phi \rangle \ket{0_j}_{c}$, which can be collected and combined in post-processing to obtain $\langle \phi\vert R(z) \vert \phi \rangle$.}
\label{fig:continuous_time}
\end{figure*}

For the continuous-time LCHS framework, we introduce efficient strategies to approximate the integral transform of \cref{eq:resolvent_unitary}: \cref{subsec:Dirac_delta_continuous} uses the Dirac-delta kernel (complex poles) and \cref{subsec:Gaussian_continuous} uses the Gaussian kernel (real poles). \cref{subsec:Dirac_delta_continuous_numerics} and \cref{subsec:Gaussian_continuous_numerics} contain the corresponding numerical illustrations using the same MFIM Hamiltonian as considered in \cref{sec:resolvent_discrete}.

\subsection{Dirac-delta kernel for complex poles}
\label{subsec:Dirac_delta_continuous}

Formally, the continuous-time implementation requires the preparation of a \textit{single} ancilla,
\begin{align}
    \ket{0}_{c} = \int_{0}^{\infty} dq \, \psi(q;b_k) \ket{q}
    \label{eq:resolvent_ancilla}
\end{align}
where we can reinterpret the integration variable $q$ as the ancillary position degree of freedom and $\psi(q)$ as the associated ancillary wavefunction. Ideally, we aim to prepare $\psi(q;b_k) = \sqrt{b_k} e^{-b_{k}q/2}$, whose squared amplitude replicates the exponential decay in the resolvent integral representation of \cref{eq:resolvent_I}. 
We then pair the system qubits with this continuous-variable ancilla. Evolving under the Hamiltonian $\Tilde{H} = (H - a_k) \otimes q$, the composite state (derivation shown in \cref{app:cont_complex}) becomes,
\begin{align}
    e^{-i\Tilde{H}} \ket{\phi}\ket{0}_{c} = i b_k  R(z_k) \ket{\phi} \ket{0}_c + \ket{\perp}\label{eq:resolvent_continuousII},
\end{align}
where the unwanted component $\ket{\perp}$ can be eliminated by post-selecting the ancilla.
Thus, with a time evolution and post-selection we can implement the application of the resolvent on a state $\ket{\phi}$.
However, despite the simple analytical form of $\psi$, its accurate preparation on quantum hardware can be a nontrivial task, for example due to a discontinuity in the wavefunction at $q=0$.
To address this issue, we present an efficient single-ancilla approach for encoding the resolvent in this section.
The basic mathematical formulation of our approach is summarized in the following scheme and elaborated in \cref{app:cont_complex}.

\bigskip

{\centering

	\begin{tikzpicture}
		\node[draw = none,minimum height = 1cm, text centered,text width=3.5cm] (R) at (-1.5,1.3) {State preparation $\ket{0}_c$};	
		\node[draw = none,minimum height = 1cm, text centered,text width=4.1cm] (I) at (6,1.3) {$R(z_k)$};
		\draw[->] (R.east) --node[above,sloped]{Evolve} node[below,sloped]{Post-selection}  (I.west);	
		
		\node[draw = none,minimum height = 1cm, text centered,text width=4cm] (Resolvent) at (-1.5,0) {$\sqrt{ b_k } \int_{0}^{\infty} dq \, e^{-b_{k}q/2} \ket{q}$};	
		\node[draw = none,minimum height = 1cm, text centered,text width=4.1cm] (evolve) at (6,0) {$i b_k  R(z_k) \ket{\phi} \ket{0}_c$};	
		\draw[->] (Resolvent.east) --node[above,sloped]{$e^{i(H - a_k) \otimes \hat{q}}$}  (evolve.west);	
		
		\node[draw = none,minimum height = 1cm, text centered,text width=6cm] (symm) at (-1.5,-2.3) {$\frac{1}{b_k} \int_{-\infty}^{\infty} dq \, \int_{0}^{\infty} d\Gamma \, \rho_{\rm mix}(\Gamma) \varphi_{g}(q;\Gamma) e^{i (a_k - H) \lvert q \rvert}$};
		\draw[->] (Resolvent.south) --node[mid left]{Gaussian}node[mid right]{representation} (symm.north);	
		
		\node[draw = none,minimum height = 1cm, text centered,text width=5cm] (MC) at (-2,-4.6) {$\frac{1}{b_k G} \sum_{j=1}^{G} \int_{-\infty}^{\infty} dq \, \sqrt{\varphi_{g}(q;\Gamma_{j})} \ket{q}\qquad$};
		\draw[->] (symm.south) --node[mid left]{Monte Carlo}node[mid right]{on $\rho_{\rm mix}(\Gamma)$} (-1.5,-4.1);

		\node[draw = none,minimum height = 1cm, text centered,text width=7cm] (ev) at (6.7,-4.6) {$\underbrace{\frac{1}{b_k G} \sum_{j=1}^{G} \int_{-\infty}^{\infty} dq \, \varphi_{g}(q;\Gamma_{j}) e^{i(a_k - H) \lvert q \rvert }}_{\approx i R(z_k)}\ket{\phi}\ket{0}_c$};	
		\draw[->] (MC.east) --node[above,sloped]{$e^{i(H - a_k) \otimes \lvert \hat{q} \rvert}$}  (ev.west);	

		\draw[black,thick] ($(R.north west)+(-1.2,0)$)  rectangle ($(ev.south east)+(-0.1,-0)$);
	\end{tikzpicture}

}
\smallskip
Our proposed single-ancilla approach essentially relies on a Gaussian approximation of $R(z_k)$.
The approximation expands the decaying term $e^{-b_k q}$ using a smooth Gaussian basis set,
\begin{equation}
    R(z_k) = i\int_{0}^{\infty} dq \, e^{-b_k q + i(a_k - H)q} = \frac{i}{b_k} \int_{-\infty}^{\infty} dq \, \int_{0}^{\infty} d\Gamma \, \rho_{\rm mix}(\Gamma|b_k) \varphi_{g}(q;\Gamma) e^{i (a_k - H) \lvert q \rvert}. \label{eq:resolvent_gaussian}
\end{equation}
The Gaussian wavefunction of spatial variance $\Gamma$ centered at $q=0$, 
\begin{align}
    \varphi_{g}(q;\Gamma) = \frac{1}{\sqrt{2\pi\Gamma}} e^{-q^2/2\Gamma}, \label{eq:GaussianWavefunction}
\end{align}
allows us to use ancillae prepared in the state
\begin{align}
   \ket{0}_{c} = \int_{-\infty}^{\infty} dq \, \sqrt{\varphi_{g}(q;\Gamma)} \ket{q}.
   \label{eq:ket_c}
\end{align}
Being the ground state of a harmonic oscillator, such a Gaussian state constitutes one of the most accessible resources in continuous-variable computing and is easier to prepare compared to the wavefunction $\psi(q;b_k)$ introduced in \cref{eq:resolvent_ancilla}.

For the approximation of the integral in \cref{eq:resolvent_gaussian}, we note that the mixing density
\begin{align}
    \rho_{\rm mix}(\Gamma|b_k) = \frac{b_k^2 }{2} e^{- b_k^2 \Gamma/2}.\label{eq:mixing_density}
\end{align} 
is a probability measure.
Therefore, \cref{eq:resolvent_gaussian} can be reinterpreted in terms of the expectation $\mathbb{E}_{\rm mix}$, i.e., an average with respect to $\rho_{\rm mix}$, and can be approximated with Monte-Carlo sampling
\begin{align}
    iR(z_k) &= \frac{1}{b_k} \mathbb{E}_{\rm mix} \bigg[ \int_{-\infty}^{\infty} dq \, \varphi_g(q;\Gamma) e^{i(a_k - H) \lvert q \rvert} \bigg], \\
    &\approx\frac{1}{b_k G} \sum_{j=1}^{G} \underbrace{\int_{-\infty}^{\infty} dq \, \varphi_{g}(q;\Gamma_{j}) e^{i(a_k - H) \lvert q \rvert}}_{:= A_j} =: \mathcal{I}_G(z),\label{eq:Dirac_delta_stochastic_MC}
\end{align}
where we call $A_j$ the stochastic resolvent.
This interpretation of the resolvent as an expectation over Gaussian functions is the key to our continuous approach. 
It allows us to state the following lemma, which shows that Gaussian wavefunctions can indeed be used to prepare the action of the resolvent on any state.

\begin{lem}\label{lem:continuous_Gaussian}
    Consider the ancilla states $\ket{0_j}_c = \int_{-\infty}^\infty dq\, \sqrt{\varphi_g(q;\Gamma_j)}\ket{q}$, $j=1,\dots, G$, where $\Gamma_j$ is sampled from the probability measure in \cref{eq:mixing_density}, $\rho_{\rm mix}(\Gamma; b_k)$, and $\varphi_g(q;\Gamma_j)$ is the Gaussian wavefunction in \cref{eq:GaussianWavefunction}.
    The sum of time evolutions of unit duration under the total Hamiltonian $\tilde{H} = (H-a_k)\otimes\vert q \vert$ applied to composite states $\ket{\phi}\ket{0_j}_c$,
    \begin{equation}
        \frac{1}{b_k G}\sum_{j=1}^G e^{-i\tilde{H}}\ket{\phi}\ket{0_j}_c,
    \end{equation}
    prepares the state $iR(z_k)/\Vert R(z_k)\ket{\phi}$.
    The approximation error decays as $\frac{1}{\sqrt{G}}$ and, for each $j$, the state $e^{-i\tilde{H}}\ket{\phi}\ket{0_j}_c$ is prepared with success probability 
    \begin{equation}
        P_{A_j} = \Vert A_j \ket{\phi} \Vert^2_2\sum_{n=0}^{N-1} p_n e^{- \Gamma_{j} E_n^2} \bigg[ 1 + {\rm erfi}^2  \bigg( \sqrt{\frac{\Gamma_{j} }{2}} E_n \bigg) \bigg],
        \label{eq:continuous_single_sucess_prob}
    \end{equation}
    where ${\rm erfi}(\cdot)$ denotes the imaginary error function, $p_n = \lvert \braket{E_n|\phi} \rvert^2$ is the squared overlap between the $n$th eigenstate $\ket{E_n}$ of $H$ and the system state $\ket{\phi}$.
\end{lem}
\begin{proof}
    See \cref{app:cont_complex}.
\end{proof}

By \cref{{lem:continuous_Gaussian}}, the success probability depends on the initial state $\ket{\phi}$, the Hamiltonian spectrum ${\rm spec}(H) := (E_0, E_1, \cdots, E_{N-1})$, and the pole position $z_k = a_k+ib_k$ by the appearance of $b_k$ in $\rho_{\rm{mix}}$.

Put differently, we have shown that the real-time evolution $e^{-i\Tilde{H}}$ of \textit{unit} duration can be viewed algorithmically as an efficient $(1,\mathbf{1},0)$-BE of the stochastic resolvent $A_{j}$, except that $(i)$ the ancilla is bosonic (we will use bold number to highlight this difference) and $(ii)$ its initialization depends on $A_{j}$. This approach avoids the need of discretizing the integral of \cref{eq:resolvent_I} in time. Moreover, the maximal runtime of the Hamiltonian simulations, a major expense in the discrete-time LCHS, becomes independent of the Hamiltonian structure or pole location. In exchange, the cost of ancilla initialization now reflects the dependence on the pole as suggested by the dependence of $\rho_{\rm mix}$ on $b_k$. Here, preparing a broader Gaussian wavefunction, which results from a pole located closer to the real line, incurs a higher expense. This is due to the greater complexity in coherently loading and manipulating the wavefunction across a larger spatial region.

With such continuous-variable BE, we may formally implement the SEL and PREP oracles,
\begin{align}
    U^{\rm SEL} = -i \sum_{j=1}^{G} \ket{j} \bra{j} \otimes \big( e^{-i\tilde{H}} U_{c, j} \big), \qquad  U^{\rm PREP} \ket{0^{\log_2 G}} = \frac{1}{\sqrt{ G }} \sum_{j=1}^{G} \ket{j},
\end{align}
where $U_{c,j}$ is a circuit capable of preparing the Gaussian ancilla $\ket{0_j}_{\rm c}$. For example, we can define $U_{c,j} \ket{e}_{c} = \ket{0_{j}}_{c}$ for a fixed ancillary excited state $\ket{e}_{c}$ that can decay to the relevant ground state via tunable stimulated emission. Therefore, we obtain a $(\frac{1}{b_k}, \mathbf{1}+\log_{2} G,\epsilon_{G})$-BE of the resolvent in \cref{eq:Dirac_delta_stochastic_MC} with $\log_{2} G$ regular ancillae and an error $\epsilon_{G}$ from stochastic sampling.

In addition to our unbiased estimator from \cref{eq:Dirac_delta_stochastic_MC}, we also consider a biased estimator as a complementary benchmark. While the unbiased estimator represents the resolvent as a stochastic Gaussian mixture and samples the Gaussian widths from a continuous distribution, this alternative replaces the random sampling with a deterministic approximation: a finite sum of equally weighted Gaussian wavefunctions with prescribed widths. Specifically, we approximate the resolvent via the following construction:
\begin{align}
    iR(z_k) \approx \frac{1}{b_k G} \sum_{j=1}^{G} \int_{-\infty}^{\infty} dq \, \varphi_{g}(q;\Gamma_{j}) e^{- i H \lvert q \rvert}, \quad \text{where } \Gamma_j = - \frac{2}{b_k^2} \log \frac{1+2(j-1)}{2G},
    \label{eq:resolvent_biased}
\end{align}
whose bias can be systematically suppressed by increasing the number of Gaussians $G$. We will evaluate its performance numerically alongside the unbiased estimator, highlighting the resulting bias-variance tradeoff in the next section.

Finally, we leave two remarks. First, the effective Hamiltonian $\Tilde{H} = H \otimes \abs{\hat{q}}$ in the continuous-time LCHS contains a non-differentiable ancillary potential $\abs{\hat{q}}$. For the potential to be physically implemented on analog platforms, we may consider its smooth modifications, \textit{e.g.}, $\abs{\hat{q}} \mapsto \sqrt{\hat{q}^2 + q_0^2}$ for sufficiently small $q_0$. Second, instead of evolving the total Hamiltonian for a unit duration (as insisted in this work), we may in principle increase the simulation duration to reduce the ancilla-related cost. This can be seen from \cref{eq:ket_c}: a simulation $e^{-i\tilde{H}t}$ of duration $t$ can be associated with a rescaled mixing density $\rho_{\rm mix}(\Gamma|tb_k)$ through the change of variables $(q, \Gamma) \mapsto (tq, t^2 \Gamma)$. For $t>1$, the rescaling leads to a smaller sampled Gaussian variance on average, thus simplifying the ancilla preparation and enhancing the post-selection probability as indicated by $\rho_{\rm mix}$ and \cref{eq:continuous_single_sucess_prob}, respectively. Such possible trade-off between the simulation time and ancillary cost merits further investigation.

\subsection{Numerical experiment for complex poles}
\label{subsec:Dirac_delta_continuous_numerics}
For the MFIM Hamiltonian considered in \cref{subsec:Dirac_delta_discrete_numerics,subsec:Gaussian_discrete_numerics}, we explore the continuous LCHS strategy discussed above for complex poles. 

In the continuous-time scheme, the Gaussian variance $\Gamma$ plays a crucial role in determining the quantum cost, impacting both the ancilla preparation and the post-selection process. Specifically, the Gaussian variance follows an exponential distribution $\rho_{\rm mix}(\Gamma)$ whose moments depend on the distance of the pole to the real line. We first note that the unbiased estimator of \cref{eq:Dirac_delta_stochastic_MC}, composed of $G$ Gaussians, converges to the resolvent at a rate of $\mathcal{O}\left( \frac{1}{\sqrt{G}} \right)$ as guaranteed by the central limit theorem (CLT). Moreover, the fluctuation of the Gaussian variance, $\mathbb{E}_{\rm mix}[\Gamma^2] - \mathbb{E}_{\rm mix}[\Gamma]^2= \frac{4}{b_k^4}$ (c.f. \cref{eq:mixing_density}), can become unfavorably large if the pole is located near the real line. To achieve faster convergence with a better-controlled resource estimate, we may therefore replace the continuous Gaussian mixture by a suitable finite sum.

\begin{figure*}[hbt!]
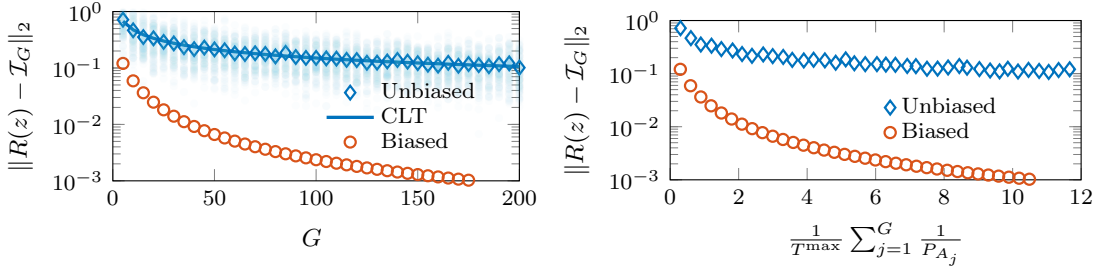

\centering
\plotconv[%
  width=0.48\textwidth,%
  xlabel={$G$ $\vphantom{\frac{1}{T^{\rm max}} \sum_{j=1}^{G} \frac{1}{P_{A_j}}}$},%
  xmax=200,%
  ylabel={$\lVert R(z) - \mathcal{I}_{G} \rVert_2$},
  ytick={1e-3, 1e-2, 1e-1 ,1e+0},%
  legend style={draw=none, font=\footnotesize, at={(0.5,0.375)},anchor=west, fill=none}
]{
\addplot[myColOne,only marks,thick,mark=diamond,mark size=2.5]%
  table[x index=0,y index=1] {\datfile{cLCU_single_G_unbiased_eps0.001}};
\addplot[myColOne,very thick,solid]%
  table[x index=0,y index=1] {\datfile{cLCU_single_G_CLT_eps0.001}};
\addplot[myColTwo,only marks,thick,solid,mark=o,mark size=2]%
  table[x index=0,y index=1] {\datfile{cLCU_single_G_biased_eps0.001}};
\addplot[myColSix,only marks,thick,solid,mark size=1,opacity=0.02]%
  table[x index=0,y index=1] {\datfile{testje}};
\legend{Unbiased,CLT,Biased};
}{0}{200}{1e-3}{1e+0}%
\hfill%
\plotconv[%
  width=0.48\textwidth,%
  xlabel={$ \frac{1}{T^{\rm max}} \sum_{j=1}^{G} \frac{1}{P_{A_j}}$},%
  xmax=12,%
  ylabel={},%
  ytick={1e-3, 1e-2, 1e-1 ,1e+0},%
  legend style={draw=none, font=\footnotesize, at={(0.5,0.375)},anchor=west, fill=none}
]{
\addplot[myColOne,only marks,thick,solid,mark=diamond,mark size=2.5]%
  table[x index=0,y index=1] {\datfile{cLCU_single_C_unbiased_eps0.001}};
\addplot[myColTwo,only marks,thick,solid,mark=o,mark size=2]%
  table[x index=0,y index=1] {\datfile{cLCU_single_C_biased_eps0.001}};
\legend{Unbiased,Biased};
}{0}{12}{1e-3}{1e+0}%
\caption{Resolvent approximation by sampling from a collection of Gaussian ancillary wavefunctions, for $z = -0.8+0.2i$ and the (scaled) MFIM Hamiltonian with $L_{\rm sys} =8$ spins. The approximation $\mathcal{I}_{G}(z)$ is computed using the unbiased ({\color{myblue}{$\diamond$}}) and biased (\textcolor{myred}{$\circ$}) estimator (\cref{eq:Dirac_delta_stochastic_MC} and \cref{eq:resolvent_biased} respectively) for a requested accuracy of $\epsilon = 10^{-3}$. To obtain the success probabilities, a reference state $\ket{\phi}$ with  $p_0 = 0.4$ and $p_n = \frac{1-p_0}{N-1}$ for $n \geq 1$ is considered. \textbf{Left:} Approximation error $\Vert R(z)-\mathcal{I}_{G}(z)\Vert_2$ as a function of the number of Gaussians $G$. For given $G$, the unbiased estimator is evaluated over a total of 100 trials (with each trial randomly drawing $G$ Gaussian variances from the mixing density $\rho_{\rm mix}$), while the biased estimator with prescribed Gaussian variances is constructed in a single trial. The error for the unbiased estimator is averaged over the 100 trials, with individual trials displayed in the background in light blue. The central limit scaling is represented as the solid line. \textbf{Right:} Approximation error as a function of the approximate quantum cost $\frac{1}{T^{\rm max}} \sum_{j=1}^{G} \frac{1}{P_{A_j}}$. The sum of the inverse probabilities is normalized by the maximal runtime $T^{\rm max}(\epsilon)$ in the discrete-time approach.}
\label{fig:MFIM_continuous_convergence_single}
\end{figure*}

To demonstrate the continuous-time LCHS approach, we perform the resolvent construction using the two kinds of Gaussian approximations introduced in \cref{subsec:Dirac_delta_continuous}, namely \cref{eq:Dirac_delta_stochastic_MC} and \cref{eq:resolvent_biased}. 


\paragraph{Convergence as a function of number of Gaussians.}
We evaluate the two Gaussian approximations for the complex pole $z = - 0.8 + 0.2 i$ in \cref{fig:MFIM_continuous_convergence_single}. The left panel shows the approximation error, $\lVert R(z) - \mathcal{I}_{G} \rVert_2$, as a function of the number of Gaussians $G$, where $\mathcal{I}_{G}$ denotes an approximation of the exact integral representation with $G$ Gaussians. We notice that the onset of convergence occurs significantly earlier for the biased estimator, resulting in an efficient continuous-time approximation.

In the right panel of \cref{fig:MFIM_continuous_convergence_single}, we quantify the actual quantum cost using the sum of inverse success probabilities, $\sum_{j=1}^{G} \frac{1}{P_{A_j}}$, which characterizes the total number of unit-duration Hamiltonian simulations attempted on average to realize a $G$-Gaussian approximation. This metric can be interpreted as the continuous-time equivalent of the total simulation time. For direct comparison with the discrete-time quadrature approach, we normalize the inverse probability metric by the maximal simulation time $T^{\rm max}$ (c.f. \cref{eq:GL_tmax}) in the discrete-time setting for $\epsilon = 10^{-3}$. Our empirical observation indicates that the biased Gaussian estimator requires a total simulation time similar to its discrete-time counterpart to achieve the requested accuracy, despite the higher cost of constructing an unbiased Gaussian estimator. Accordingly, the continuous-time approach may offer a distinct resource advantage due to the maximal simulation time always being unity. We comment that the success probabilities $P_{A_j}$ vary with the reference state $\ket{\phi}$ as reflected in \cref{eq:continuous_single_sucess_prob}. Although we have chosen a specific state $\ket{\phi}$ for illustration, the trend shown in \cref{fig:MFIM_continuous_convergence_single} remains consistent across different choices of reference state.

\subsection{Gaussian kernel for real poles}
\label{subsec:Gaussian_continuous}
For a real pole, a continuous-time compilation for the Gaussian kernel can be efficiently designed with 
\textit{two} continuous-variable ancillae. 
We consider preparation of the ancillae, 
\begin{align}
    \ket{0}_{c} = \int_{0}^{q^{\rm max}(2\epsilon)} dq \, \psi_1(q) \ket{q} \otimes \int_{-\infty}^{\infty} dy \, \psi_2(y) \ket{y},\label{eq:resolvent_ancilla_II}
\end{align}
in the uniform and Gaussian states $\psi_1(q) = \frac{1}{\sqrt{q^{\rm max} }}$ and $\psi_2(y) = \sqrt{\varphi_g(y;1)}$ (recall that $\varphi_{g}(y;\Gamma) = \frac{1}{\sqrt{2\pi\Gamma}} e^{-y^2/2\Gamma}$ as in  \cref{subsec:Dirac_delta_continuous}), with a finite $q$-truncation at some value $q^{\rm max}$
(see \cref{subsec:proof_theorem2} for the definition of $q^{\rm max}$). The two-ancillae setup follows similar ideas as for complex poles but does not require sampling of the Gaussian width; the general setup is outlined below.
\bigskip

{\centering

	\begin{tikzpicture}
		\node[draw = none,minimum height = 1cm, text centered,text width=3.5cm] (R) at (-1.5,1.3) {State preparation $\ket{0}_c$};	
		\node[draw = none,minimum height = 1cm, text centered,text width=5cm] (I) at (6.7,1.3) {$R(z_k)$};
		\draw[->] (R.east) --node[above,sloped]{Evolve} node[below,sloped]{Post selection} (I.west);	
		
		\node[draw = none,minimum height = 1cm, text centered,text width=5cm] (Resolvent) at (-2,0) {$\int_{0}^{q^{\rm max}} dq   \int_{-\infty}^{\infty} dy \, \sqrt{\frac{\varphi_g(y;1)}{q^{\rm max}} }\ket{q}\ket{y}$};	
		\node[draw = none,minimum height = 1cm, text centered,text width=5cm] (evolve) at (6.7,0) {$\frac{\sqrt{\pi}}{\sqrt{2}q^{\rm max}}  R_{\epsilon}(z_k) \ket{\phi} \ket{0}_c$};	
		\draw[->] (Resolvent.east) --node[above,sloped]{$e^{i(H - a_k) \otimes \hat{q} \otimes \hat{y}}$}  (evolve.west);	
		
		\node[draw = none,minimum height = 1cm, text centered,text width=7cm] (symm) at (-2,-3) {$\int_{-\infty}^{\infty} dq \, \int_{-\infty}^{\infty} dy \, \sqrt{\varphi_g(q;\Gamma_0 ) \varphi_g(y;1)} \ket{q} \ket{y}$};
		\draw[->] (Resolvent.south) --node[mid left]{Integral}node[mid right]{relation \cref{eq:resolvent_ancilla_III}} (symm.north);	
		
		\node[draw = none,minimum height = 1cm, text centered,text width=5cm] (ev) at (6.7,-3) {$\frac{ 1 }{ \sqrt{\Gamma_0 } } R_{\epsilon}(z_k) \ket{\phi} \ket{0}_c$};	
		\draw[->] (symm.east) --node[above,sloped]{$e^{i(H - a_k) \otimes \hat{q} \otimes \hat{y} }$}  (ev.west);	
		
		\draw[black,thick] ($(R.north west)+(-2.2,0)$)  rectangle ($(ev.south east)+(-0.1,0)$);
	\end{tikzpicture}

}
\bigskip

Driven under a total Hamiltonian $\Tilde{H} := (H - a_k) \otimes \hat{q} \otimes \hat{y}$, a real-time evolution of unit duration generates the composite state,
\begin{align}
    e^{-i\Tilde{H}} \ket{\phi}\ket{0}_{c}  &= \int_{0}^{q^{\rm max}} dq \, \int_{-\infty}^{\infty} dy \, \psi_1(q) \psi_2(y) e^{iy(a_k - H)q} \ket{\phi} \ket{q} \ket{y}, \\ 
    &= \frac{\sqrt{\pi}}{\sqrt{2}q^{\rm max}} R_{\epsilon}(z_k) \ket{\phi} \ket{0}_c + \ket{\perp} \label{eq:resolvent_II_continuousIII},
\end{align}
where the residual state $\ket{\perp}$ belongs to the kernel of the projector ${\rm Id} \otimes \ket{0}_c \bra{0}_c$ and, 
\begin{align}
    R_{\epsilon}(z_k) = \frac{1}{\pi} \int_{0}^{q^{\rm max}} dq \, \int_{-\infty}^{\infty} dy \, e^{-y^2/2} e^{iy(a_k - H)q} \approx R(z_k).
\end{align}
Therefore by post-selecting the $y$-ancilla, we can implement the action of the resolvent on any initial state $\ket{\phi}$, \textit{i.e.}, the normalized state $R_{\epsilon}(z_k) \ket{\phi} / \lVert R_{\epsilon}(z_k) \ket{\phi} \rVert_2$ that is $\mathcal{O}\left( \frac{\epsilon}{q^{\rm max}} \right)$-close to the target state. The success probability of a post-selection is,
\begin{align}
    P_{R(z_k)}^{\epsilon} &= \frac{\pi}{2 (q^{\rm max})^2} \lVert R_{\epsilon}(z_k) \ket{\phi}  \rVert_2^2 
    =  \frac{\pi (a_k^{-})^2}{8} \bigg[ \log \frac{1}{\epsilon a_k^{-} } \bigg]^{-1} \sum\limits_{n=0}^{N-1} \frac{p_n}{ \abs{ a_k - E_n }^2} ,
    \label{eq:continuous_LCU_success_II}
\end{align}
which varies with the location of the pole, $a_k$, the initial state $\ket{\phi}$, and the spectrum ${\rm spec}(H)$. The quantity $a^-_k$ is the smallest distance between the pole $a_k$ and any eigenvalue of $H$.

Although the uniform state $\psi_1(q)$ has a simple physical interpretation of a particle-on-a-ring wavefunction where the ring diameter is of $\mathcal{O}(q^{\rm max})$~\cite{LCU2023}, we note that an implementation using only Gaussian states can be adopted. Different from the stochastic implementation discussed in \cref{subsec:Dirac_delta_continuous}, here we construct a single, deterministic block-encoding of the resolvent by exploiting the following integral relation,
\begin{align}
        \int_{0}^{\infty} dq \, \int_{-\infty}^{\infty} dy \,  e^{-y^2/2 + iy(a_k - H)q} = \frac{1}{2} \int_{-\infty}^{\infty} dq \, e^{-q^2/2 \Gamma_0} \int_{-\infty}^{\infty} dy \, e^{-y^2/2 + iy(a_k - H)q} + \mathcal{O}(\pi \epsilon), 
\end{align}
where we define the variance $\Gamma_0 = \frac{1}{\epsilon (a_k^{-})^{3}}$ along the $q$ coordinate. That is, we initialize the ancillae in a product state of Gaussian wavefunctions in the variables $q$ and $y$,
\begin{align}
    \ket{0}_c = \int_{-\infty}^{\infty} dq \, \int_{-\infty}^{\infty} dy \, \sqrt{\varphi_g(q;\Gamma_0 ) \varphi_g(y;1)} \ket{q} \ket{y}.
    \label{eq:resolvent_ancilla_III}
\end{align}
Evolving this state under the total Hamiltonian $\tilde{H}$ yields an approximation to the resolvent,
\begin{align}
    e^{-i\tilde{H}} \ket{\phi}\ket{0}_c =  \frac{ 1 }{ \sqrt{\Gamma_0 }} R_{\epsilon}(z_k) \ket{\phi} \ket{0}_c + \ket{\perp},
\end{align}
where $\ket{\perp}$ labels the residual orthogonal to the Gaussian initial state. Post-selection on the $y$-ancilla implements the normalized state $R_{\epsilon}(z_k) \ket{\phi} / \lVert R_{\epsilon}(z_k) \ket{\phi} \rVert_2$ that is $\mathcal{O} \left( \frac{\epsilon}{\sqrt{\Gamma_0}} \right)$-close to the target state, therefore achieving an effective $(\sqrt{ \Gamma_0 }, \mathbf{2}, \epsilon)$-BE of the resolvent. This encoding comes with a success probability,
\begin{align}
    P_{R(z_k)}^{\epsilon} &= \frac{1}{ \Gamma_0 } \lVert R_{\epsilon}(z_k) \ket{\phi}  \rVert_2^2 
    =  \epsilon (a_k^{-})^3 \sum\limits_{n=0}^{N-1} \frac{p_n}{ \abs{ a_k - E_n }^2},
    \label{eq:continuous_LCU_success_III}
\end{align}
which scales less favorably in $\epsilon$ compared to \cref{eq:continuous_LCU_success_II}, despite the simpler ancillae preparation and loading.

\subsection{Numerical experiment for real poles}
\label{subsec:Gaussian_continuous_numerics}
We examine the continuous-time LCHS construction for real poles by initializing two continuous-variable ancillae. We compare two initialization scenarios: first, a combination of a uniform and Gaussian state as described by \cref{eq:resolvent_ancilla_II}, and alternatively, a purely Gaussian construction enabled by \cref{eq:resolvent_ancilla_III}.

\begin{figure*}[t!]
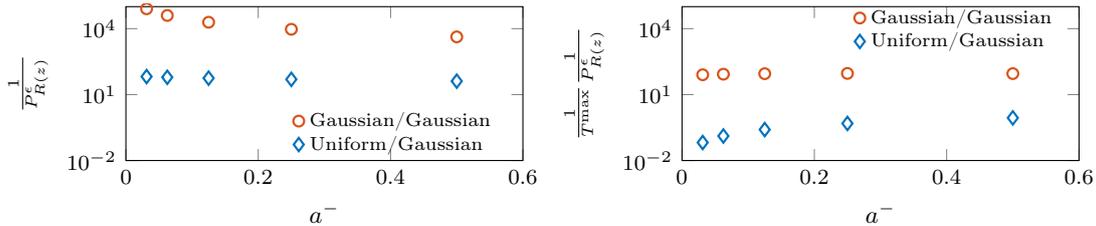

\centering
\plotconv[%
  width=0.48\textwidth,%
  xlabel={$a^{-}$},%
  xmin=0,
  xmax=0.5,%
  ymin=1e-2,
  ymax=1e+5,%
  ylabel={$\frac{1}{P^{\epsilon}_{R(z)}}$},%
  ylabel shift=-10pt,%
  legend style={draw=none, font=\scriptsize, at={(0.4,0.18)}, anchor=west, fill = none}
]{
\addplot[myColTwo,only marks,thick,solid,mark=o,mark size=2]%
  table[x index=0,y index=1] {\datfile{cLCU_double_C_gg_eps0.001}};
  \addplot[myColOne,only marks,thick,solid,mark=diamond,mark size=2.5]%
  table[x index=0,y index=1] {\datfile{cLCU_double_C_ug_eps0.001}};
\legend{Gaussian/Gaussian,Uniform/Gaussian};
}{0}{0.6}{1e+1}{1e+5}%
\hfill%
\plotconv[%
  width=0.48\textwidth,%
  xlabel={$a^{-}$},%
  xmin=0,
  xmax=0.5,%
  ymin=1e-2,
  ymax=1e+5,%
  ylabel={$\frac{1}{T^{\rm max}} \frac{1}{P^{\epsilon}_{R(z)}}$},%
  ylabel shift=-10pt,%
  legend style={draw=none, font=\scriptsize, at={(0.4125,0.85)}, anchor=west, fill=none}
]{
\addplot[myColTwo,only marks,thick,solid,mark=o,mark size=2]%
  table[x index=0,y index=1] {\datfile{cLCU_double_Cn_gg_eps0.001}};
  \addplot[myColOne,only marks,thick,solid,mark=diamond,mark size=2.5]%
  table[x index=0,y index=1] {\datfile{cLCU_double_Cn_ug_eps0.001}};
\legend{Gaussian/Gaussian,Uniform/Gaussian};
}{0}{0.6}{1e-2}{1e+2}
\caption{Cost necessary to construct an $\epsilon$-approximation to $(z-H)^{-1}$ with a requested accuracy $\epsilon = 10^{-3}$. We consider the MFIM Hamiltonian with $L_{\rm sys} =8$ spins and  $z = -1 -2^{-k-1}$ for $k=0,1,\ldots,4$. The approximation is computed using two continuous-variable ancillae prepared in the uniform/Gaussian ({\color{myblue}{$\diamond$}}) and Gaussian/Gaussian (\textcolor{myred}{$\circ$}) initial states (\cref{eq:resolvent_ancilla_II} and \cref{eq:resolvent_ancilla_III}) respectively. To obtain the success probability, a reference state $\ket{\phi}$ with $p_0 = 0.4$ and $p_n = \frac{1-p_0}{N-1}$ for $1 \leq n \leq N-1$ is chosen. \textbf{Left:} The quantum cost $\frac{1}{P^{\epsilon}_{R(z)}}$ as a function of the pole position $a^{-} = \min_n \abs{a-E_n} $. \textbf{Right:} The relative cost, normalized by the maximal runtime $T^{\rm max}(\epsilon;z)$ in the discrete-time approach, as a function of the pole position.}
\label{fig:MFIM_continuous_convergence_double}
\end{figure*}

We evaluate the two continuous-time approximations for the set of real poles $z_k = -1 -2^{-1-k}$ in \cref{fig:MFIM_continuous_convergence_double}. In particular, we report the quantum cost $\frac{1}{P^{\epsilon}_{R(z_k)}}$, understood as the number of unit-duration Hamiltonian simulations required on average to successfully realize an $\epsilon$-approximation. We consider a target accuracy of $\epsilon = 10^{-3}$ as in \cref{subsec:Dirac_delta_continuous_numerics}. Notice that an approximation using the uniform-Gaussian initialization from \cref{eq:resolvent_ancilla_II} is significantly less expensive than using the purely Gaussian initialization from \cref{eq:resolvent_ancilla_III}, although the latter is more practical to implement on hardware since Gaussian states are typically easier to manipulate. 
For the selected poles, we also assess the cost relative to the maximal runtime, $T^{\rm max}_{k}$, in the discrete-time setting. The normalized cost in the right pane in fact decreases as the pole approaches the Hamiltonian spectrum, exhibiting an opposite trend to that observed in \cref{fig:cost_double} of \cref{subsec:Gaussian_discrete_numerics}. This suggests a slower cost growth relative to the discrete-time approach and highlights the resource efficiency of the continuous-time approach.

\section{Construction of quantum rational transformations}
\label{sec:filter}
In \cref{sec:resolvent_discrete,sec:resolvent_continuous}, we have shown the construction of resolvents for varying pole locations and, additionally in \cref{sec:repeatedPoles}, the extension to multiple poles. 
These constructions allow for rational transformations on a quantum computer since any QRT can be written as a linear combination of resolvents, \textit{i.e.},
\begin{align}
    r(H;\bc,\bz,M) = \sum_{k=0}^{K-1} \sum_{m=1}^{M_k} c_{k,m} (z_k - H)^{-m} = \frac{ \sum\limits_{d=0}^{M-1} \beta_d H^d }{\prod\limits_{k=0}^{K-1} (z_k - H)^{M_k} }, \label{eq:rationalApproximant}
\end{align}
where $\bc$, extending the single resolvent setting, is now the vector of coefficients $c_{k,m}$ associated with the poles $z_k$ accounting for their multiplicities $M_k$, and $(\beta_0, \cdots, \beta_{M-1})$ defines polynomial coefficients in the numerator of the second equality, and $M = \sum_{k=0}^{K-1} M_k$ measures the total pole multiplicity.
The rational function in \cref{eq:rationalApproximant} is said to be of degree $(M-1,M)$.

Efficient QRTs opens the path to the development of new quantum rational algorithms.
Two emergent applications are spectral estimation and matrix function approximation.
The latter can be tackled, \textit{e.g.}, using a contour integration approach \cite{Higham} where the poles are placed on a customized contour. However, for some important functions, the poles for the optimal rational approximation are known and do not necessarily originate from a contour integral.
We illustrate our efficient QRT construction for the signum function ${\rm sgn}: [-1, 1] \rightarrow \{0,\pm 1\}$. Although the ${\rm sgn}(\omega)$ function appears simple, its discontinuity at $\omega = 0$ poses a challenge for approximation. Polynomials cannot capture the discontinuity as accurately, thus making rational functions the preferred approximants.

A QRT approximating the matrix signum function is useful for computing the singular value decomposition and symmetric eigenvalue decomposition, as demonstrated by recent classical algorithms  \cite{NaFr16}.
The ${\rm sgn}$ QRT can also be employed to construct a step function $\Theta(\omega) = \frac{1 - {\rm sgn}(\omega)}{2}$ for filtering out high energy eigenstates in ground and excited state problems. Instead of relying on an idealized step filter, in practice we seek a rational approximant  $r_{E}(\omega) \approx \Theta(\omega-E)$ so that,
\begin{align}
   \lvert r_{E}(E_n) \rvert \approx \begin{cases}
   1 & {\rm if}~ E_n \leq E - \Delta E,  \\
   0 & {\rm if}~ E_n \geq E + \Delta E,
   \end{cases}
   \label{eq:rE}
\end{align}
where $\Delta E$ characterizes the width of the buffer region over which $\lvert r_E \rvert$ transitions roughly from $1$ to $0$. The width $\Delta E$ can be chosen to balance accuracy of the filter and its construction cost, since a smaller buffer width requires a higher degree rational approximant.

The Zolotarev rational functions, known for their abilities to approximate ${\rm sgn}(\omega)$ with a tight $L^{\infty}$ error bound, are introduced in \cref{subsec:Zolotarev_approx}. 
While a Zolotarev approximant does not follow from the contour integration approach, it can be generated with our proposed methods.
We detail the quantum construction of the Zolotarev approximation through both discrete- and continuous-time LCHS approaches. 
As the degree of the Zolotarev approximation increases, the poles tend to cluster towards the real line, which requires extended quantum resources (see \cref{thm:1} and \cref{lem:continuous_Gaussian}). To mitigate the resource requirements, we take an iterative approach to construct a rational filter in \cref{subsec:iterative_filtering}, which is more efficient as we use poles further away from the real line.

\subsection{Zolotarev approximation to ${\rm sgn}$}
\label{subsec:Zolotarev_approx}
Zolotarev \cite{akhiezer1990elements} provides a concrete expression for the $L^{\infty}$-optimal rational approximant of fixed degree to the signum function.
For any given window $\mathcal{W}_{\overline{\omega}} = [- 1, -\overline{\omega}] \cup [\overline{\omega}, 1] \subseteq [-1, 1]$ with $\overline{\omega}>0$, the Zolotarev approximant, of degree $(2K-1,2K)$, is given by
\begin{align}
    r_K(\omega)  = - \gamma_K \sum_{k=0}^{K-1} \left( \frac{c_k}{i b_k - \omega } + \frac{c_k}{- i b_{k} - \omega}\right), \label{eq:rf_Zolo}
\end{align}
where the poles $\pm i b_{k}$, residuals $c_k$, and multiplicative constant $\gamma_K$ are explicitly known~\cite{li2020interior}.
The error of this approximant decays exponentially in the number of poles and can be bounded by
\begin{align}
    {\rm err}_K := \sup_{\omega \in \mathcal{W}_{ \overline{\omega} }} \lvert r_{K}(\omega) - {\rm sgn}(\omega)  \rvert  \leq 4 e^{- K \pi^2/2 \log(4/\overline{\omega})}.
    \label{eq:Zolo_errBound}
\end{align}
See \cref{app:Zolotarev} for more details.
In the following discussions, we examine the discrete-time and continuous-time LCHS approaches for constructing the Zolotarev approximant.

\subsubsection{Discrete-time LCHS}
To implement the Zolotarev approximant on the quantum computer, we observe that its action on a state $\ket{\phi}$ can be measured via the imaginary part of the inner product $\braket{\phi | A | \phi}$, 
\begin{align}
   \bra{\phi} r_{K}(H) \ket{\phi} = -i \left[ \braket{\phi | A | \phi} - \braket{\phi | A | \phi}^{\ast} \right] = 2\Im \braket{\phi | A | \phi},
\end{align}
where, for the Zolotarev poles $z_k = ib_k$ with $b_k\neq 0$, we get from \cref{eq:rf_Zolo} that
\begin{align}
    A = - \gamma_K \sum_{k=0}^{K-1} c_k R(z_k) = - \int_{0}^{\infty} dq \, c(q) e^{-iHq}, \qquad \text{for }c(q) = \gamma_K \sum_{k=0}^{K-1}  c_{k} e^{- b_k q}.\label{eq:Zolo_discrete_LCU}
\end{align}
Thus we obtain $r_K(H)$ through $A$, which takes the form of an LCHS. An efficient discrete-time LCHS, implementing the $K$ relevant resolvents simultaneously, can be constructed with the quadrature technique discussed in \cref{subsec:Dirac_delta_discrete}. This is summarized in the following corollary.
\bigskip

\begin{cor}
\label{cor:1}
For any given $\epsilon > 0$, the Zolotarev signum approximant $r_K(H)$ {\color{red}} admits a real-time LCHS construction $r_{K,J}(H)$ with $\lVert r_K(H) - r_{K,J}(H) \rVert_2 < \epsilon$. The Gauss-Legendre rule defines a time grid $\bt$ such that
\begin{align}
    J = \log_{2} \log \frac{4 \gamma_K }{ \epsilon b^{-} }  + (\eta + 1) \log_2 \frac{4 \gamma_K }{\epsilon b^{-} } +3,
    \label{eq:Zolo_J}
\end{align}
distinct time evolution circuits suffice to construct $r_{K,J}(H)$, where $\eta = \frac{3b^{-} - 2\sqrt{2} b^{-}  + \lVert H \rVert_2}{4\sqrt{2} b^{-}}$ and $b^{-} = \underset{0 \leq k \leq K-1}{\min} b_k$. The maximal and total evolution time, respectively, can be controlled by
\begin{align}
    T^{\rm max} = \frac{1}{b^{-}} \log \frac{4 \gamma_K}{\epsilon b^{-}}, \quad
    \text{and} \quad T^{\rm tot} = \frac{J}{2 b^{-}} \log \frac{4 \gamma_K}{\epsilon b^{-}}.
    \label{eq:Zolo_tmetrics}
\end{align}
\end{cor}
\begin{proof}
The proof is provided in \cref{subsec:proof_cor1}.
\end{proof}

We note that it suffices to consider the time grid $\bt$ associated with the pole closest to the real line, $z^{-}=ib^{-}$. This finest time grid can be reused for other poles, allowing all quantum measurements to be combined classically during the postprocessing step. However, $z^{-}$ approaches the real axis as the degree of the rational approximant increases \cite{akhiezer1990elements}. This could present a serious challenge for the discrete-time LCHS due to the $\frac{1}{b^{-}}$ factor appearing in the maximal and total simulation times (cf. \cref{eq:Zolo_tmetrics}).
The continuous-time LCHS may alleviate this challenge via importance sampling in a stochastic implementation.

\subsubsection{Continuous-time LCHS}
The continuous-time LCHS construction of a Zolotarev approximant can be achieved via the Monte-Carlo strategy introduced in \cref{subsec:Dirac_delta_continuous,subsec:Dirac_delta_continuous_numerics}. Since the coefficients $\{c_k\}_{k=0}^{K-1}$ from \cref{eq:rf_Zolo} satisfy the properties of a discrete probability measure, we acquire a Gaussian representation of $A$ that generalizes \cref{eq:Dirac_delta_stochastic_MC},
\begin{align}
    A &= - \gamma_K \mathbb{E}_{K} \bigg[ \frac{1}{b_k} \mathbb{E}_{\rm mix} \bigg[ \int_{-\infty}^{\infty} dq \, \varphi_g(q;\Gamma) e^{- i H \lvert q \rvert} \bigg] \bigg] \approx -\frac{\gamma_K}{N_K} \sum_{\alpha=1}^{N_K} \frac{1}{b_{k_{\alpha}} G_{k_\alpha}} \sum_{\beta=1}^{G_{k_{\alpha} }} A_{\alpha \beta}, 
    \label{eq:Zolo_expectation}
\end{align}
where $\mathbb{E}_K$ denotes an average over $[K]$ with respect to $\left\{ c_k \right\}_{k=0}^{K-1}$ and
\begin{align}
   A_{\alpha \beta} = \int_{-\infty}^{\infty} dq \, \varphi_{g}(q;\Gamma_{\alpha \beta}) e^{ -i  H \lvert q \rvert},
   \label{eq:Zolo_expectation2}
\end{align}
is a stochastic Zolotarev approximant that can be block-encoded by a continuous-variable ancilla. In 
\cref{eq:Zolo_expectation}, we sample Gaussian variance jointly from two probability distributions: $k_{\alpha}$ is sampled i.i.d.~from the marginal distribution $\left\{ c_k \right\}_{k=0}^{K-1}$ and $\Gamma_{\alpha \beta}$ i.i.d.~from the conditional distribution $\rho_{\rm mix}(\cdot|b_{k_{\alpha}})$ as defined in \cref{subsec:Dirac_delta_continuous}.

Recall that the cost of a continuous-time LCHS scheme is determined by the Gaussian variance $\Gamma$ as a random variable. The Gaussian variance follows a distribution, 
\begin{align}
     \rho_K(\Gamma) = \sum_{k=0}^{K-1} c_{k} \rho_{\rm mix}(\Gamma|b_k).
    \label{eq:full_var_distn}
\end{align}
Due to the intimate connection between the Zolotarev parameters $(c_{k}, b_k)$ and the elliptic functions and integrals~\cite{akhiezer1990elements}, we find that the average Gaussian width $\braket{\Gamma}_{K} := \int_{0}^{\infty} d\Gamma \, \rho_K(\Gamma) \Gamma$ remains independent of $K$. Since the Gaussian width sets the spatial extent of continuous-variable wavefunction and thereby the ancilla preparation cost, this independence allows for the use of higher-degree rationals. To control the variance $\braket{\Gamma^2}_K - \braket{\Gamma}_K^2$, a practical approach involves employing a finite convex sum of Gaussians rather than a continuous mixture, as discussed within \cref{subsec:Dirac_delta_continuous}.

To illustrate the utility of the continuous-time approach, we perform a QRT that constructs the Zolotarev approximant $r_K$ for given rational degree $2K$ and buffer region width $2\overline{\omega}$. We follow the recipe of \cref{eq:Zolo_expectation,eq:Zolo_expectation2} to generate a Monte-Carlo estimate $\hat{r}_K$ of $r_K$. To accelerate the convergence and control the sample variance, here we adopt the biased resolvent estimator using the Gaussian wavefunctions $\varphi_{g}(q;\Gamma_{j})$ with $\Gamma_j = - \frac{2}{b_k^2} \log \frac{1+2(j-1)}{2G_k}$ as in \cref{eq:resolvent_biased}. For a given number $N_{\rm MC}$ of Monte-Carlo attempts, we then sample the Gaussian variance using two distributions, a marginal corresponding to the Zolotarev coefficients $\left\{c_k \right \}_{k}$ and a uniform conditional defined by \cref{eq:resolvent_biased}. The bias of such a stochastic construction can be controlled by increasing the number of Gaussians $G_k$.

\begin{figure*}[hbt!]
\centering
\begin{tikzpicture}
\begin{axis}[
  width=0.9\columnwidth,
  height=0.385\columnwidth,
  xlabel={$\omega$},
  xmin=-1,
  xmax=1,
  ylabel={Value of Approximant},
  ytick={0, 0.5, 1},
  title={Zolotarev Approximant $(K, \overline{\omega}) = (4, 10^{-1})$},
  legend style={draw=none, row sep=-2pt},
  legend pos=north east,
]{
\addplot[black,dashed,very thick]%
  table[x index=0,y index=1] {\datfile{Zolo_deg8_width0.1}};
\addplot[myColOne,thick,solid,mark=o,mark size=1.5]%
  table[x index=0,y index=1] {\datfile{Zolo_deg8_width0.1_r}};
\addplot[myColTwo,thick,solid,mark=diamond,mark size=1.5]%
  table[x index=0,y index=1] {\datfile{poly_deg8}};
\addplot[myColThr,thick,solid,mark=asterisk,mark size=1.5]%
  table[x index=0,y index=1] {\datfile{poly_deg16}};
\legend{Exact Zolotarev,Stochastic Zolotarev,Polynomial (Degree 8),Polynomial (Degree 16)};
}{-1}{1}{0}{1}
\end{axis}

\begin{axis}[
  at={(0.225\columnwidth,0.125\columnwidth)},
  anchor=center,
  width=0.325\columnwidth,
  height=0.225\columnwidth,
  ymode=log,
  xlabel={$\omega$},
  ylabel={Error},
  xmin=0.2, xmax=1,
  ymin=10^-6, ymax=10^-1,
  minor tick num=1,
]{
\addplot[black,dashed,very thick]%
table[x index=0,y index=1] {\datfile{Zolo_deg8_width0.1_inset}};
\addplot[myColOne,thick,solid,mark=o,mark size=1]%
table[x index=0,y index=1] {\datfile{Zolo_deg8_width0.1_r_inset}};
\addplot[myColTwo,thick,solid,mark=diamond,mark size=1]%
table[x index=0,y index=1] {\datfile{poly_deg8_inset}};
\addplot[myColThr,thick,solid,mark=asterisk,mark size=1]%
table[x index=0,y index=1] {\datfile{poly_deg16_inset}};
}
\end{axis}
\end{tikzpicture}
\caption{Rational and polynomial approximants of the step function $\Theta(\omega)$ over spectral region $[-1, 1]$. The optimal Zolotarev approximant with $2K = 8$ poles is indicated as the black dashed line while its stochastic construction via a continuous-time LCHS is shown in blue. For comparison, the Chebyshev approximations of degrees $8$ and $16$, respectively, are colored in red and yellow. The inset displays the approximation error compared to $\textrm{sgn}$ near the edge of the spectrum.}
\label{fig:numerical_Zolo}
\end{figure*}
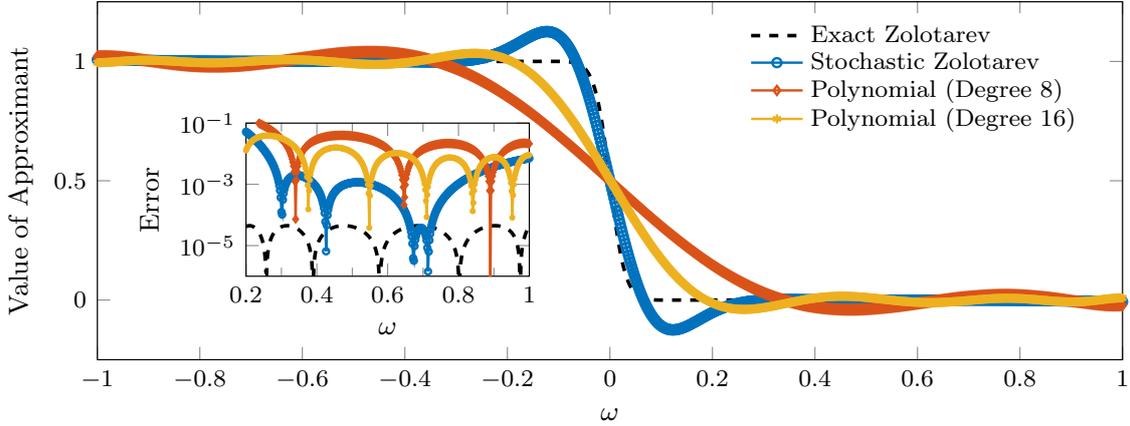

For $K=4$ and $G_k \equiv 8$, we display in \cref{fig:numerical_Zolo} approximations of the step function $\Theta(H)$ by the exact Zolotarev QRT $r_{K}(H)$ and its stochastic implementation $\hat{r}_{K}(H)$ using $N_{\rm MC} = 8$ Monte-Carlo samples. The same scaled Hamiltonian is considered as in \cref{sec:resolvent_discrete,sec:resolvent_continuous}. Since $\lVert H \rVert_2 \leq 1$, we simply examine values of the approximants over the spectral region $[-1, 1]$. We note that the sampling procedure is purely classical, allowing for repeated trials to refine an optimal stochastic approximant. Empirically, the stochastic approximant $\hat{r}_K$ uniformly filters out eigenvalues of $H$ within $[2\overline{\omega}, 1]$. Moreover we show polynomial approximants of varying degrees in the Chebyshev basis. The rational approximations, in contrast to the polynomial counterparts, exhibit $(i)$ a sharper transition around $\omega = 0$ and $(ii)$ a smaller error across $\mathcal{W}_{\overline{\omega}} =[-1, -\overline{\omega}] \cup [\overline{\omega}, 1]$. These distinctive features contribute to effective spectral filtering.

\subsection{Efficient rational filtering by iterative function composition}
\label{subsec:iterative_filtering}
The Zolotarev rational function can approximate the step function to a desired accuracy by increasing the number of poles. However, as these poles approach the real line, the cost of computing the corresponding resolvents becomes more expensive as indicated by \cref{cor:1}.
Alternatively, a fixed rational function (fixed number of poles) achieving a lower approximation accuracy can be applied iteratively and also obtain accurate eigenenergy approximations. Since the cost of resolvent simulation is fully determined by the distance $b_k$ of the pole $z_k=a_k+b_k i$ to the real line and by the distance of $a_k$ to the spectrum of the (filtered) Hamiltonian of interest, each iteration can be performed at a constant cost.

We assume knowledge of the spectral bounds, $E_{-} \leq E_n \leq E_{+}$, although the precise largest and smallest eigenenergies do not need to be known in advance. To achieve a reduction of spectral range $[E_{-}, E_{+}] \mapsto [E_{-}, E_{-} + \delta E]$ for arbitrary $\delta E$, we propose the repeated application of a fixed rational matrix filter $r_E(H)$ alternated by a shift and rescale of the resulting matrix.
Suppose that applying $r_E(H)$ suppresses the eigenvalues in some interval $[E+\Delta E,E_+]$, thereby reducing the effective spectral range by a factor $\xi = \frac{E + \Delta E - E_{-}}{E_{+} - E_{-}}$.
We define the following linear transformation, $f_{\xi}(\omega) = \frac{1}{\xi}(\omega- E_{-}) + E_{-}$, which rescales the spectrum of $r_E(H)$ so that it is again approximately $[E_-,E_+]$, allowing for the same filter $r_E$ to be applied again.
For a given $\delta E$, this process can be repeated $D = \mathcal{O}(\log_{\xi} \delta E)$ times,
\begin{equation}
    (r_E \circ f_{\xi^D}) \cdots (r_E \circ f_{\xi^3}) (r_E \circ f_{\xi^2})(r_E \circ f_{\xi})r_{E}(H)  = \prod_{d=0}^{D} r_{E} \circ f_{\xi^{d}} (H)=:r_{\star D} (H;\xi), 
   \label{eq:iterativeFilter}
\end{equation}
which reduces the effective spectral range to $[E_-,E_-+\delta E]$. For the discrete-time LCHS approach, the efficiency of the iterative filtering process is captured by the corollary below.

\begin{cor}\label{cor:iterative}
    A discrete-time $\epsilon$-approximation to the iterative filter constructed in \cref{eq:iterativeFilter}, $\lVert r_{\star D}(H) - \sum_{j=0}^{J-1} x_j e^{-iHt_j} \rVert_2 < \epsilon$ for $r_{E}(H) = \frac{1 - r_{K}(H)}{2}$, can be obtained with
\begin{align}
    J = \log_{2} \log \frac{ 2 (\gamma_K)^{D} \mathcal{C}_K }{ \epsilon \xi^{D} b^{-} } +   \frac{3b^{-} + 2\sqrt{2} b^{-}  + \lVert H \rVert_2}{4\sqrt{2} b^{-}} \log_2 \frac{ 2 (\gamma_K)^{D} \mathcal{C}_K }{ \epsilon \xi^{D} b^{-} } +3,
\end{align}
distinct time evolution circuits and a maximum runtime of
\begin{align}
    T^{\rm max} = \frac{1}{ \xi^{D} b^{-}} \log \frac{ (\gamma_K)^{D} \mathcal{C}_K }{ \epsilon \xi^{D} b^{-}},
\end{align}
for $\gamma_K$ defined in \cref{eq:rf_Zolo} and a constant $\mathcal{C}_K$ dependent on the rational function $r_K$.
\end{cor}
\begin{proof}
    The proof is provided in \cref{subsec:proof_cor2}.
\end{proof}

We note that the iterative rational filter $r_{\star D}$ can also be constructed using the continuous-time LCHS approach discussed in \cref{sec:resolvent_continuous}. For a wide variety of many-body Hamiltonians $H$, the low-energy sector of the spectrum $\delta E$ is $\mathcal{O} (\frac{1}{{\rm polylog}N} )$ relative to the spectral range. This results in $D = \log_{1/\xi} {\rm polylog} N$ iterations, where the cost of the filtering procedure is determined by the pole locations $\frac{b_k}{{\rm polylog}N}$ (the imaginary parts are scaled down by a factor of $\xi$ with each iteration).

\section{Application of QRT to the ground and excited state problem}
\label{sec:numerics_groundExcitedState}
Once a quantum rational transformation is available, it can be employed in the context of ground and excited state problems.
For example, the QRT from \cref{sec:filter}, which approximates the step function, can be used as a pre-processing routine for certain quantum eigensolvers to speed up their convergence.
We illustrate this for the recently developed eigensolver, the observable dynamic mode decomposition (ODMD)~\cite{shen2023estimating}.
We then compare the performance of ODMD with and without QRT pre-processing to solve the ground and excited state problem for a TFIM Hamiltonian.

For a given Hamiltonian $H = \sum_{n=0}^{N-1} E_n \ket{E_n}\bra{E_n}$, an initial state $\ket{\psi}$, and a measurement state $\bra{\phi}$, ODMD measures the dynamical expectations,
\begin{align}
    o(\tau_j) =  \braket{\phi| e^{-iH \tau_j } |\psi} = \sum_{n=0}^{N-1} \tilde{p}_n e^{-iE_n \tau_j}, \qquad \tau_j = j \Delta \tau,
    \label{eq:odmd}
\end{align}
on the quantum computer. These measurements allow us to identify a least-squares (LS) solution $B$ satisfying,
\begin{align}
  \begin{bmatrix}  
  o(\tau_1) & o(\tau_2) & o(\tau_3) & \cdots \\
  o(\tau_2) &  o(\tau_3) &  o(\tau_4) & \cdots \\
  o(\tau_3) &  o(\tau_4) &  o(\tau_5) & \cdots \\
  \vdots & \vdots & \vdots & \ddots
  \end{bmatrix} \stackrel{{\rm LS}}{=} B \begin{bmatrix}  
  o(\tau_0) & o(\tau_1) & o(\tau_2) & \cdots \\
  o(\tau_1) &  o(\tau_2) &  o(\tau_3) & \cdots \\
  o(\tau_2) &  o(\tau_3) &  o(\tau_4) & \cdots \\
  \vdots & \vdots & \vdots & \ddots
  \end{bmatrix}.
  \label{eq:odmd_LS}
\end{align}
The matrix $B$ is called the system matrix and contains the key frequency/energy information for advancing each observable vector a timestep $\Delta \tau$ forward. We sample over a relatively small number of timesteps so that $B$, upon a regularization, remains
well-conditioned and compact in size. This makes classical diagonalization viable: the extremal eigenphase, $\Tilde{E}_0 = -\frac{1}{\Delta \tau} \max_{\lambda_{B}} {\rm arg}(\lambda_{B})$ for eigenvalues $\lambda_{B}$ of $B$, provides a reliable estimate of the exact ground state energy $E_0$.

By applying a suitable QRT $r(H)$ to the initial state $\vert \phi\rangle$, \textit{i.e.}, $r(H)\ket{\phi}$, we can improve the convergence of ODMD.
This is inspired by adopting a signal processing perspective, \cref{eq:odmd} can be understood as sinusoidal data comprised of exponentially many components $e^{-iE_n \tau_j}$, each oscillating in time at its own frequency. The application of a rational filter $r(\omega)$ through a QRT can remove the high-energy components from $\ket{\psi}$. Then, applying ODMD leads to the measurement of the filtered expectations,
\begin{align}
    o_{\star}(\tau_j) = \braket{\phi| e^{-iH \tau_j }  r(H) |\psi} \approx \sum_{n = 0}^{N_{c}} \tilde{p}_n e^{-iE_n \tau_j},
    \label{eq:filtered_obs}
\end{align}
where $N_{c}$ is an eigenindex such that $E_{N_c} \lesssim E_{-} + \delta E \lesssim E_{N_c + 1}$. The filtered expectations $o_{\star}$ can be computed with a sequence of time-evolution circuits as discussed in \cref{sec:filter}. Following the filtering procedure, we find a new system matrix $B_{\star}$ that effectively describes the dynamics within some low-energy subspace. Diagonalization of $B_{\star}$ thereby enables a simultaneous estimation of eigenenergies in this subspace.

We now apply our real-time rational framework to the problem of finding low-energy states of a many-body Hamiltonian, a fundamental task encountered within physical, chemical, and materials sciences. 
We illustrate the QRT pre-processing using the MFIM Hamiltonian in \cref{eq:MFIM} with parameters $(h, g) = (0, 1)$, so that we recover a transverse-field Ising model (TFIM) at the critical point, where the gap between the ground and first excited state closes in the thermodynamic limit.
Our goal is to compute the $N_c=3$ lowest eigenenergies of a system with $L_{\rm sys} = 12$ spins. 

\begin{figure*}[hbt!]
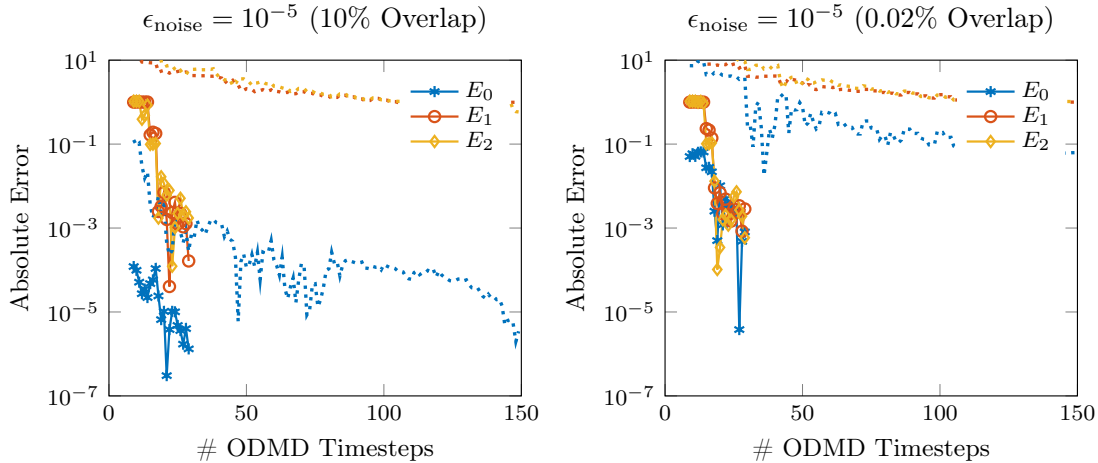

\centering
\plotconv[%
  width=0.48\textwidth,%
  xlabel={\# ODMD Timesteps},%
  xmax=150,%
  ytick={1e-7,1e-5,1e-3,1e-1,1e+1},%
  title={$\epsilon_{\rm noise} = 10^{-5}$ (10\% Overlap)},%
]{
\addplot[myColOne,thick,solid,mark=asterisk,mark size=2]%
  table[x index=0,y index=1] {\datfile{fodmd_TFIM_E0}};
\addplot[myColTwo,thick,solid,mark=o,mark size=2]%
  table[x index=0,y index=1] {\datfile{fodmd_TFIM_E1}};
\addplot[myColThr,thick,solid,mark=diamond,mark size=2]%
  table[x index=0,y index=1] {\datfile{fodmd_TFIM_E2}};
\addplot[myColOne,very thick,dotted]%
  table[x index=0,y index=1] {\datfile{odmd_TFIM_E0}};
\addplot[myColTwo,very thick,dotted]%
  table[x index=0,y index=1] {\datfile{odmd_TFIM_E1}};
\addplot[myColThr,very thick,dotted]%
  table[x index=0,y index=1] {\datfile{odmd_TFIM_E2}};
\legend{$E_0$,$E_1$,$E_2$};
}{0}{150}{1e-7}{1e+1}%
\hfill%
\plotconv[%
  width=0.48\textwidth,%
  xlabel={\# ODMD Timesteps},%
  xmax=150,%
  ytick={1e-7, 1e-5,1e-3,1e-1,1e+1},%
  title={$\epsilon_{\rm noise} = 10^{-5}$ (0.02\% Overlap)},%
]{
\addplot[myColOne,thick,solid,mark=asterisk,mark size=2]%
  table[x index=0,y index=1] {\datfile{fodmdu_TFIM_E0}};
\addplot[myColTwo,thick,solid,mark=o,mark size=2]%
  table[x index=0,y index=1] {\datfile{fodmdu_TFIM_E1}};
\addplot[myColThr,thick,solid,mark=diamond,mark size=2]%
  table[x index=0,y index=1] {\datfile{fodmdu_TFIM_E2}};
\addplot[myColOne,very thick,dotted]%
  table[x index=0,y index=1] {\datfile{odmdu_TFIM_E0}};
\addplot[myColTwo,very thick,dotted]%
  table[x index=0,y index=1] {\datfile{odmdu_TFIM_E1}};
\addplot[myColThr,very thick,dotted]%
  table[x index=0,y index=1] {\datfile{odmdu_TFIM_E2}};
\legend{$E_0$,$E_1$,$E_2$};
}{0}{150}{1e-7}{1e+1}%
\caption{Absolute errors in the TFIM eigenenergies calculated as a function of the number of ODMD timesteps $\tau_j = j \Delta \tau$. Convergence of the $N_{c} = 3$ lowest eigenenergies are shown both with (solid lines) and without (dashed lines) the application of the rational filter $r_{\star 3}(H) = (r_E \circ f_{\xi^3}) (r_E \circ f_{\xi^2})(r_E \circ f_{\xi})r_{E}$. The rational filter is constructed iteratively from a base filter specified by $(K, E_{\pm}, \Delta E) = (4, \pm 1, 10^{-1})$ where $K$, $E_{\pm}$, and $\Delta E$ determine the rational degree, spectral bounds, and width of the filter buffer region. \textbf{Left:} ODMD algorithm run on noisy data with $(\epsilon_{\rm noise}, p_0) = (10^{-5}, 10^{-1})$ where $\epsilon_{\rm noise}$ and $p_0$ specify the noise level and ground state overlap. \textbf{Right:} ODMD algorithm run on noisy data with $(\epsilon_{\rm noise}, p_0) = (10^{-5}, 2\times 10^{-4})$. The singular value threshold parameter $\tilde{\delta}$ is adjusted for the filtered and unfiltered cases to ensure optimal performance.}
\label{fig:odmd_filter}
\end{figure*}

For our eigenstate problem, as an illustration, we set $\ket{\psi} = \ket{\phi}$ with fixed ground state overlap $\tilde{p}_0 = p_0 = \lvert \braket{E_0|\phi} \rvert^2$ and uniform excited state overlap $\tilde{p}_{n} = p_{n} = \frac{1 - p_0}{2^{L_{\rm sys}} - 1}$ for $n \geq 1$. To account for the effects of noise, we first
compute the exact overlap matrix elements in \cref{eq:odmd,eq:filtered_obs}. We then introduce Gaussian perturbations $\mathcal{N}(0, \epsilon_{\rm noise}^2)$ to the matrix elements. According to the central limit theorem, the Gaussian noise model can faithfully capture effects of the shot noise, \textit{i.e.}, the statistical errors due to taking a finite number of measurements. Some types of hardware
noise may also be Gaussian; we do not investigate hardware noise models that specifically include decoherence
or systematic errors. In order to classically mitigate the impact of noise, we employ singular value thresholding~\cite{shen2023estimating} to regularize the LS solution in \cref{eq:odmd_LS}. In particular, we truncate all singular values of the system matrix that are smaller than a set threshold $\Tilde{\delta} \sigma_{\rm max}$ relative to the largest singular value $\sigma_{\rm max}$.

\cref{fig:odmd_filter} shows the obtained approximations to the lowest three eigenenergies with different ground state overlap $p_0$. 
The dashed curves represent the ODMD results with $\ket{\psi} = \vert \phi \rangle$. For $p_0=10\%$, the approximation to the ground state energy $E_0$ converges rapidly to a reasonable accuracy of $\epsilon = 10^{-3}$. Nevertheless, the approximations to the excited state energies $E_1$ and $E_2$ stagnate, highlighting the challenge in capturing these excited states. 
This issue is more pronounced for a poorer initial state with $p_0=0.02\%$, where none of the three energy approximations converge effectively. By applying a suitable rational filter $r(H)$ on the initial state $\vert \phi\rangle$, \textit{i.e.}, choosing $\ket{\psi} = r(H)\ket{\phi}$, and running the ODMD with this filtered state, the convergence improves significantly. This is illustrated by the solid curves in \cref{fig:odmd_filter} with $r=r_{\star 3}$, \textit{i.e.}, $D=3$ in \cref{eq:r_star}. Such an improvement in performance is evident even when 
$\ket{\phi}$ has uniform eigenstate overlap ($p_{0} = p_{n} \equiv 2^{-L_{\rm sys}}$ in the right panel). Thus, the application of a rational filter accelerates the simultaneous estimation of eigenenergies, especially for excited state energies where the convergence of regular ODMD becomes stagnant.

\section{Conclusions}
\label{sec:conclusions}
In this work, we investigate the capability of real-time evolution for implementing rational transformations on quantum hardware. Leveraging LCHS as our algorithmic primitive, we focus on the effective construction of operator resolvents via Hamiltonian simulations. Specifically, we present 
two distinct and complementary LCHS strategies that optimize the overall simulation cost. First, we consider a discrete-time scheme based on exponentially convergent quadrature rules. As a counterpart tailored for analog implementation, we also introduce a stochastic continuous-time scheme that employs resourceful bosonic ancillae. Together, these two schemes provide a comprehensive toolkit for composing a rational function as a linear combination of resolvents. 

Drawing on the unique ability of rational functions to capture singularities, we investigate the implementation of optimal rational approximations to the signum function on a quantum computer. We demonstrate the application of such rational transformations to the ground and excited state problem, where we construct a spectral filter to extract the low-lying eigenenergies of a many-body system. 
As rational functions outperform polynomials by a landslide for a variety of central tasks across scientific computing, our framework offers original insights into exploiting such advantage on quantum platforms. Based on these insights, we propose novel ways of constructing quantum rational transformations that not only advance existing approaches, but also serve as a foundational building block for the development of efficient rational algorithms.

\section{Acknowledgements}
This work was funded by the U.S. Department of Energy under Contract No.~DE-AC02-05CH11231, through the Office of Science, Office of Advanced Scientific Computing Research (ASCR) Exploratory Research for Extreme-Scale Science.
This research used resources of the National Energy Research Scientific Computing Center (NERSC), a U.S. Department of Energy Office of Science User Facility located at Lawrence Berkeley National Laboratory, operated under Contract No. DE-AC02-05CH11231. 
This work started during NVB's research visits to Lawrence Berkeley National Laboratory, partially supported by Charles University Research program No. PRIMUS/21/SCI/009 and Universidad Carlos III de Madrid travel grant Proyectos Jóvenes PPIT2024 No. 2024/00735/001. 
The authors thank Siddharth Hariprakash for helpful discussions.

\bibliographystyle{unsrtnat}
\bibliography{main}

\clearpage
\appendix

\section{Proofs}

\subsection{\textbf{Theorem 1}}
\label{subsec:proof_theorem1}
\noindent \textit{Proof.} By the triangle inequality, we have
\begin{align}
    \lVert R(z_k) - \mathcal{I}_{J}(z_k) \rVert_2 \leq  \lVert R(z_k) - \mathcal{I}_{\infty}(z_k) \rVert_2 + \lVert \mathcal{I}_{\infty}(z_k) - \mathcal{I}_{J}(z_k) \rVert_2,
\end{align}
where the two terms on the RHS represent the truncation and discretization errors, denoted as $\epsilon_1$ and $\epsilon - \epsilon_1$. By \cref{eq:Tmax_discrete}, we already know that $\lVert \bt \rVert_\infty \leq \frac{1}{b_k} \log \frac{1}{\epsilon_1 b_k}$. Moreover, the error bound of \cref{eq:GL_error} clearly applies to the entire function $f_n(t) = e^{i(z_k - E_n)T_{k}^{\rm max} (1+t)/2}$ where $E_n \in [-1, 1]$ is the $n$th eigenvalue of $H$. It is rather straightforward to show through direct 
calculations that $\forall \sigma > 1$,
\begin{align}
    \max_{0 \leq n \leq N-1} \sup_{t \in \mathcal{E}_{\sigma}} \lvert f_n(t) \rvert \leq e^{[ a_k^{+} (\sigma - \sigma^{-1}) +  b_k (\sigma -2 + \sigma^{-1})] T_{k}^{\rm max} / 4 }.
\end{align}
This supremum estimate implies an overall error estimate (cf. \cref{eq:Tmax_discrete,eq:GL_error}),
\begin{align}
     \lVert R(z_k) - \mathcal{I}_{J}(z_k) \rVert_2 \leq \epsilon_1 + \frac{4 \sigma^{2-2J}}{(\sigma^2 - 1) b_k} \left( \frac{1}{\epsilon_1 b_k} \right)^{\eta_k} {\log \frac{1}{\epsilon_1 b_k}},
     \label{eq:GL_overall_qbound}
\end{align}
where 
\begin{align}
    \eta_k(\sigma) = \frac{(\sigma - 1)^2 b_k + (\sigma^2 - 1) a_k^{+} }{4 \sigma b_k}
\end{align}
appears as an exponent. Therefore \cref{eq:GL_overall_qbound} implies that
\begin{align}
    J(\epsilon,\epsilon_1) = \frac{1}{2} \log_{\sigma} \log \frac{1}{\epsilon_1 b_k}  + \frac{\eta_k}{2} \log_\sigma \frac{1}{\epsilon_1 b_k} + \frac{1}{2}\log_{\sigma} \frac{4}{(\sigma^2-1)(\epsilon - \epsilon_1)b_k} + 1
\end{align}
Legendre nodes suffice for an $\epsilon$-accurate approximation.
By setting, for example, $\sigma = \sqrt{2}$ and $\epsilon_1 = \frac{\epsilon}{2}$, we arrive at the upper bound as claimed. We conclude our proof by noting that the Gauss-Legendre time grid $\bt$ is symmetric around $0$: the center of symmetry shifts to $\frac{T_{k}^{\rm max}}{2}$ after we revert the change of variable. $\qed$

\subsection{\textbf{Theorem 2}}
\label{subsec:proof_theorem2}
\noindent \textit{Proof.} Again by the triangle inequality, we have
\begin{align}
    \lVert R(z_k) - \mathcal{I}_{J}(z_k) \rVert_2 \leq  \lVert R(z_k) - \mathcal{I}_{\infty}(z_k) \rVert_2 + \lVert \mathcal{I}_{\infty}(z_k) - \mathcal{I}_{J}(z_k) \rVert_2,
\end{align}
where the truncation and approximation errors contain contributions from both the Gauss-Legendre and trapezoidal rules. For $T_k^{\rm max} = \lVert \bt \rVert_{\infty} = q^{\rm max} y^{\rm max}$, observe that
\begin{align}
    \lVert R(z_k) - \mathcal{I}_{\infty}(z_k) \rVert_2 \leq \frac{1}{a_k^{-}} {\rm erfc} \bigg( \frac{a_k^{-} q^{\rm max}}{2} \bigg) + \frac{\sqrt{2} q^{\rm max}} {\sqrt{\pi} } {\rm erfc} \left(\frac{y^{\rm max}}{\sqrt{2}} \right),
\end{align}
where the RHS involving the complementary error function ${\rm erfc}(\cdot)$ accounts for truncation error along the $q$ and $y$ coordinate respectively. Given $0 < \epsilon_1 < \epsilon$ sufficiently small, setting $q^{\rm max} = \frac{2}{a_k^{-}} \sqrt{\log \frac{2}{\epsilon_1 a_k^{-} }}$ and $y^{\rm max} = \sqrt{ 2 \log \frac{2}{\epsilon_1 a_k^{-} } }$ yields
\begin{align}
    \lVert R(z_k) - \mathcal{I}_{\infty}(z_k) \rVert_2 \leq \bigg[ \frac{1}{a_k^{-}} + \frac{2 \sqrt{2}}{\pi a_k^{-}} \bigg] \frac{\epsilon_1 a_k^{-}}{2} < \epsilon_1,
\end{align}
where we invoke simple bounds, such as a Chernoff-type upper bound ${\rm erfc}(x) \leq \frac{1}{\sqrt{\pi}x}e^{-x^2}$, to obtain an estimate of the truncation error.

For the approximation error, we can apply the bound of \cref{eq:GLT_error} to the entire function $F_n(y|q) = e^{-y^2/2} e^{iy(a_k - E_n)q}$. Direct calculations show that $\forall \sigma_y > 0$,
\begin{align}
    \max_{0 \leq n \leq N-1} e_{\Dy}[F_{n}(\cdot|q)] \lesssim  2\sqrt{2\pi} e^{\sigma_y^2/2 + \sigma_y a_k^{+} \lvert q \rvert - 2 \pi\sigma_y / \Dy },
\end{align}
for a trapezoidal step $\Dy = \frac{2y^{\rm max}}{L_y - 1}$ and given $q$ coordinate. The supremum estimate implies an approximation error,
\begin{align}
    \begin{split}
        \lVert \mathcal{I}_{\infty}(z_k) - \mathcal{I}_{J}(z_k) \rVert_2 \leq \frac{1}{\pi} \int_{0}^{q^{\rm max}} & dq \, \max_{0 \leq n \leq N-1} e_{\Dy}[F_{n}(\cdot|q); \mathcal{S}_{\sigma_y} ] \\
        &+ \frac{q^{\rm max} \Dy}{2 \pi} \sum_{\ell =0}^{L_{y}-1} e^{-y_\ell^2/2} \max_{0 \leq n \leq N-1} e_{L_q} [f_{n\ell}; \mathcal{E}_{\sigma_q}],
    \end{split}
\end{align}
where $f_{n\ell}(q) = e^{i y_\ell (a_k - E_n) q^{\rm max} (1+q)/2}$ denotes the integrand in \cref{eq:resolvent_II_marginal}, now transformed to the interval $[-1, 1]$, and $L_q$ the number of discretization nodes for the Gauss-Legendre rule. The terms $e_{\Dy}[F_{n}(\cdot|q)]$ and $e_{L_{q}} [f_{n\ell}]$ capture discretization errors associated with the trapezoidal and Legendre rule respectively, derived in relation to region of analyticity $\mathcal{S}_{\sigma_y}$ and $\mathcal{E}_{\sigma_q}$. Setting $\sigma_y = y^{\rm max}$ and $\sigma_q = \sqrt{2}$ for simplicity, we hence arrive at an estimate of the overall error,
\begin{align}
    \begin{split}
        \lVert R(z_k) - \mathcal{I}_{J}(z_k) \rVert_2 \leq \epsilon_1 + \frac{4\sqrt{2}}{\sqrt{\pi} \epsilon_1 a_k^{+} a_k^{-} y^{\rm max} } & e^{- \pi(L_y-1)} \bigg[ \frac{2}{\epsilon_1 a_k^{-}} \bigg]^{4\sqrt{2}\zeta_k}  \\
        &+ \frac{\sqrt{2} 2^{5 -L_q}}{\pi a^{-}_k }  \bigg[ \frac{2}{\epsilon_1 a_k^{-}} \bigg]^{ \zeta_k} \log \frac{2}{\epsilon_1 a_k^{-}},\end{split}\label{eq:cost_theorem2}
\end{align}
where the exponent reads $\zeta_k = \frac{ a_k^{+} }{2 a_k^{-}}$. We can separately control the second and third term within arbitrary small tolerance $\epsilon_2$ and $\epsilon_3$, \textit{i.e.}, taking
\begin{align}
    L_y(\epsilon_1, \epsilon_2) = \frac{1}{\pi} \bigg[ \log \frac{4 e^{\pi}}{ \sqrt{\pi} \epsilon_1 \epsilon_2  a_k^{+} a_k^{-} } - \frac{1}{2} \log \log \frac{2}{\epsilon_1 a_k^{-}} + 4 \sqrt{2} \zeta_k \log \frac{2}{\epsilon_1 a_k^{-}} \bigg],
\end{align}
and
\begin{align}
    L_q(\epsilon_1, \epsilon_3) = \log_2 \frac{ 32 \sqrt{2}}{\pi\epsilon_3 a^{-}_k } + \zeta_k \log_2 \frac{2}{\epsilon_1 a_k^{-}} + \log_2 \log \frac{2}{\epsilon_1 a_k^{-}}, \label{eq:Lq}
\end{align}
suffices to obtain an $\epsilon$-accurate approximation.
The total number of discretization nodes is $J(\epsilon_1, \epsilon_2, \epsilon_3) = L_y L_q$. Further setting $\epsilon_1 = 2\epsilon_2 = 2\epsilon_3 = \frac{\epsilon}{2}$, we recover the upper bound as claimed. $\qed$

\subsection{Formulation of continuous-time approach for complex poles}
\label{app:cont_complex}
In this section, the full derivation of our continuous-time approach for complex poles, outlined in \cref{subsec:Dirac_delta_continuous}, is provided. We also emphasize the choices made in the basis functions to ensure that our approach is efficient to implement on hardware.

To elucidate the advantage of continuous-variable computing, we first consider the na\"{\i}ve case, \cref{lem:continuous_naive} states the result but hides the complexity of preparing the ancilla state.

\begin{claim}[Naive case]\label{lem:continuous_naive}
    Consider a continuous-variable ancilla $\ket{0}_c = \int_0^\infty dq\,\psi(q;b_k) \ket{q}$ as in \cref{eq:resolvent_ancilla}. 
    Let $\psi(q;b_k) = \sqrt{b_k}e^{-b_k q/2}$ be the ancillary wavefunction and $\tilde{H}=(H-a_k)\otimes q$ the total Hamiltonian. Time evolution under $\hat{H}$ applied to the composite state $\ket{\phi}\ket{0}_c$ of unit duration prepares the (normalized) action of the resolvent on $\ket{\phi}$:
     \begin{equation}
         iR(z_k)\ket{\phi}/\Vert R(z_k)\ket{\phi}\Vert_2.
     \end{equation}
     The success probability of the state preparation is, for $z_k = a_k+ib_k$,
     \begin{equation}
         P_{R(z_k)} = b_k^2 \lVert R(z_k) \ket{\phi}  \rVert_2^2 
    =  \sum\limits_{n=0}^{N-1} \frac{p_n b_k^2}{ \abs{ a_k - E_n }^2 + b_k^2 },
    \label{eq:continuous_success_0}
     \end{equation}
     where $p_n = \lvert \braket{E_n|\phi} \rvert^2$ denotes the squared overlap between the $n$th eigenstate $\ket{E_n}$ of $H$ and the initial system state $\ket{\phi}$.
\end{claim}
\begin{proof}
    Evolving $\ket{\phi}\ket{0}_c$ for unit duration and substituting the ancilla expression yields:
    \begin{align}
    e^{-i\Tilde{H}} \ket{\phi}\ket{0}_{c} &= \int_{0}^{\infty} dq \, \psi(q) e^{i(a_k - H)q} \ket{\phi} \ket{q} \label{eq:resolvent_continuousI}, \\
    &= \int_{0}^{\infty} dq \, \psi(q) e^{i(a_k - H)q} \ket{\phi} \psi^{\ast}(q) \ket{0}_{c} + \ket{\perp} = i b_k  R(z_k) \ket{\phi} \ket{0}_c + \ket{\perp}.
\end{align}
 In the last equality, we recover the integral transform of \cref{eq:resolvent_I} using ancilla wavefunction given by \cref{eq:resolvent_ancilla}.
 The residual state $\ket{\perp}$ belongs to the kernel of the projector ${\rm Id} \otimes \ket{0}_c \bra{0}_c$ and is orthogonal to $R(z_k) \ket{\phi}\ket{0}_c$.
 Thus by post-selection on the ancilla, we can implement the action of a resolvent, \textit{i.e.}, we can prepare the normalized state $i R(z_k) \ket{\phi} / \lVert R(z_k) \ket{\phi} \rVert_2$. In particular, the success probability of post-selection is given by
\begin{align}
    P_{R(z_k)} &= b_k^2 \lVert R(z_k) \ket{\phi}  \rVert_2^2 
    =  \sum\limits_{n=0}^{N-1} \frac{p_n b_k^2}{ \abs{ a_k - E_n }^2 + b_k^2 },
    \label{eq:continuous_success_0}
\end{align}
where $p_n = \lvert \braket{E_n|\phi} \rvert^2$ is the squared overlap between the $n$th eigenstate $\ket{E_n}$ of $H$ and the initial system state $\ket{\phi}$.
\end{proof}

\cref{lem:continuous_naive} shows that evolving the composite state under $\Tilde{H}$ applies the resolvent to any system state $\ket{\phi}$. 
By \cref{eq:resolvent_continuousI,eq:resolvent_continuousII}, the system-ancilla Hamiltonian simulation has fixed maximal and total runtime $\lVert \bt \rVert_{\infty} = \lVert \bt \rVert_{1} = 1$. While time evolution itself can be executed much more efficiently compared to the discrete approach, the complexity of this continuous approach primarily falls onto the preparation and measurement of the ancilla $\ket{0}_c$. Note that, similar to the discrete setting, we may truncate the spatial integral at finite value of $q$.

Despite the elementary analytical form of $\psi$, its accurate preparation on quantum hardware can be a nontrivial task, for example due to the discontinuity in the wavefunction at $q=0$. Instead, we would like to exploit smooth wavefunctions with rapid spatial decay. To significantly simplify the ancilla initialization, we now express \cref{eq:resolvent_ancilla,eq:resolvent_continuousI} in an alternative representation in terms of Gaussian wavefunctions.
\begin{proof}[Proof of \cref{lem:continuous_Gaussian}]
    We start from \cref{eq:Dirac_delta_stochastic_MC}, and assign variables $A_j$ to the integrals in the sum,
\begin{align}
    iR(z_k) &= \frac{1}{b_k} \mathbb{E}_{\rm mix} \bigg[ \int_{-\infty}^{\infty} dq \, \varphi_g(q;\Gamma) e^{i(a_k - H) \lvert q \rvert} \bigg]  \approx\frac{1}{b_k G} \sum_{j=1}^{G} \underbrace{\int_{-\infty}^{\infty} dq \, \varphi_{g}(q;\Gamma_{j}) e^{i(a_k - H) \lvert q \rvert}}_{:= A_j}.
\end{align}
Each Monte-Carlo sample $\Gamma_j$ is generated according to the probability measure $\rho_{\rm mix}$ and $A_j$ in \cref{eq:Dirac_delta_stochastic_MC} can be obtained by evolving a composite state.
The composite state consists of the system qubits and an ancilla initialized in the Gaussian state with spatial variance $\Gamma_j$,
\begin{align}
   \ket{0_{j}}_{c} = \int_{-\infty}^{\infty} dq \, \sqrt{\varphi_{g}(q;\Gamma_{j})} \ket{q}.
\end{align}
Each sampled value $\Gamma_j$ then specifies the spatial variance for the Gaussian state $\ket{0_j}_c$ of the form \cref{eq:ket_c}. 
Evolving the composite state under the total Hamiltonian $\Tilde{H} = (H - a_k) \otimes \lvert \hat{q} \rvert$ for a unit duration, where $\hat{q}\ket{q} = q\ket{q}$, we obtain 
\begin{align}
    e^{-i\Tilde{H}} \ket{\phi} \ket{0_j}_{c} = A_{j} \ket{\phi} \ket{0_j}_{c} + \ket{\perp_{j}}.
\end{align}
Since the residual $\ket{\perp_{j}}$ is in the kernel of the projector $\textrm{Id}\otimes \ket{0_j}_c \bra{0_j}_c$, a post-selection on the ancilla, i.e., measuring the ancilla in the state $\ket{0_j}_c$, projects this state onto zero. 
Thus, the action $A_j\ket{\phi}$ is available through measurement, and can be summed together to provide an approximation to the action of the resolvent $R(z_k)\ket{\phi}$.

Since
\begin{align}
    \rho_{\rm mix}\left(\{ \Gamma_j\}_{j=1}^{G}|b_k \right) = \prod_{j=1}^{G} \rho_{\rm mix}(\Gamma_j|b_k) = \prod_{j=1}^{G} \frac{b_k^2 }{2} e^{- b_k^2 \Gamma_{j}/2},
\end{align}
can unbiasedly recover the resolvent in the large sample limit $G \rightarrow \infty$.
When employing a Monte-Carlo procedure \cite{Go16} to sample $\rho_{\rm mix}$, increasing the number of samples $G$ improves the accuracy of the resolvent approximation, where the error scales as  $ \sim \frac{1}{\sqrt{G}}$.
The post-selection on ancilla prepares the normalized state $i R(z_k) \ket{\phi} / \lVert R(z_k) \ket{\phi} \rVert_2$ with a success probability,
\begin{align}
    P_{A_{j} } = \left\lVert A_{j} \ket{\phi} \right\rVert_2^2 = \sum_{n=0}^{N-1} p_n e^{- \Gamma_{j} E_n^2} \bigg[ 1 + {\rm erfi}^2  \bigg( \sqrt{\frac{\Gamma_{j} }{2}} E_n \bigg) \bigg],
\end{align}
where ${\rm erfi}(\cdot)$ denotes the imaginary error function, $p_n = \lvert \braket{E_n|\phi} \rvert^2$ is the squared overlap between the $n$th eigenstate $\ket{E_n}$ of $H$ and the system state $\ket{\phi}$.
\end{proof}

\subsection{\textbf{Corollary 1}}\label{subsec:proof_cor1}
\begin{proof}
    Note that the finest time grid, \textit{i.e.}, the one sufficient to approximate the resolvent with the pole $i b^-$ up to an accuracy $\epsilon$, suffices to approximate the resolvents for all $i b_k$ up to an accuracy $\epsilon$.
    Then the statement follows from a direct application of \cref{thm:1} to $A$ in \cref{eq:Zolo_discrete_LCU} together with the fact that $\sum_{k=0}^{K-1} c_k = 1$. 
    The expression of the total evolution time follows immediately from \cref{eq:Zolo_J} and that of the maximal evolution time.
\end{proof}

\subsection{\textbf{Corollary 2}}\label{subsec:proof_cor2}
Before giving the proof to \cref{cor:iterative}, we need the following property.
\begin{prop}\label{prop}
    The application of $D$ iterations, as in \cref{eq:iterativeFilter}, requires computation of $(D+1)K$ resolvents.
\end{prop}
\noindent \textit{Proof of \cref{prop}.} According to \cref{eq:iterativeFilter}, the application of $D$ iterations immediately implies the need to compute $K^{D+1}$ resolvents.
For a rational filter of the form $r_E (H) = \sum_{k=0}^{K-1} c_{k} R(z_{k}) + {\rm h.c.}$ (where ${\rm h.c.}$ denotes the Hermitian conjugate), we have,
\begin{align}
    (r_E \circ f_{\xi})(H) = \sum_{k=0}^{K-1} \xi c_{k} R( f_{1/\xi}(z_{k}) ) + {\rm h.c.},
\end{align}
where the resolvent identity \cite{Yo95} implies,
\begin{align}
   [(r_{E} \circ f_{\xi}) r_E] (H) &= \sum_{k,k'=0}^{K-1} \sum_{\iota, \iota' = \pm} \xi c_{k \iota} c_{k' \iota'}  R( f_{1/\xi}(z_{k \iota}) ) R(z_{k' \iota'}), 
   \label{eq:resolvent_product}\\
   &= \sum_{k,k'=0}^{K-1} \sum_{\iota , \iota' = \pm} \frac{\xi c_{k \iota} c_{k' \iota'}}{\tilde{z}_{k' \iota'} - \xi \tilde{z}_{k \iota} } \left[  R( f_{1/\xi}(z_{k \iota}) ) -  R(z_{k' \iota'}) \right] ,
\end{align}
with $c_{k\pm} = \Re c_k \pm i \Im c_k$, $z_{k\pm} = a_k \pm i b_k$, and $\tilde{z}_{k \pm} = z_{k \pm} - E_{-}$. For simplicity, we introduce the subscript $\iota$ to designate the complex conjugation of a scalar or the Hermitian conjugation of an operator. The composed filter $r_{\star 1}(H;\xi) = [(r_{E} \circ f_{\xi}) r_E ] (H)$ acts on an initial state $\ket{\phi}$ as
\begin{align}
    \begin{split}
        \braket{\phi| f(H) r_{\star 1}(H;\xi) |\phi} = \xi \sum_{k=0}^{K-1} \sum_{\iota = \pm} & c_{k\iota}^{(1)}(\xi) \braket{\phi|f_{\iota}(H) R( f_{1/\xi}(z_{k}) )|\phi}_{\iota} \\
        &+ \sum_{k=0}^{K-1} \sum_{\iota = \pm} c^{(1)}_{k\iota}(1/\xi) \braket{\phi|f_{\iota}(H) R(z_{k})|\phi}_{\iota},
    \end{split}
\end{align}
where
\begin{align}
    c_{k\iota}^{(1)}(\xi) = \sum_{k'=0}^{K-1} \sum_{\iota'=\pm} \frac{ c_{k \iota} c_{k' \iota'} }{\tilde{z}_{k' \iota'} - \xi \tilde{z}_{k \iota} }.
\end{align}
That is, we have converted the evaluation of $K^2$ resolvent products (c.f. \cref{eq:resolvent_product}) to that of $2K$ resolvents within a single iterative filtering step. 
From a recursive argument it follows immediately that \cref{eq:iterativeFilter} can, therefore, be constructed by only $D(K+1)$ resolvents. $\qed$
\bigskip

\noindent \textit{Proof of \cref{cor:iterative}.}
An efficient rational approximation to the step function can be obtained by the Zolotarev approximant $r_K(H)$ in \cref{eq:rf_Zolo}, \textit{i.e.}, $r_{E}(H) = \frac{1 - r_{K}(H)}{2}$. 
Below we describe in detail a discrete-time procedure for constructing the corresponding iterative filter $r_{\star D}(H)$. First observe that 
\begin{align}
    r_{\star D}(H;\xi) = \frac{1}{2^{D+1}} \sum_{d=0}^{D+1} (-1)^d \sum_{\bs \subseteq [D+1]: \lvert \bs \rvert = d} \prod_{i=0}^{d-1} r_{K} \circ f_{\xi^{s_i}} (H), \label{eq:r_star}
\end{align}
where each $\bs = [s_0, s_1, \cdots, s_{d -1}]$ is a multi-index specified by distinct $s_i \in [D+1]$. We simply assume that $\bs$ is ordered with $s_i < s_{i+1}$. Thus to construct a discrete-time filter $\epsilon$-close to $r_{\star D}(H)$, it suffices to adopt a Legendre rule $\bt$ such that
\begin{align}
    \bigg\lVert \prod_{i=0}^{d-1} r_{K} \circ f_{\xi^{s_i}}(H) - \sum_{j=0}^{J-1} x_j e^{-iHt_j} \bigg\rVert_2 \leq \epsilon,
\end{align}
for all multi-indices $\bs$. We assert that $\bs = [D+1]$ incurs the largest error among multi-indices since by \cref{cor:1},
\begin{align}
    J_{\bs} = \log_{2} \log \frac{ 2 (\gamma_K)^{d} \lVert \bc_{\bs} \rVert_1 }{ \epsilon \xi^{\lVert \bs \rVert_{\infty} } b^{-} } +   \frac{3b^{-} + 2\sqrt{2} b^{-}  + \lVert H \rVert_2}{4\sqrt{2} b^{-}} \log_2 \frac{ 2 (\gamma_K)^{d} \lVert \bc_{\bs} \rVert_1}{ \epsilon \xi^{\lVert \bs \rVert_{\infty}} b^{-} } +3,
\end{align}
and 
\begin{align}
    \lVert \bt_{\bs} \rVert_{\infty} = \frac{1}{ \xi^{\lVert \bs \rVert_{\infty}} b^{-}} \log \frac{ (\gamma_K)^{d} \lVert \bc_{\bs} \rVert_1 }{ \epsilon \xi^{\lVert \bs \rVert_{\infty}} b^{-}},
\end{align}
where $\bc_{\bs}$ is the vector of rational coefficients satisfying
\begin{align}
    \prod_{i=0}^{d-1} r_{K} \circ f_{\xi^{s_i}}(H) = (-\gamma_K)^{d} \sum_{i=0}^{d-1} \sum_{k=0}^{K-1} \sum_{\iota = \pm}  (c_{\bs})_{i k \iota}  R ( f_{1/\xi^{s_i}}(z_{k\iota}) ),
\end{align}
with $\lVert \bc_{\bs} \rVert_1$ bounded by $\xi^{\lVert \bs \rVert_{\infty}}$.  We notice that the spectral transformation $H \mapsto f_{\xi}(H)$ does not stretch the effective Hamiltonian spectral range, since we apply the transformed filters in a sequential way. For example, an application of the initial filter $r_K(H)$ almost `eliminates' the components of $\ket{\phi}$ in the energy range $[E+\Delta E, E_{+}]$, which also holds during succeeding iterative steps that operate on a smaller and smaller effective spectral range. $\qed$

\section{Extension from simple to repeated poles}
\label{sec:repeatedPoles}
The real-time constructions in Section \cref{sec:resolvent_discrete,sec:resolvent_continuous} can be generalized to account for a pole of higher multiplicity, thus extending our analysis beyond a simple pole associated with the resolvent. A natural generalization is to consider the $m^{\rm th}$-power of a resolvent $R^{m}(z_k) = (z_k -H)^{-m}$ for any positive integer $m$. Notably, a higher-order pole can be expressed through time evolutions if we simply differentiate \cref{eq:resolvent_unitary},
\begin{align}
    (z_k -H)^{-m} &= \frac{(-1)^{m-1}}{(m-1)!} \frac{d^{m-1}}{d \omega^{m-1}} (z_k +\omega - H )^{-1} \bigg \rvert_{\omega = 0}, \label{eq:higher_order_resolvent} \\
    &= \frac{(-i)^{m-1}}{(m-1)!} \int_{0}^{\infty} dq \, q^{m-1} \int_{-\infty}^{\infty} dy \, y^{m-1} \kernel_{k}(q,y) e^{iy(z_k - H)q},\label{eq:higher_order_resolvent_2}
\end{align}
where \cref{eq:higher_order_resolvent} holds by functional calculus \cite{Yo95}. Observe that the polynomial growth of $q^{m-1}$ or $y^{m-1}$ in the integrand is dominated by the exponential decay of $\kernel_k e^{-iy(z_k-H)q}$. 

For discrete-time LCHS, the observation above implies that the maximal evolution time, at the leading order, remains the same as in the case of a simple pole. For continuous-time LCHS, we can block-encode $R^{m}(z_k)$ stochastically with Gaussian states when $b_k > 0$, \textit{i.e.},
\begin{align}
    \begin{split}
        R^{m}(z_k) = \frac{(-i)^{m}}{(m-1)!} \int_0^{\infty} & dq \, q^{m-1} e^{-b_k q + i(a_k- H)q} \\
        &\approx \frac{(-i)^{m}}{(m-1)!} \sum_{j=1}^{G} \int_{-\infty}^{\infty} dq \, \varphi_{g}(q; q_j, \Gamma_j) e^{i (a_k - H) q},
    \end{split}
\end{align}
where $q^{m-1}e^{-b_k q}$ follows the density of a gamma distribution with shape parameter $m$ and rate parameter $b_k$, and we invoke Gaussian mixture approximation detailed in \cref{sec:resolvent_continuous} now with a set of center-displaced Gaussians $\varphi_{g}(q; q_j, \Gamma_j) := \varphi_{g}(q - q_j, \Gamma_j)$. In the large $m$ limit, the approximation holds for a single Gaussian mode by the central limit theorem, since a gamma random variable is statistically identical to the sum of $m$ i.i.d. exponential variables. When $b_k = 0$, we can block-encode $R^{m}(z_k)$ deterministically with two bosonic excited states,
\begin{align}
    R^{m}(z_k) \approx \frac{C}{2} \int_{-\infty}^{\infty} dq \, \psi_{m}^2(q;\Gamma^{{\rm max},m}) \int_{-\infty}^{\infty} dy \, \psi_{m}^2(y;1) e^{iy(a_k - H)q},
\end{align}
where $\psi_{m}(q;\Gamma^{{\rm max},m})$, for instance, is the $m^{\rm th}$-lowest eigenstate of the harmonic oscillator with a characteristic width $\sqrt{\Gamma^{{\rm max},m}}$, and $C$ is a normalization constant. Therefore the overall cost of implementing higher-order resolvent is marginally affected by pole multiplicity.

\section{Numerical benchmark for the Heisenberg models}
\label{app:Heisenberg}
Using the 2D and 3D Heisenberg spin Hamiltonians on different lattices, we run more experiments illustrating the main findings in this paper. All Hamiltonians considered are taken from the HamLib collection \cite{Hamlib}, where more details on the Hamiltonians can be found.

For the Hamiltonian "2D triag pbc qubitnodes Lx=4, Ly=6, h=2" with a pole located at $z = -0.49+0.1i$, \cref{fig:2D_triag} shows the convergence of the quadrature rules described in \cref{subsec:Dirac_delta_discrete}. The results confirm the conclusion in \cref{subsec:Dirac_delta_discrete_numerics}: the Legendre shows fast exponential convergence, outperforming the trapezoidal and Laguerre rule.
\begin{figure}[!ht]
    \centering
    \setlength\figureheight{2.5cm}
    \setlength\figurewidth{0.85\textwidth}	
%
\begin{tikzpicture}

\begin{axis}[%
width=0.476\figurewidth,
height=\figureheight,
at={(-0.56\figurewidth,0\figureheight)},
scale only axis,
xmin=0,
xmax=120,
ymode=log,
ymin=1e-07,
ymax=100,
yminorticks=true,
axis background/.style={fill=white},
legend style={draw=none, font=\footnotesize}
]
\addplot [color=myblue, draw=none, mark=o, mark options={solid, myblue}]
  table[row sep=crcr]{%
1	10.865\\
3.5	9.4695\\
6	6.3671\\
8.5	4.0939\\
11	3.4057\\
13.5	2.2696\\
16	1.0669\\
18.5	0.051527\\
21	0.00099925\\
23.5	0.00099925\\
26	0.00099925\\
28.5	0.00099925\\
31	0.00099925\\
33.5	0.00099925\\
36	0.00099925\\
38.5	0.00099925\\
41	0.00099925\\
43.5	0.00099925\\
46	0.00099925\\
48.5	0.00099925\\
51	0.00099925\\
53.5	0.00099925\\
};
\addlegendentry{Leg}

\addplot [color=myred, draw=none, mark=asterisk, mark options={solid, myred}]
  table[row sep=crcr]{%
0.043429	5.3865\\
0.53201	2.9086\\
1.5635	1.9324\\
3.1378	1.2127\\
5.255	0.7468\\
7.915	0.45102\\
11.118	0.26879\\
14.864	0.15634\\
19.152	0.089824\\
23.984	0.051184\\
29.358	0.028973\\
35.276	0.016317\\
41.736	0.0091525\\
48.739	0.0051171\\
56.285	0.002853\\
64.373	0.0015868\\
73.005	0.00088077\\
82.179	0.00048797\\
91.897	0.00026991\\
102.16	0.00014907\\
112.96	8.2227e-05\\
};
\addlegendentry{Lag}

\addplot [color=myorange, draw=none, mark=+, mark options={solid, myorange}]
  table[row sep=crcr]{%
1	46.014\\
3.5	11.862\\
6	10.551\\
8.5	10.26\\
11	10.152\\
13.5	3.6938\\
16	1.7674\\
18.5	1.1142\\
21	0.78681\\
23.5	0.59217\\
26	0.46467\\
28.5	0.37571\\
31	0.31077\\
33.5	0.26172\\
36	0.22367\\
38.5	0.1935\\
41	0.16914\\
43.5	0.14917\\
46	0.13259\\
48.5	0.11865\\
51	0.10682\\
};
\addlegendentry{Trap}
\draw [dashed,thick,black] (0,1e-3) -- (100,1e-3);
\end{axis}

\begin{axis}[%
width=0.476\figurewidth,
height=\figureheight,
at={(0\figurewidth,0\figureheight)},
scale only axis,
xmin=-3.8031,
xmax=100,
ymode=log,
ymin=1e-07,
ymax=48.04,
yminorticks=true,
axis background/.style={fill=white},
legend style={draw=none, font=\footnotesize}
]
\addplot [color=myblue, draw=none, mark=o, mark options={solid, myblue}]
  table[row sep=crcr]{%
1	9.4932\\
3.5	10.799\\
6	7.9378\\
8.5	6.3467\\
11	4.7035\\
13.5	3.877\\
16	3.0206\\
18.5	2.2393\\
21	1.4688\\
23.5	0.82803\\
26	0.37612\\
28.5	0.1148\\
31	0.009348\\
33.5	0.00012213\\
36	9.9925e-07\\
38.5	9.9925e-07\\
41	9.9925e-07\\
43.5	9.9925e-07\\
46	9.9925e-07\\
48.5	9.9925e-07\\
51	9.9925e-07\\
53.5	9.9925e-07\\
56	9.9925e-07\\
58.5	9.9925e-07\\
61	9.9925e-07\\
63.5	9.9925e-07\\
66	9.9925e-07\\
68.5	9.9925e-07\\
71	9.9925e-07\\
};
\addlegendentry{Leg}
\addplot [color=myred, draw=none, mark=asterisk, mark options={solid, myred}]
  table[row sep=crcr]{%
0.024817	5.3865\\
0.30401	2.9086\\
0.89341	1.9324\\
1.793	1.2127\\
3.0028	0.7468\\
4.5229	0.45102\\
6.3531	0.26879\\
8.4936	0.15634\\
10.944	0.089824\\
13.705	0.051184\\
16.776	0.028973\\
20.157	0.016317\\
23.849	0.0091525\\
27.851	0.0051171\\
32.163	0.002853\\
36.785	0.0015868\\
41.717	0.00088077\\
46.96	0.00048797\\
52.512	0.00026991\\
58.375	0.00014907\\
64.549	8.2227e-05\\
71.032	4.5302e-05\\
77.826	2.4931e-05\\
84.929	1.3707e-05\\
92.343	7.5291e-06\\
100.07	4.1322e-06\\
108.1	2.2662e-06\\
116.45	1.2419e-06\\
};
\addlegendentry{Lag}
\addplot [color=myorange, draw=none, mark=+, mark options={solid, myorange}]
  table[row sep=crcr]{%
1	80.548\\
3.5	15.324\\
6	11.675\\
8.5	10.769\\
11	10.462\\
13.5	10.307\\
16	10.158\\
18.5	9.806\\
21	7.2375\\
23.5	3.4888\\
26	2.2114\\
28.5	1.5898\\
31	1.2223\\
33.5	0.97976\\
36	0.80827\\
38.5	0.6811\\
41	0.58345\\
43.5	0.50644\\
46	0.44439\\
48.5	0.39353\\
51	0.35123\\
53.5	0.31561\\
56	0.2853\\
58.5	0.25927\\
61	0.23673\\
63.5	0.21706\\
66	0.1998\\
68.5	0.18455\\
};
\addlegendentry{Trap}
\draw [dashed,thick,black] (0,1e-6) -- (100,1e-6);
\end{axis}

\end{tikzpicture}%
    \caption{Error $\Vert R(z)-\mathcal{I_J}(z)\Vert_2$ shown as a function of the scaled total evolution time $\frac{T^{\rm tot}}{T^{\rm max}}$, for $z=-0.49+0.1i$ and the (scaled) Hamlib Hamiltonian "2D triag pbc qubitnodes Lx=4, Ly=6, h=2" of size $2^{12}\times 2^{12}$. The approximation $\mathcal{I_J}(z)$ is computed using the trapezoidal ({\color{green}+}), Legendre ({\color{red}$\circ$}) and Laguerre (\color{black}$\ast$) rule for a requested accuracy $\epsilon=10^{-3}$ (left) 
    and $\epsilon=10^{-6}$ (right).}
    \label{fig:2D_triag}
\end{figure}
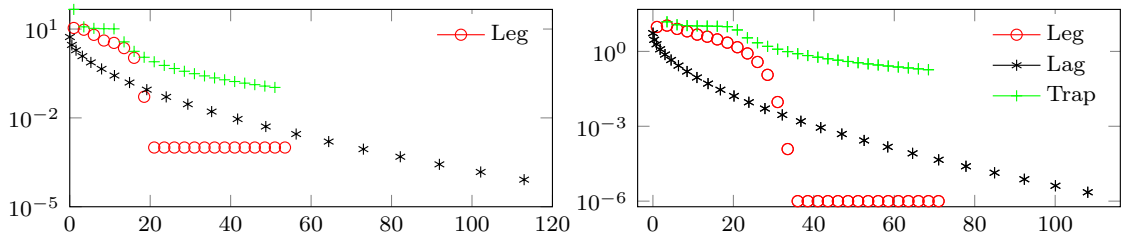
We perform the same experiment for different Hamiltonians and pole locations, two of the considered Hamlib Hamiltonians and poles are:
\begin{itemize}
    \item "2D-triag-pbc-qubitnodes Lx-3 Ly-5 h-2" ($z$=-0.49+0.1i and $z$=0+0.1i)
    \item "3D-grid-pbc-qubitnodes Lx-2 Ly-2 Lz-2 h-2" ($z$=-0.71+0.1i)
\end{itemize}
The Legendre rule requires a lower cost for the approximation of the resolvent for the above Hamiltonians, and we expect that this is true for any Hamiltonian in this collection and for any choice of pole.

\section{Zolotarev rational approximant}
\label{app:Zolotarev}
The Zolotarev rational approximant in \cref{eq:rf_Zolo} achieves a favorable error that is asymptotically sharp with respect to the supremum norm,
\begin{align}
    {\rm err}_K = \sup_{\omega \in \mathcal{W}_{ \overline{\omega} }} \lvert r_{K}(\omega) - {\rm sgn}(\omega)  \rvert  = \frac{2\sqrt{Z_K}}{1+Z_K},
\end{align}
where $Z_K(\overline{\omega})$, known as the Zolotarev number, depends on the approximation region set by $\overline{\omega}$ and can be bounded by $0 \leq Z_K \leq 4 e^{- K\pi^2/\log(4/\overline{\omega})}$~\cite{Beckermann2017}. This implies an exponentially decaying error ${\rm err}_K \leq 4 e^{- K \pi^2/2 \log(4/\overline{\omega})}$, also reported in \cref{eq:Zolo_errBound}. We remark that a tighter bound can be derived in terms of the Gr\"{o}tzsch ring function (see \cite{Beckermann2017}).

The Zolotarev number enjoys the important invariance property that
\begin{align}
    Z_K(\mathcal{W}_{ \overline{\omega} })  \equiv Z_K(\mathcal{T}([-1, -\overline{\omega}]) \cup \mathcal{T}([ \overline{\omega}, 1]) ),
\end{align}
for any M{\"o}bius transformation $\mathcal{T}$. Such invariance allows us to restrict to windows of the previous form. This is because for any pair of disjoint, ordered intervals $[\omega_a, \omega_b] \cup [\omega_c, \omega_d]$, we can identify a M{\"o}bius transformation $\mathcal{T}_{abcd}$ such that
\begin{align}
    \mathcal{T}_{abcd}([\omega_a, \omega_b]) = [ -1, -\overline{\omega}], \qquad \mathcal{T}_{abcd}([\omega_c, \omega_d]) = [\overline{\omega}, 1],
\end{align}
where the transformed window bound $\overline{\omega}$ satisfies
\begin{align}
    \frac{(1 + \overline{\omega})^2}{4 \overline{\omega}} = \frac{(\omega_c - \omega_a)(\omega_d - \omega_b)}{(\omega_c - \omega_b)(\omega_d - \omega_a)} = \frac{(\omega_{ab} + \omega_{bc}) (\omega_{bc} + \omega_{cd})}{\omega_{bc}(\omega_{ab} + \omega_{bc} + \omega_{cd}) },
\end{align}
with, for example, $\omega_{ab} := \omega_{b} - \omega_{a}$ (so a larger value of $\omega_{bc}$, i.e., larger buffer width $\Delta E$ in \cref{eq:rE}, results in a faster convergence).

\end{document}